\documentclass[12pt,a4paper,reqno]{amsart}
\usepackage{amsfonts}
\usepackage{amsthm}
\usepackage{amsmath}
\usepackage{times}
\usepackage{amscd}
\usepackage[latin2]{inputenc}
\usepackage{t1enc}
\usepackage{enumitem}
\usepackage[mathscr]{eucal}
\usepackage{indentfirst}
\usepackage{graphicx}
\usepackage{graphics}
\usepackage{pict2e}
\usepackage{epic}
\usepackage{esint}
\usepackage[margin=2.9cm]{geometry}
\usepackage{epstopdf} 
\usepackage{verbatim}
\usepackage{cancel}
\usepackage{amssymb}
\usepackage{xcolor}
\usepackage{dsfont}
\usepackage{hyperref}
\usepackage{physics}

\setitemize{itemsep=5pt, topsep=7pt}
\numberwithin{equation}{subsection}

\theoremstyle{plain}
\newtheorem{theo}{Theorem}[section]
\newtheorem{lem}[theo]{Lemma}
\newtheorem{cor}[theo]{Corollary}
\newtheorem{prop}[theo]{Proposition}

\theoremstyle{definition}
\newtheorem{defin}[theo]{Definition}

\newtheorem{rem}[theo]{Remark}

\newtheorem{theorem}{Theorem}

\newcommand{\unit}{1\!\!1}
\newcommand{\R}{\mathbb{R}}
\newcommand{\Rtre}{\mathbb{R}^{3}}

\newcommand{\Gkerspace}{(-\Delta_{\Rtre})^{-1}(x,y)}

\DeclareMathOperator{\spn}{span}
\DeclareMathOperator{\infspec}{inf\,spec}
\DeclareMathOperator{\ran}{ran}

\DeclareMathOperator{\dist}{dist}
\DeclareMathOperator{\const}{const.}

\DeclareMathOperator{\Err}{Err}

\newcommand{\LL}{L^L}
\newcommand{\LLn}{L^{L_n}}
\newcommand{\XL}{X^L}
\newcommand{\XLn}{X^{L_n}}
\newcommand{\TL}{\mathbb{T}^3_L}
\newcommand{\TLn}{\mathbb{T}^3_{L_n}}
\newcommand{\EL}{\mathcal{E}_L}
\newcommand{\ELn}{\mathcal{E}_{L_n}}
\newcommand{\Einf}{\mathcal{E}}
\newcommand{\Emininf}{\varPsi}

\newcommand{\einf}{e_{\infty}}
\newcommand{\eL}{e_L}
\newcommand{\eLn}{e_{L_n}}
\newcommand{\Tker}{(-\Delta_L)^{-1}(x,y)}
\newcommand{\bary}{\bar y}

\newcommand{\ropsi}{\rho_{\psi}}
\newcommand{\MinLe}{\mathcal{M}^{\mathcal{E}}_{L}}
\newcommand{\MinLf}{\mathcal{M}^{\mathcal{F}}_L}

\newcommand{\GL}{\mathcal{G}_L}
\newcommand{\HL}{\mathbb{H}_L}
\newcommand{\FL}{\mathcal{F}_L}

\newcommand{\numero}{\mathbb{N}}
\newcommand{\ad}{a^{\dagger}}
\newcommand{\elecphononcoupl}{v^x_{L}}
\newcommand{\elecphononcouplCUT}{v^x_{L,\Lambda}}
\newcommand{\elecphononcouplREST}{w^x_{L,\Lambda}}
\newcommand{\LagrML}{\mu^L}
\newcommand{\LagrMLn}{\mu^{L_n}}
\newcommand{\LagrMinf}{\mu}

\newcommand{\HF}{H^{\FL}_{\varphi_L}}
\newcommand{\GradProj}{\Pi^L_{\nabla}}
\newcommand{\GradProjT}{\Pi^L_{\nabla,T}}

\newcommand{\GSOrthProj}{Q_{\psi_L}}
\newcommand{\Tneigh}{\left[\Pi\Omega_L(\varphi_L)\right]_{\varepsilon,T}}

\begin{document}

\title[Strongly Coupled Polaron on the Torus]{The Strongly Coupled Polaron on the Torus: Quantum Corrections to the Pekar Asymptotics}

\author{Dario Feliciangeli and Robert Seiringer}
\address{IST Austria, Am Campus 1, 3400 Klosterneuburg, Austria}

\begin{abstract}
We investigate the Fr\"ohlich polaron model on a three-dimensional torus, and give a proof of the 
 second-order quantum corrections  to its ground-state energy in the strong-coupling limit. Compared to previous work in the confined case, the translational symmetry (and its breaking in the Pekar approximation) makes the analysis substantially more challenging. 
\end{abstract}

\date{\today}

\maketitle

\section{Introduction}

The underlying physical system we are interested in studying is that of a charged particle (e.g., an electron) interacting with the quantized optical modes of a polar crystal (called phonons). In this situation, the electron excites the phonons by inducing a polarization field, which, in turn, interacts with the electron. 
In the case of a `large polaron' (i.e., when the De Broglie wave-length of the electron is much larger than the lattice spacing in the medium), this system is described by the Fr\"ohlich Hamiltonian \cite{frohlich1937theory}, which represents a simple and well-studied model of non-relativistic quantum field theory (see \cite{alexandrov2010advances,frank2014ground,gerlach1991analytical,moller2006polaron,seiringer2020polaron,spohn1987effective} for properties, results and further references).

A key parameter that appears in the problem is the coupling constant, usually denoted by $\alpha$. We study the strong coupling regime of the model, i.e., its asymptotic behavior as $\alpha \to \infty$. In this limit, the ground state energy of the Fr\"ohlich Hamiltonian agrees to leading order with the prediction of the Pekar approximation \cite{pekar1954untersuchungen}, which assumes a classical behavior for the phonon field. 
This was first proved  in \cite{donsker1983asymptotics}, using a path integral approach (see also \cite{mukherjee2018identification} and \cite{mukherjee2018strong}, for recent work on the Pekar process \cite{spohn1987effective}). Later, the result was improved in \cite{lieb1997exact}, by providing explicit bounds on the leading order correction term.

The object of our study is, precisely, the main correction to the classical (Pekar) approximation of the polaron model, i.e., the leading error term in the aforementioned asymptotics for the ground state energy. Such correction is expected to be of order $O(\alpha^{-2})$ smaller than the leading term, and arises from the quantum fluctuations about the classical limit \cite{allcock1963strong}. This claim was first verified rigorously in \cite{frank2019quantum}, where both the electron and the phonon field are confined to a bounded domain (of linear size adjusted to the natural length scale set by the Pekar ansatz) with Dirichlet boundary conditions. Such restriction breaks translation invariance and simplifies the structure of the Pekar problem in comparison with the unconfined case, guaranteeing, at least in the case of the domain being a ball \cite{feliciangeli2020uniqueness}, uniqueness up to phase of the Pekar minimizers and non-degeneracy of the Hessian of the Pekar functional. 
We build upon the strategy  developed in  \cite{frank2019quantum} to treat the ultraviolet singularity of the model, which in turn relies on multiple application of the Lieb--Yamazaki commutator method \cite{lieb1958ground} and a subsequent use of Nelson's Gross transformation \cite{gross1962particle,nelson1964interaction}.

The key novelty of the present work is to deal with a translation invariant setting. 
We investigate the quantum correction to the Pekar approximation of the polaron model on a torus, and prove the validity of the predictions in \cite{allcock1963strong} also in this setting. As a first step, we analyze the structure of the set of minimizers of the corresponding Pekar functional, proving uniqueness of minimizers up to symmetries, which 
was so far known to hold only in the unconfined case \cite{lieb1977existence,lenzmann2009uniqueness} and on balls with Dirichlet boundary conditions \cite{feliciangeli2020uniqueness}. The translation invariance leads to a degeneracy of the Hessian of the Pekar functional and corresponding zero modes, substantially complicating the analysis of the quantum fluctuations. In order to `flatten' the surface of minimizers, we introduce a convenient diffeomorphism inspired by formal computations  in \cite{gross1976strong}, which effectively allows us to decouple the zero modes.

\section{Setting and Main Results}

\subsection{The Model}
We consider a $3$-dimensional flat torus of side length $L>0$. We denote by $\Delta_L$ the  Laplacian on $\TL$ and by $\Tker$ the integral kernel of its `inverse', which we define by 
\begin{equation}
\begin{cases}
&(-\Delta_L) \left[(-\Delta_L)^{-1}(\,\cdot\,,y)\right]= \delta_y\\
&\int_{\TL} \Tker dx =0. 
\end{cases}
\end{equation}
An explicit formula for $\Tker$ is given by
\begin{equation}
\label{expltker}
\Tker=\sum_{0\neq k\in \frac {2\pi}{L} \mathbb{Z}^3} \frac 1 {|k|^2} \frac{e^{ik\cdot(x-y)}}{L^3},
\end{equation}
which, for any $x\in \TL$, yields an $L^2$ function of $y$, its Fourier coefficients being in $\ell^2$. Analogously we define $(-\Delta_L)^{-s}$ for any $s>0$. In the following, we identify  $\TL$ with the box $[-L/2,L/2]^3\subset \Rtre$, and the Laplacian with the corresponding one on  $[-L/2,L/2]^3$ with periodic boundary conditions.

Let
\begin{align}
\label{eq:ElPhCouplDef}
v_L(y) :=(-\Delta_L)^{-1/2}(0,y)=\sum_{0\neq k\in \frac {2\pi}{L} \mathbb{Z}^3} \frac 1 {|k|}\frac {e^{-ik\cdot y}}{L^3},
\end{align} 
and $\elecphononcoupl(y):=v_L(y-x)$. The Fr\"ohlich Hamiltonian \cite{frohlich1937theory} for the polaron is given by 
\begin{align}
\label{eq:FrHam}
\HL&:=(-\Delta_L)\otimes \unit+\unit\otimes\numero -a(\elecphononcoupl)-\ad(\elecphononcoupl)\nonumber\\
&\;=(-\Delta_L)\otimes \unit+\unit\otimes \left(\sum_{ k\in \frac {2\pi}{L} \mathbb{Z}^3}\ad_k a_k\right) -\frac 1 {L^{3/2}}\sum_{0\neq k\in \frac {2\pi} L \mathbb{Z}^3}\frac 1 {|k|}\left(a_ke^{ik\cdot x}+\ad_k e^{-ik\cdot x}\right),
\end{align}
acting on $L^2(\TL)\otimes \mathcal{F}(L^2(\TL))$, where $\mathcal{F}(L^2(\TL))$ denotes the bosonic Fock space over $L^2(\TL)$. The number operator, denoted by $\numero$, accounts for the field energy, whereas $-\Delta_L$ accounts for the electron kinetic energy. The creation and annihilation operators for a plane wave of momentum $k$ are denoted by $a_k^\dagger$ and $a_k$, respectively, and they are assumed to  satisfy the rescaled canonical commutation relations
\begin{align}
\label{eq:commrel}
[a_k,\ad_j]=\alpha^{-2} \delta_{k,j}.
\end{align}
In light of \eqref{eq:commrel}, $\numero$ has spectrum $\sigma(\numero)=\alpha^{-2} \{0,1,2,\dots\}$. We note that the definition \eqref{eq:FrHam} is somewhat formal, since $v_L\not \in L^2(\TL)$. It is nevertheless possible to define $\HL$ via the associated quadratic form, and to find a suitable domain on which it is self-adjoint and bounded from below (see \cite{griesemer2016self}, or Remark \ref{rem:Gross} in Section \ref{sec:ProofMainResult} below). 

We shall investigate the ground state energy of $\mathbb{H}_L$, for fixed $L$ and $\alpha \to \infty$. 

\begin{rem}
By rescaling all lengths by $\alpha$, $\HL$ is unitarily equivalent to the operator $\alpha^{-2} \widetilde{\mathbb{H}}_L$, where $\widetilde{\mathbb{H}}_L$ can be written compactly as
\begin{align}
\widetilde{\mathbb{H}}_L=(-\Delta_{\alpha^{-1} L})\otimes \unit-\sqrt{\alpha}\left[\tilde{a}(v_{\alpha^{-1}L}^x)+\tilde{a}^{\dagger}(v_{\alpha^{-1}L}^x)\right]+\unit\otimes\widetilde{\mathbb{N}},
\end{align}
with the creation and annihilation operators $\tilde{a}^{\dagger}$ and $\tilde{a}$ now satisfying  the (un-scaled) canonical commutation relations $[\tilde{a}(f),\tilde{a}^{\dagger}(g)]=\bra{f}\ket{g}$, and $\tilde{\mathbb{N}}$  the corresponding number operator. Large $\alpha$ hence corresponds to the strong-coupling limit of a polaron confined to a torus of side length $L\alpha^{-1}$. We find it more convenient to work in the variables defined in \eqref{eq:FrHam}, however.
\end{rem}

\begin{rem}
The Fr\"ohlich polaron model is typically considered without confinement, i.e., as a model on $L^2(\Rtre)\otimes \mathcal{F}(L^2(\Rtre))$ with electron-phonon coupling function given by $(-\Delta_{\Rtre})^{-1/2}(x,y)= (2\pi^2)^{-1} |x-y|^{-2}$. In the confined case studied in \cite{frank2019quantum}, $\R^3$ was  replaced by a  bounded domain $\Omega$, and thus  the electron-phonon coupling function was given by $(-\Delta_{\Omega})^{-1/2}(x,y)$, where $\Delta_{\Omega}$ denotes the Dirichlet Laplacian on $\Omega$. The latter setting, similarly to ours, has the advantage of guaranteeing compactness for the corresponding inverse Laplacian, which is a key technical ingredient both for  \cite{frank2019quantum} and  our main results. In addition, for generic domains $\Omega$ the Pekar functional has a unique minimizer up to phase (which is proved in \cite{feliciangeli2020uniqueness} for $\Omega$ a ball, and enters the analysis in \cite{frank2019quantum} for general $\Omega$ as an assumption). 
Compared with \cite{frank2019quantum}, setting the problem on the torus (or on $\R^3$) introduces the extra difficulty of having to deal with translation invariance and a whole continuum of Pekar minimizers. Hence the present work can be seen as a first step in the direction of generalizing the results of \cite{frank2019quantum} to the case of $\R^3$.
\end{rem} 

\subsection{Pekar Functional(s)} 
\label{sec:PekarFunctionalsTorus}

For $\psi\in H^1(\TL)$, $\|\psi\|_2=1$, and $\varphi\in L^2_{\R}(\TL)$, we introduce the classical energy functional corresponding to \eqref{eq:FrHam} as
\begin{align}
\label{eq:Gfun}
\GL(\psi,\varphi):=\expval{h_{\varphi}}{\psi}+\|\varphi\|_2^2,
\end{align}
where $h_\varphi$ is the Schr\"odinger operator
\begin{align}
\label{eq:hVphi}
h_{\varphi}:=-\Delta_L+V_{\varphi}, \quad V_{\varphi}:= -2 (-\Delta_L)^{-1/2} \varphi.
\end{align}  
We define the Pekar energy as
\begin{align}
\label{eq:pekaren}
\eL:=\min_{\psi,\varphi} \GL(\psi,\varphi).
\end{align}
In the case of $\Rtre$, it was shown in \cite{donsker1983asymptotics} and \cite{lieb1997exact} that the infimum of the spectrum of the Fr\"ohlich Hamiltonian converges to the minimum of the corresponding classical energy functional as $\alpha\to \infty$. In \cite{frank2019quantum}, it was shown that the same holds for the model confined to a bounded domain with Dirichlet boundary conditions and the subleading correction in this asymptotics was computed. Our goal is to extend the results of \cite{frank2019quantum} to the case of $\TL$.

We define the two functionals
\begin{align}
\label{eq:EFfun}
\EL(\psi):=\min_{\varphi} \GL(\psi,\varphi),\quad
\FL(\varphi):=\min_{\psi} \GL(\psi,\varphi),
\end{align}
and their respective sets of minimizers 
\begin{align}
\label{eq:minEL}
\MinLe:=&\left\{\psi\in H^1(\TL) \,|\, \|\psi\|_2=1, \; \EL(\psi)=\eL\right\},\\
\label{eq:minLf}
&\MinLf:=\{\varphi\in L^2_{\R}(\TL) \,\,|\,\, \FL(\varphi)=\eL\}.
\end{align}
Clearly, $\EL$ is invariant under translations and changes of phase and $\FL$ is invariant under translations. It is thus useful to introduce the notation
\begin{align}
\label{eq:invariantsurfaceE}
\Theta_L(\psi):=\{e^{i\theta}&\psi^y(\,\cdot\,):=e^{i\theta} \psi(\,\cdot\,-y)\,\,|\,\,\theta\in [0,2\pi), \,y\in \TL\},\\
\label{def:invariantsurfaceF}
&\Omega_L(\varphi)=\{\varphi^y \,\,|\,\, y\in \TL\},
\end{align}
for $\psi \in H^1(\TL)$ and $\varphi \in L^2_{\R}(\TL)$, respectively.

Our first result, Theorem \ref{uniquenessANDcoercivity} (or, more precisely, Corollary \ref{cor:uniquenessANDcoercivity}) is a fundamental ingredient to prove our main result, Theorem \ref{theo:GSEexpansion}. It concerns the uniqueness of minimizers of $\EL$ up to symmetries and shows the validity of a quadratic lower bound for $\EL$ in terms of the $H^1$-distance from the surface of minimizers. We shall prove these properties for sufficiently large $L$. 

\begin{theo}[Uniqueness of Minimizers and Coercivity for $\EL$]
	\label{uniquenessANDcoercivity}
	There exist $L_1>0$ and a positive constant $\kappa_1$ independent of $L$, such that for $L>L_1$ there exists $0<\psi_L\in C^{\infty}(\TL)$ such that
	\begin{align}
	\label{bigLregime}
	\eL<0, \quad \MinLe=\Theta_L(\psi_L).
	\end{align}
	Moreover $\psi_L^y\neq \psi_L$ for any $0\neq y\in \TL$ and, for any $L^2$-normalized $f\in H^1(\TL)$,
	\begin{align}
	\label{globalquadbound}
	\EL(f)-\eL\geq \kappa_1\dist^2_{H^1}\left(\MinLe,f\right).
	\end{align} 
\end{theo}

These properties of $\EL$ translate easily to analogous properties for the functional $\FL$, as stated in the following corollary. 

\begin{cor}[Uniqueness of Minimizers and Coercivity for $\FL$]
	\label{cor:uniquenessANDcoercivity}
	For $L>L_1$ (where $L_1$ is the same as in Theorem \ref{uniquenessANDcoercivity}) there exists $\varphi_L\in C^{\infty}(\TL)$ such that 
	\begin{align}
	\MinLf=\Omega_L(\varphi_L).
	\end{align}
	Moreover, with $\psi_L$ as in Theorem \ref{uniquenessANDcoercivity}, we have
	\begin{align}
	\label{eq:psilphil}
	\varphi_L=\sigma_{\psi_L}:=(-\Delta_L)^{-1/2} |\psi_L|^2, \quad \psi_L= \text{unique positive g.s. of } h_{\varphi_L}.
	\end{align}
	Finally, there exists $\kappa'>0$ independent of $L$ such that, for all $\varphi\in L^2(\TL)$,
	\begin{align}
	\label{eq:Fglobalquadbound}
	\FL(\varphi)-\eL&\geq \min_{y\in \TL} \expval{\unit -(\unit+\kappa'(-\Delta_L)^{1/2})^{-1}}{\varphi-\varphi_L^y}+\left|L^{-3/2}\int_{\TL}\varphi\right|^2,
	\end{align}
	and this implies
	\begin{align}
	\label{eq:Fglobalquadbound2}
	\FL(\varphi)-\eL\geq \tau_L \dist_{L^2}^2(\MinLf, \varphi)
	\end{align}
	with $\tau_L:=\frac {\kappa' (2\pi/L)^2}{1+\kappa' (2\pi/L)^2}$.
\end{cor}

In the case of $\Rtre$, similar results are known to hold. In particular, the analogue of \eqref{bigLregime} was shown in \cite{lieb1977existence} and the analogue of \eqref{globalquadbound} follows from the results in \cite{lenzmann2009uniqueness}. In the case of a bounded domain with Dirichlet boundary conditions, an equivalent formulation of Theorem \ref{uniquenessANDcoercivity} was taken as working assumption in \cite{frank2019quantum}. In the case of a ball in $\Rtre$ with Dirichlet boundary conditions, the analogue of Theorem \ref{uniquenessANDcoercivity} was proved in \cite{feliciangeli2020uniqueness}. In both the case of $\Rtre$ and of balls, rotational symmetry plays a key role in the proof of these results. Rotational symmetry is  not present in our setting, hence a different approach is required. Our method of proof of Theorem \ref{uniquenessANDcoercivity} relies on a comparison of the models on $\TL$ and $\Rtre$, for large $L$. As a consequence, our analysis does not easily yield quantitative estimates on $L_1$. 

To state our main result, which also holds in the case $L>L_1$, we need to introduce the Hessian of the functional $\FL$ at its unique (up to translations) minimizer $\varphi_L$,
\begin{align}
\lim_{\varepsilon\to 0} \frac 1 {\varepsilon^2} \left(\FL(\varphi_L+\varepsilon \phi)-\eL\right)=:\expval{\HF}{\phi} \quad \forall \phi\in L^2_{\R}(\TL).
\end{align}
An explicit computation gives (see Proposition \ref{prop:FHessian})
\begin{align}
\label{eq:Hessianexpr} 
\HF=\unit-4(-\Delta_L)^{-1/2} \psi_L &\frac {Q_{\psi_L}} {h_{\varphi_L}-\infspec h_{\varphi_L}}\psi_L (-\Delta_L)^{-1/2},
\end{align}
where $h_{\varphi_L}$ is defined in \eqref{eq:hVphi}, $\psi_L$ is interpreted as a multiplication operator and $Q_{\psi_L}:=\unit-\ket{\psi_L}\bra{\psi_L}$. Clearly, by minimality of $\varphi_L$, $\HF\geq 0$, and it is also easy to see that $\HF\leq1$. We shall show that $\HF$ has a three-dimensional kernel,  given by $\spn\{\partial_j \varphi_L\}_{j=1}^3$, 
%(which we shall often denote as $\spn\{\nabla \varphi_L\}$ in the following with a slight abuse of notation)
corresponding to the invariance under translations of the functional. Note that we could define the Hessian of $\FL$ at any other minimizer $\varphi_L^y$, obtaining a unitarily equivalent operator $H^{\FL}_{\varphi_L^y}$.

\subsection{Main Result} Recall the definition \eqref{eq:pekaren} for the Pekar energy $\eL$ as well as \eqref{eq:Hessianexpr} for the Hessian of $\FL$ at its minimizers, for $L>L_1$. Our main result is as follows.
\begin{theo}
\label{theo:GSEexpansion}
For any $L>L_1$, as $\alpha \to \infty$ 
	\begin{align}
	\label{eq:infspecHrough}
	\infspec \HL=\eL- \frac 1 {2\alpha^2} \Tr\left(\unit-\sqrt{H_{\varphi_L}^{\FL}}\right)+o(\alpha^{-2}).
	\end{align}
	More precisely, the bounds
	\begin{align}
	\label{eq:infspecHsharp}
	-C_L\alpha^{-1/7}\leq\alpha^2\infspec \HL -\alpha^2\eL+\frac 1 2  \Tr \left(\unit-\sqrt{\HF}\right)\leq C_L\alpha^{-2/11}
	\end{align}
	hold for some $C_L>0$ and $\alpha$ sufficiently large. 
\end{theo}

The trace appearing in \eqref{eq:infspecHrough} and \eqref{eq:infspecHsharp} is over $L^2(\TL)$. Note that, since $\HF\leq1$, the coefficient of $\alpha^{-2}$ in \eqref{eq:infspecHrough}  is negative.

In the case of bounded domains with Dirichlet boundary conditions, an analogue of Theorem \ref{theo:GSEexpansion} was proven in \cite{frank2019quantum} (where logarithmic corrections appear in the bounds that correspond to \eqref{eq:infspecHsharp} as a consequence of technical complications due to the boundary). Showing the validity of an analogous result on $\Rtre$ still remains an open problem, however. Indeed, the constant $C_L$ appearing in the lower bound in \eqref{eq:infspecHsharp} diverges as $L\to \infty$. This is mainly due to the lack of compactness of the resolvent of the full-space Laplacian (which leads, for instance, to a zero lower bound  in \eqref{eq:Fglobalquadbound2} and, in particular, a divergence of the effective number of modes in \eqref{eq:Nlambdaasymp}). On the other hand, our method of proof used in Section \ref{Sec:UpperBound} to show the upper bound in \eqref{eq:infspecHsharp} does apply, with little modifications, to the full space case. In any case, both the upper and lower bound are expected to hold in the case of $\Rtre$ as well \cite{allcock1963strong,gross1976strong,frank2019quantum,seiringer2020polaron}.

Compared to the results obtained in \cite{frank2019quantum}, Theorem \ref{theo:GSEexpansion} deals with the additional complication of the invariance under translations of the problem, which implies that the set of minimizers of $\FL$ is a three-dimensional manifold. This substantially complicates the proof of the lower bound in \eqref{eq:infspecHsharp}, as we shall see in Section \ref{Sec:LowerBoundII}. In particular, we need to perform a precise local study around the manifold of minimizers $\Omega_L(\varphi_L)$, which we carry out by introducing a suitable  diffeomorphism (inspired by \cite{gross1976strong}).

\begin{rem}[Small $L$ Regime]
\label{rem:smallL}
As we show in Lemma \ref{existence}, there exists $L_0>0$ such that the analogue of Theorem \ref{uniquenessANDcoercivity} for $L<L_0$  can be proven with a few-line-argument. In this case, $\EL$ is simply non-negative and is therefore minimized by the constant function. In particular, $e_L = 0$ and $\varphi_L=0$. 

 Also an analogue of Theorem \ref{theo:GSEexpansion} can be proven in the regime $L< L_0$, i.e., it is possible to show that for $L<L_0$ there exists $C_L>0$ such that
\begin{align}
-C_L\alpha^{-1/7}\leq\alpha^2\infspec \HL+\frac 1 {2} \sum_{0\neq k\in \frac {2\pi} L \mathbb{Z}^3} \left(1-\sqrt{1-\frac 4 {L^3|k|^4}}\right)\leq C_L\alpha^{-2/11}
\end{align} 
for large $\alpha$. 
In this case (unlike the regime $L>L_1$ where the set of minimizers $\MinLf$ is a three-dimensional manifold)  $\MinLf$ only consists of  the  $0$ function, and this allows to follow essentially the same arguments of \cite{frank2019quantum} (with only small modifications, which are also needed in the regime $L>L_1$ and hence are discussed in this paper). We shall therefore not carry out the details of this analysis here. 

Whether uniqueness of Pekar minimizers up to symmetries holds for {\em all} $L>0$ (i.e., also in the regime $L_0\leq L \leq L_1$) remains an open problem. 
\end{rem}  

Throughout the paper, we use the word \emph{universal} to describe any constant (which is generally denoted by $C$) or property that is independent of all the parameters involved and in particular independent of $L$, for $L\geq L_0$ (for some fixed $L_0>0$). Also, we write $a\lesssim b$ whenever $a\leq Cb$ for some universal and positive $C$. We write $C_L$ whenever a constant depends on $L$ but is otherwise universal with respect to all other parameters. Finally, we write $a\lesssim_L b$ whenever $a\leq C_L b$ for some positive $C_L$. 

\subsection{Proof Strategy and Structure of the Paper} In Section \ref{sec:PekarFun} we study the properties of the Pekar functionals $\EL$ and $\FL$ defined in \eqref{eq:EFfun}.  %  in Section \ref{sec:PekarFunctionalsTorus}.
We start by recalling the relevant properties of the Pekar functionals on $\R^3$ in Section~\ref{sec:FullSpacePekar}. 
 In the long Section \ref{subsec:EL} %(organized as explained below)
 we give the proof of Theorem \ref{uniquenessANDcoercivity}. Our method of proof %of Theorem \ref{uniquenessANDcoercivity} 
 relies on showing the convergence, as $L\to \infty$, of $\EL$ to its full-space counterpart $\mathcal{E}$. % and exploiting the properties $\mathcal{E}$ (whose definition and properties are recalled in Section \ref{sec:FullSpacePekar}).
 Proposition \ref{prop:En&MinConv} in Section \ref{subsubsec:prelres} formalizes the precise meaning of this convergence. 
 % in a mathematical useful way the concept of $\EL$ converging to $\mathcal{E}$ in Section \ref{subsubsec:prelres}, namely with Proposition \ref{prop:En&MinConv}. 
 Then, in Section \ref{subsubsec:HessianEL}, we prove a stronger notion of convergence, % $\EL$ converges to $\mathcal{E}$ in an even stronger sense, 
 namely  that the Hessian of $\EL$ at any minimizer  converges to the Hessian  of $\mathcal{E}$ at a corresponding minimizer  (in the sense of Proposition \ref{coercivity}); in particular, it  is strictly positive above its trivial zero modes for large $L$. By combining  the results obtained in Sections \ref{subsubsec:prelres} and  \ref{subsubsec:HessianEL}, we conclude the proof of Theorem \ref{uniquenessANDcoercivity} in Section \ref{subsubsec:MainResEL}.  Section \ref{subsec:FL} is dedicated to the investigation of the properties of $\FL$. First, in Section \ref{sec:LBforFL}, we show the validity of Corollary \ref{cor:uniquenessANDcoercivity}. 
 Subsequently we compute the Hessian of $\FL$ (in Proposition \ref{prop:FHessian} in Section \ref{sec:HFL}) and characterize its kernel (in Proposition \ref{prop:Hessianstrictpos} in Section \ref{sec:localstudyFL}). Finally, in Section \ref{sec:localstudyFL} we introduce a family of weighted norms (see \eqref{eq:weightednorm}) which is of key importance in Section \ref{sec:ProofMainResult} and we show, in Lemma \ref{lemma:WTneighbourhood}, that the surface of minimizers of $\FL$ locally admits a unique projection w.r.t. any of these norms. % (note that the $L^2$-norm is part of this family). 

In Section \ref{sec:ProofMainResult} we prove Theorem \ref{theo:GSEexpansion}. First of all, in Section \ref{Sec:UpperBound} we construct a trial state and use it to obtain an upper bound to the ground state energy of $\HL$.
% for fixed $L>L_1$ (recall that $L_1$ is defined in Theorem \ref{uniquenessANDcoercivity}). 
This is carried out using the $Q$-space representation of the bosonic Fock space $\mathcal{F}(L^2(\TL))$ (see \cite{reed1975ii}) and follows ideas contained in \cite{frank2019quantum}, with only small modifications. The remaining sections are devoted to the lower bound. In Section \ref{Sec:LowerBoundI}, we show that it is possible to apply an ultraviolet cutoff on momenta of size larger than some  $\Lambda$ to $\HL$ at an expense of order $\Lambda^{-5/2}$ (see Proposition \ref{prop:cutoffH}). This is proven following closely the approach used in \cite{frank2019quantum}: as a first step we apply a triple Lieb--Yamazaki bound \cite{lieb1958ground} (in Section \ref{Sec:TripleLY}) and then make use of a Gross transformation \cite{gross1962particle,nelson1964interaction} (in Section \ref{Sec:Gross}).  In Section \ref{Sec:LowerBoundII} we show the validity of the lower bound in \eqref{eq:infspecHsharp}, thus completing the proof of Theorem \ref{theo:GSEexpansion}. With Proposition \ref{prop:cutoffH} at hand, we have good estimates on the cost of applying an ultraviolet cutoff to $\HL$ and this allows to reduce the problem to a finite dimensional one (with dimension $N$ diverging as $\alpha\to \infty$). We adopt a similar strategy to \cite{frank2019quantum}, using IMS localization to split the space into an inner region close to the surface of minimizers of $\FL$ and an outer region far away from it. The goal is to extract the relevant quantum correction to the ground state energy  from the inner region and to show, using the bound \eqref{eq:Fglobalquadbound2}, that the outer region contributes only as an error term. Compared to  \cite{frank2019quantum}, the translation invariance substantially complicates the analysis. In contrast to the case considered in \cite{frank2019quantum}, the set of minimizers of $\FL$ is a three-dimensional manifold and does not only consist of a single function. Hence, in order to treat the inner region and decouple the zero-modes of the Hessian of $\FL$, we have to introduce a suitable diffeomorphism (see Definition~\ref{def:GroosCoord} in Section~\ref{Sec:IMSinner}) that 'flattens' the manifold of minimizers and the region close to it. It is  here where we make use Lemma \ref{lemma:WTneighbourhood}, which allows us to understand the local structure of the tubular neighborhood of the surface of minimizers of $\FL$. Another technical complication relates to the metric used to distinguish between the inner and outer region, as simply considering the $L^2$-norm is not sufficient for our purposes, and we need the weighted norms defined in \eqref{eq:weightednorm} (in particular we apply the IMS localization with respect to a metric which depends on $\alpha$).

\section{Properties of the Pekar Functionals}
\label{sec:PekarFun}

In this section we derive important properties of the functionals $\EL$ and $\FL$, introduced in Section \ref{sec:PekarFunctionalsTorus} and defined in \eqref{eq:EFfun}. In Section \ref{subsec:EL}, we show the validity of Theorem \ref{uniquenessANDcoercivity}, relying on the comparison of the models on $\TL$ and $\Rtre$ for large $L$. In Section \ref{subsec:FL}, we study the functional $\FL$. In particular, we prove Corollary \ref{cor:uniquenessANDcoercivity} and compute the Hessian of $\FL$ at its minimizers. 

Given a function $f\in L^2(\TL)$ and $k\in \frac{2\pi}{L}\mathbb{Z}^3$, we denote by $f_k$ the $k$-th Fourier coefficient of $f$. We also denote
\begin{align}
\hat{f}:=f-L^{-3}\int_{\TL} f.
\end{align}
We shall use the following definition of \emph{fractional Sobolev semi-norms} for functions $f\in L^2(\TL)$, $0\neq s\in \mathbb{R}$:
\begin{align}
\label{def:fracsobnorm}
\|f\|_{\mathring{H}^s(\TL)}^2=\expval{(-\Delta_L)^{s}}{f}=\sum_{0\neq k\in \frac {2\pi}{L} \mathbb{Z}^3} |k|^{2s} |f_k|^2.
\end{align} 
Before moving on with the discussion, we recall in the following subsection the definition and relevant properties of the full-space Pekar functional.

\subsection{The Full-Space Pekar Functional}
\label{sec:FullSpacePekar}

 Let $\psi\in H^1(\Rtre)$ be an $L^2(\Rtre)$-normalized function and $\varphi\in L^2_{\R}(\Rtre)$. Then 
\begin{align}
\mathcal{G}(\psi,\varphi):=\expval{h^{\Rtre}_\varphi}{\psi} + \|\varphi\|_2^2
\end{align}
where $h_{\varphi}^{\Rtre}$ is the Schr\"odinger operator
\begin{align}
h_{\varphi}^{\Rtre}:=-\Delta_{\R^{3}}+V_{\varphi}, \quad V_{\varphi}:=-2(-\Delta_{\Rtre})^{-1/2}\varphi.
\end{align}
Comparing with \eqref{eq:Gfun} and \eqref{eq:hVphi}, we note the analogy between the definitions and observe that we are slightly abusing notation by denoting both potentials with the same symbol (we do this for simplicity and since ambiguity does not arise). Analogously to \eqref{eq:EFfun}, we define 
\begin{align}
\label{eq:EinfF}
\Einf(\psi):=\inf_{\varphi} \mathcal{G}(\psi,\varphi), \quad \mathcal{F}(\varphi):=\inf_{\psi} \mathcal{G}(\psi,\varphi).
\end{align}
In analogy with  \eqref{eq:pekaren}, we denote 
\begin{align}
\label{eq:pekareninf}
\einf:=\inf_{\psi,\varphi}\mathcal{G}(\psi,\varphi)=\inf_{\psi}\mathcal{E}(\psi)=\inf_{\varphi}\mathcal{F}(\varphi).
\end{align}  
For our purposes, in the case of $\Rtre$, it is sufficient to focus our discussion on the functional $\Einf$, of which we now recall the main properties. As shown in \cite{lieb1977existence}, $\Einf$ admits a \emph{unique} positive and radially decreasing minimizer $\Emininf$ which is also smooth, the set of minimizers of $\Einf$ coincides with 
\begin{align}
\label{eq:SpaceSurfaceofMin}
\Theta(\Emininf):=\{e^{i\theta}\Emininf^y\,\,|\,\, \theta\in[0,2\pi), \,\, y\in \Rtre\},
\end{align} 
and $\Emininf$ satisfies the Euler--Lagrange equation
\begin{align}
\left(-\Delta_{\Rtre} +V_{\sigma_{\Emininf}} -\LagrMinf_{\Emininf}\right)\Emininf=0,
\end{align}
with
\begin{align}
\label{eq:fullspaceVmu}
\sigma_\Emininf:=(-\Delta_{\Rtre})^{-1/2}|\Emininf|^2, \quad V_{\sigma_\Emininf}= -2 (-\Delta_{\Rtre})^{-1} |\Emininf|^2, \quad \LagrMinf_{\Emininf}=T(\Emininf)-2W(\Emininf),
\end{align}
where $T$ and $W$ are defined in \eqref{eq:FullSpaceE} below. 
Furthermore, as was shown in \cite{lenzmann2009uniqueness}, the Hessian of $\Einf$ at its minimizers is strictly positive above the trivial zero modes resulting from the invariance under translations and changes of phase. This implies the validity of the following Theorem, which is not stated explicitly in \cite{lenzmann2009uniqueness} but can be obtained by standard arguments (see, e.g., \cite[Appendix A]{feliciangeli2020persistence} or \cite{frank2013symmetry}) as a consequence of the results therein contained. 

\begin{theorem}
	\label{h1coerc}
	There exists a constant $C>0$, such that, for any $L^2$-normalized $f \in H^1(\Rtre)$
	\begin{align}
	\Einf(f)-\einf\geq C\dist_{H^1}^2\left(\Theta (\Emininf),f\right).
	\end{align}
\end{theorem}

Our strategy to prove Proposition \ref{prop:En&MinConv} relies on Theorem \ref{h1coerc} and in comparing $\TL$ with $\Rtre$ for large $L$.

\subsection{Study of $\EL$and Proof of Theorem \ref{uniquenessANDcoercivity}} \label{subsec:EL} 

To compare $\EL$ and $\Einf$, we prefer to write both of them in the following form, which can be obtained from \eqref{eq:EFfun} and \eqref{eq:EinfF}, respectively, by a simple completion of the square,  
\begin{align}
\label{eq:FullSpaceE}
\Einf(\psi)&=\int_{\Rtre} |\nabla \psi(x)|^2 dx - \int_{\Rtre}\int_{\Rtre} \rho_{\psi}(x)(-\Delta_{\Rtre})^{-1}(x,y)\rho_{\psi}(y)dxdy=: T(\psi)-W(\psi),\\
\label{eq:EFunChar}
\EL(\psi)&=\int_{\TL} |\nabla \psi(x)|^2 dx - \int_{\TL}\int_{\TL} \rho_{\psi}(x)(-\Delta_L)^{-1}(x,y)\rho_{\psi}(y)dxdy=: T_L(\psi)-W_L(\psi).
\end{align}
The next ingredient, needed for the comparison of $\EL$ and $\Einf$, is the following lemma.

\begin{lem}
	\label{kernelbounds}
	There exists a universal constant $C$ such that 
	\begin{align}
	\sup_{x,y\in \TL}\left| \Tker-(4\pi)^{-1} (\dist_{\TL}(x,y))^{-1}\right|\leq \frac C L.
	\end{align}
\end{lem}
\begin{proof}
	We define $F_L(x):=-\Delta^{-1}_L(x,0)$ and $F(x)=(4\pi)^{-1}|x|^{-1}$ and observe that our statement is equivalent to showing that
	\begin{align}
	\label{star}
	\|F_L-F\|_{L^{\infty}([-L/2,L/2]^3)}\leq \frac C L.
	\end{align}
	By definition, we have $F_L(x)=\frac 1 L F_1(\frac x L)$. Hence, \eqref{star} is equivalent to 
	\begin{align}
	\|F_1-F\|_{L^{\infty}([-1/2,1/2]^3)}\leq C.
	\end{align} 
	Again by definition, $F_1-F$ is harmonic (distributionally and hence also classically) on $\left(\Rtre \setminus \{\mathbb{Z}^3\}\right)\cup \{0\}$ (when $F_1$, and only $F_1$, is extended to the whole space by periodicity). Thus we  conclude that $F_1-F$ is in $C^{\infty}\left( (-1,1)^3\right)$ and, in particular, bounded on $[-1/2,1/2]^3$. 
\end{proof}

The analogy between \eqref{eq:EFunChar} and \eqref{eq:FullSpaceE}, combined with Lemma \ref{kernelbounds}, clearly suggests that $\EL$ formally converges to $\Einf$ as $L\to \infty$. Hence, we set out to show that this convergence can be made rigorous and allows to infer properties of $\EL$ by comparing it to $\Einf$, in the large $L$ regime. 

In Section \ref{subsubsec:prelres} we derive an important preliminary result, namely Proposition \ref{prop:En&MinConv}. It formalizes in a mathematical useful way the concept of $\EL$ converging to $\Einf$. In Section \ref{subsubsec:HessianEL}, we study the Hessian of $\EL$, showing that it converges (in the sense of Proposition \ref{coercivity}) to the Hessian of $\Einf$ and therefore is strictly positive above its trivial zero modes for large $L$. Finally, in Section \ref{subsubsec:MainResEL} we use the results obtained in Sections \ref{subsubsec:prelres} and  \ref{subsubsec:HessianEL} to show the validity of Theorem \ref{uniquenessANDcoercivity}. 

We remark that our approach differs from the one used on $\Rtre$ and on balls to show, for the related $\mathcal{E}$-functional, uniqueness of minimizers and strict positivity of the Hessian (see \cite{lieb1977existence} and \cite{lenzmann2009uniqueness} for the case of $\Rtre$ and \cite{feliciangeli2020uniqueness} for the case of balls). In those cases, rotational  symmetry allows to first show uniqueness of minimizers and then helps to derive the positivity of the Hessian at the minimizers. We take somewhat the opposite road: comparing $\EL$ to $\Einf$, we first show that minimizers (even if not unique) all localize around the full-space minimizers (see Proposition \ref{prop:En&MinConv}) and that the Hessian at each minimizer is universally strictly positive (see Proposition \ref{coercivity}) for large $L$. We then use these two properties to derive, as a final step, uniqueness of minimizers.

\subsubsection{Preliminary Results} \label{subsubsec:prelres} The next Lemma proves the existence of minimizers for any $L>0$. Moreover, it shows that there exists $L_0>0$ such that, for $L<L_0$, $\EL$ is strictly positive on any non-constant $L^2$-normalized function, as already mentioned in Remark \ref{rem:smallL}.

\begin{lem}
	\label{existence}
	For any $L>0$, $\eL$ in \eqref{eq:pekaren} is attained, and there exists a universal constant $C>0$ such that $\eL>-C$. Moreover, there exists $L_0>0$ such that, for $L<L_0$, $\EL(\psi)>0$ for any non-constant $L^2$-normalized $\psi$.
\end{lem}
\begin{proof}
	We consider any $L^2$-normalized $\psi \in H^1(\TL)$ and begin by observing that in terms of the Fourier coefficients we have
	\begin{align}
	&W_L(\psi)=\sum_{0\neq k\in \frac {2\pi}{L} \mathbb{Z}^3} \frac {|(\ropsi)_k|^2} {|k|^2} ,\\
	&(\ropsi)_k= \sum_{j\in \frac {2\pi}{L} \mathbb{Z}^3} \frac {\bar{\psi}_j \psi_{j+k}}{L^{3/2}}=(\rho_{\hat{\psi}})_k+\frac{\bar{\psi}_0\psi_k}{L^{3/2}}+\frac{\bar{\psi}_{-k}\psi_0}{L^{3/2}}.
	\end{align}
	By Parseval's identity $|\psi_0|\leq 1$ and thus, using the Cauchy--Schwarz inequality, we can deduce that
	\begin{align}
	\label{robound}
	|(\ropsi)_k|^2\leq 
	\begin{cases}
	L^{-3},\\
	3|(\rho_{\hat{\psi}})_k|^2+\frac 3 {L^3} (|\psi_k|^2+|\psi_{-k}|^2).
	\end{cases}
	\end{align}
	Therefore
	\begin{align}
	W_L(\psi)&\leq 3 \left(\sum_{0\neq k\in \frac {2\pi}{L} \mathbb{Z}^3} \frac {|(\rho_{\hat{\psi}})_k|^2} {|k|^2}\right) + \frac 6 {L^3} \left(\sum_{0\neq k\in \frac {2\pi}{L} \mathbb{Z}^3} \frac {|\psi_k|^2} {|k|^2}\right) \nonumber\\
	&\leq 3W_L(\hat{\psi})+\frac {6}{(2\pi)^2 L}\|\hat{\psi}\|_{L^2(\TL)}^2.
	\end{align}
	We can bound both terms on the r.h.s. in two different ways, one which is good for small $L$ and one which is good for all the other $L$. Indeed, by applying estimate \eqref{robound} and using the Poincar\'e-Sobolev inequality (see \cite{lieb2001analysis}, chapter 8) on the zero-mean function $\hat{\psi}$, we get
	\begin{align}
	W_L(\hat{\psi})&\leq \left(\sum_{0\neq k\in \frac {2\pi}{L} \mathbb{Z}^3} \frac{|(\rho_{\hat{\psi}})_k|^2}{|k|^4}\right)^{1/2}\left(\sum_{0\neq k\in \frac {2\pi}{L} \mathbb{Z}^3} |(\rho_{\hat{\psi}})_k|^2\right)^{1/2}\lesssim L^2\|(\rho_{\hat{\psi}})_k\|_{l^{\infty}}\|\hat{\psi}\|_{L^4(\TL)}^2\nonumber\\
	&\lesssim L^{1/2}\|\hat{\psi}\|_{L^4(\TL)}^2\lesssim L\|\hat{\psi}\|_{L^6(\TL)}^2\lesssim LT_L(\hat{\psi})=L T_L(\psi).
	\end{align}
	Moreover, 
	\begin{align}
	L^{-1} \|\hat{\psi}\|_{L^2(\TL)}^2\lesssim L T_L(\hat{\psi})=L T_L(\psi).
	\end{align}
	Therefore, we can conclude that
	\begin{align}
	W_L(\psi)\lesssim L T_L(\psi)\;\;\Rightarrow\;\; \EL(\psi)\geq (1-CL)T_L(\psi).
	\end{align}
	Thus, for $L< L_0:=C^{-1}$, either ${\psi\equiv \const}$ and $\EL(\psi)=0$ or $\EL(\psi)\gtrsim T_L(\psi)>0$. Moreover, this also implies
	\begin{align}
	\EL(\psi)\gtrsim T_L(\psi)\geq \frac {(2\pi)^2}{2L_0^2}\|\hat\psi\|_2^2+\frac 1 2 T_L(\psi)\gtrsim \dist^2_{H^1}\left(\Theta_L\left(\frac 1 {L^{3/2}}\right), \psi\right),
	\end{align}
	which is the analogue  of \eqref{globalquadbound} from Theorem \ref{uniquenessANDcoercivity}  in the case $L< L_0$.
	
	We now proceed to study the more interesting regime $L\geq L_0$. By Lemma \ref{kernelbounds}, splitting $\dist^{-1}_{\TL}(x,\cdot)$ into an $L^{3/2}$ part and the remaining $L^{\infty}$ part (whose norms can be chosen to be  proportional to $\varepsilon$ and $\varepsilon^{-1}$, respectively, for any $\varepsilon>0$), and by applying again the Poincar\'e-Sobolev inequality, we obtain
	\begin{align}
	W_L(\hat{\psi})\leq \int_{\TL\times \TL} \frac{\rho_{\hat{\psi}}(x) \rho_{\hat{\psi}}(y)} {4\pi \dist_{\TL}(x,y)} dx dy + \frac C L\lesssim \varepsilon\|\hat{\psi}\|_{L^6(\TL)}^2+\varepsilon^{-1}+1\leq \frac{T_L(\psi)} 6+C.
	\end{align}
	Moreover, since $L\geq L_0$, trivially $ L^{-1} \|\hat{\psi}\|_{L^2(\TL)}^2\lesssim 1$ and we can conclude that for any $L^2$-normalized $\psi\in H^1(\TL)$
	\begin{align}
	\label{Tbounds}
	W_L(\psi)\leq \frac{T_L(\psi)} 2+ C \;\;\Rightarrow\;\;
	\EL(\psi)\geq \frac {T_L(\psi)} 2 - C.
	\end{align} 
	From this we can infer that $\eL\geq-C$ for any $L$. To show existence of minimizers, we observe that by \eqref{Tbounds} any minimizing sequence $\psi_n$ on $\TL$ must be bounded in $H^1(\TL)$. Therefore, there exists a subsequence (which we still denote by $\psi_n$ for simplicity)  that  converges weakly in $H^1(\TL)$ and strongly in $L^p(\TL)$, for any $1\leq p<6$ to some $\psi$ (by the Banach-Alaoglu Theorem and the Rellich-Kondrachov embedding Theorem). The limit function $\psi$ is $L^2$-normalized and 
	\begin{align}
	T_L(\psi)\leq \liminf_{n\to \infty} T_L(\psi_n)
	\end{align} 
	by weak lower semicontinuity of the norm. Using the $L^4$-convergence of $\psi_n$ to $\psi$ and the fact that $\|\cdot\|_{\mathring{H}^{-1}(\TL)}\lesssim L\|\cdot\|_{L^2(\TL)}$, we finally obtain
	\begin{align}
	|W_L(\psi_n)-W_L(\psi)|&=\left(\|\rho_{\psi}\|_{\mathring{H}^{-1}(\TL)}+\|\rho_{\psi_n}\|_{\mathring{H}^{-1}(\TL)}\right)\left|\|\rho_{\psi}\|_{\mathring{H}^{-1}(\TL)}-\|\rho_{\psi_n}\|_{\mathring{H}^{-1}(\TL)}\right|\nonumber\\
	&\lesssim L\|\rho_{\psi_n}-\rho_{\psi}\|_{\mathring{H}^{-1}(\TL)}\lesssim L^2 \|\rho_{\psi_n}-\rho_{\psi}\|_{L^2(\TL)}\nonumber\\
	&\leq L^2\|\psi_n-\psi\|_{L^4(\TL)}\left(\|\psi_n\|_{L^4(\TL)}+\|\psi\|_{L^4(\TL)}\right)\to 0.
	\end{align}
	This implies that 
	\begin{align}
	\EL(\psi)\leq \liminf_{n\to \infty} \EL(\psi_n)=\eL,
	\end{align}
	and thus that $\psi$ is a minimizer. Note that, since $\EL(\psi_n)\to \eL=\EL(\psi)$ by definition of $\psi_n$ and, as shown, $W_L(\psi_n)\to W_L(\psi)$, it also holds
	\begin{align}
	T_L(\psi_n)=\EL(\psi_n)+W_L(\psi_n)\to \EL(\psi)+W_L(\psi)=T_L(\psi)
	\end{align} 
	which implies that $\psi_n$ actually converges to $\psi$ strongly in $H^1(\TL)$.
\end{proof}

Once we have shown existence of minimizers, we need to investigate more carefully their properties. Some of them are derived in the following Lemma. Recall that 
\begin{align}
\label{eq:Vsigmatorus}
V_{\psi}= -2(-\Delta_L)^{-1/2} \psi, \quad \sigma_{\psi}= (- \Delta_L)^{-1/2} |\psi|^2,
\end{align}
and that, as stated above, we call  any property universal which does not depend on $L\geq L_0$.

\begin{lem}
	\label{minprop}
	Let $\psi\in \MinLe$ (as defined in \eqref{eq:minEL}). Then $\psi$ satisfies the following Euler-Lagrange equation
	\begin{align}
	\label{eulag}
	&(-\Delta_L+V_{\sigma_{\psi}}-\LagrML_{\psi})\psi=0, \quad \text{with} \quad \LagrML_{\psi}= T_L(\psi)-2W_L(\psi). 
	\end{align}
	Moreover, $\psi \in C^{\infty}(\TL)$, is universally bounded in $H^2(\TL)$ (and therefore in $L^{\infty}(\TL)$), has constant phase and never vanishes. Finally, any $L^2$-normalized sequence $f_n\in H^1(\mathbb{T}^3_{L_n})$ such that $\ELn(f_n)$ is universally bounded, is universally bounded in $H^1(\mathbb{T}^3_{L_n})$.
\end{lem}

\begin{proof}	
	The fact that sequences $f_n\in H^1(\mathbb{T}^3_{L_n})$ of $L^2$-normalized functions for which $\ELn$ is universally bounded are universally bounded in $H^1(\mathbb{T}^3_{L_n})$ follows trivially from estimate \eqref{Tbounds}. This immediately yields a universal bound on the $H^1$-norm of minimizers. 
	
	The Euler--Lagrange equation \eqref{eulag} for the problem is derived by standard computations omitted here. By Lemma \ref{kernelbounds} and by splitting $(\dist_{\TL}(0,\,\cdot\,))^{-1}$ in its $L^{3/2}$ and $L^{\infty}$ parts, we have
	\begin{align}
	|V_{\sigma_{\psi}}(x)|\leq 2\int_{\TL} \frac 1 {\dist_{\TL}(x,y)} |\psi(y)|^2 dy +\frac C L\lesssim \left( \|\psi\|_{L^6(\TL)}^2+1\right)\lesssim\left( T_L(\psi)+1\right).
	\end{align}
	Therefore, by the universal $H^1$-boundedness of minimizers, $V_{\sigma_{\psi}}$ is universally bounded in $L^{\infty}(\TL)$, for any $\psi \in \MinLe$. This immediately allows to conclude universal $\mathring{H}^2$ (and hence $H^2$) bounds for functions in $\MinLe$, using the Euler--Lagrange equation \eqref{eulag}, Lemma \ref{existence} and the universal $H^1$-boundedness of minimizers, which guarantee that 
	\begin{align*}
	0\geq \LagrML_{\psi}=2\EL(\psi)-T_L(\psi)\geq -C.
	\end{align*}
	Since $L\geq L_0$, universal $H^2$-boundedness also implies universal $L^{\infty}$-boundedness of minimizers by the Sobolev inequality. 
	
	For any $L>0$, any $\psi\in \MinLe$ satisfies $\eqref{eulag}$, is in $H^1(\TL)$ and is such that $V_{\sigma_{\psi}}\in L^{\infty}(\TL)$. Therefore $\psi$ also satisfies, for any $\lambda>0$
	\begin{align}
	\psi=(-\Delta_L+\lambda)^{-1}(-V_{\sigma_{\psi}}+\LagrML_{\psi}+\lambda)\psi.
	\end{align}
	In particular, by a bootstrap argument we can conclude that $\psi\in C^{\infty}(\TL)$. Moreover, picking $\lambda>-\LagrML_{\psi}+\|V_{\sigma_{\psi}}\|_{L^{\infty}(\TL)}$ and using that $(-\Delta_L+\lambda)^{-1}$ is positivity improving, we can also conclude that if $\psi\geq 0$ then $\psi>0$. By the convexity properties of the kinetic energy (see \cite{lieb2001analysis}, Theorem 7.8), we have that $T_L(|\psi|)\leq T_L(\psi)$ which implies that if $\psi\in \MinLe$ then  $T_L(\psi)=T_L(|\psi|)$ and also $|\psi|\in \MinLe$. Hence both $\psi$ and $|\psi|$ are eigenfunctions of the least and \emph{simple} (by positivity of one of the eigenfunctions) eigenvalue $\LagrML_{\psi}=\LagrML_{|\psi|}$ of the Schr\"odinger operator $-\Delta_L+V_{\sigma_{\psi}}$, which allows us to infer that $\psi$ has constant phase and never vanishes. 
\end{proof}

We now proceed to develop the tools that will allow to show the validity of Theorem \ref{uniquenessANDcoercivity}. We begin with a simple Lemma. 

\begin{lem}
	\label{lem:Hsrho}
	For $\psi\in H^1(\TL)$, 
	\begin{align}
	\|\rho_{\psi}\|_{\mathring{H}^{1/8}(\TL)}\lesssim \|\psi\|_{H^{1}(\TL)}^{2}.
	\end{align}
\end{lem}
\begin{proof}
	We have
	\begin{align}
	\label{eq:rhophibound}
	&\|\rho_{\psi}\|^2_{\mathring{H}^{1/8}(\TL)}=|\langle \nabla \rho_{\psi} | \nabla ((-\Delta_{L})^{-7/8}\rho_{\psi})\rangle|\nonumber\\
	&=2 \left|\int_{\TL}|\psi(x)| \nabla(|\psi(x)|) \cdot \nabla_x \left(\sum_{0\neq k\in \frac {2\pi}{L} \mathbb{Z}^3} \frac{(\rho_{\psi})_k}{|k|^{7/4}}\frac{e^{ik\cdot x}}{L^{3/2}}\right) dx\right|\nonumber\\
	&=\left|\sum_{i=1}^3\int_{\TL}|\psi(x)| \partial_i \left(|\psi(x)|\right) \sum_{0\neq k\in \frac {2\pi}{L} \mathbb{Z}^3} \frac{k_i(\rho_{\psi})_{k}}{|k|^{7/4}} \frac {e^{ik\cdot x}}{L^{3/2}}dx\right|.
	\end{align}
	We define
	\begin{align}
	g_i(x):=\sum_{0\neq k\in \frac {2\pi}{L} \mathbb{Z}^3} \frac{k_i(\rho_{\psi})_{k}}{|k|^{7/4}} \frac {e^{ik\cdot x}}{L^{3/2}},
	\end{align}
	and observe that $(g_i)_0=0$ and  $|(g_i)_k|=\frac {|k_i(\rho_{\psi})_{k}|}{|k|^{7/4}}\leq\frac {|(\rho_{\psi})_k|}{|k|^{3/4}}$ for $k\neq 0$. These estimates on the Fourier coefficients of $g_i$ imply that, for $i=1,2,3$, 
	\begin{align}
	\|g_i\|_{\mathring{H}^{3/4}(\TL)}^2=\sum_{0\neq k\in \frac {2\pi}{L} \mathbb{Z}^3} |k|^{3/2} |(g_i)_k|^2\leq\sum_{0\neq k\in \frac {2\pi}{L} \mathbb{Z}^3} |(\rho_{\psi})_k|^2\leq\|\psi\|_{L^4(\mathbb{T}^3_{L})}^4.
	\end{align}
	Moreover, using the fractional Sobolev embeddings (see, for example, \cite{benyi2013sobolev}) and that $g_i$ has zero mean, we have
	\begin{align}
	\|g_i\|_{L^4(\TL)}\lesssim \|g_i\|_{\mathring{H}^{3/4}(\TL)}\leq \|\psi\|_{L^4(\TL)}^2.
	\end{align} 
	Applying these results to \eqref{eq:rhophibound} and using H\"older's inequality two times, the Poincar\'e-Sobolev inequality and the convexity properties of the kinetic energy (see \cite{lieb2001analysis}, Theorem 7.8), we  conclude that
	\begin{align}
	\|\rho_{\psi}\|^2_{\mathring{H}^{1/8}(\mathbb{T}^3_{L})}&\lesssim  \|\psi\|_{L^4(\TL)}\|g^{1/8}_i\|_{L^4(\mathbb{T}^3_{L})}\|\nabla(|\psi|)\|_{L^2(\TL)}\leq \|\psi\|^3_{L^4(\TL)}\|\psi\|_{\mathring{H}^1(\TL)}\nonumber\\
	&\leq \|\psi\|_{L^2(\TL)}^{3/4}\|\psi\|_{L^6(\TL)}^{9/4}\|\psi\|_{\mathring{H}^1(\TL)}\lesssim \|\psi\|_{H^1(\TL)}^{4}.
	\end{align}
\end{proof}

Our next goal is to show that $\eL\to \einf$ as $L\to \infty$, and that in the large $L$ regime the states that are relevant for the minimization of $\EL$ are necessarily close to the full space minimizer (or any of its translates). This is a key ingredient for the discussion carried out in the following sections, and is stated in a precise way in the next proposition. The coercivity results obtained in \cite{lenzmann2009uniqueness} are of fundamental importance here as they guarantee that, at least for the full space model, low energy states are close to minimizers. 

We recall that the full-space Pekar functional, defined in \eqref{eq:FullSpaceE}, admits a unique positive and radial minimizer $\Emininf$ which is also smooth (see \eqref{eq:SpaceSurfaceofMin}), and we introduce the notation 
\begin{align}
\label{eq:PsiL}
\Emininf_L:=\Emininf \chi_{[-L/2,L/2]^3}.
\end{align} 
Note that $\Emininf_L\in H^1(\TL)$, by radiality and regularity of $\Emininf$. 

\begin{prop}
	\label{prop:En&MinConv}
	We have
	\begin{align}
	\lim_{L\to \infty}\eL= \einf.
	\end{align}
	Moreover, for any $\varepsilon>0$ there exist $L_{\varepsilon}$ and $\delta_{\varepsilon}$ such that for any $L>L_{\varepsilon}$ and any $L^2$-normalized $\psi \in H^1(\TL)$ with $\EL(\psi)-\eL<\delta_{\varepsilon}$,
	\begin{align}
	\label{eq:convofmin}
	\dist_{H^1}\left(\Theta_L(\psi),\Emininf_L\right)\leq \varepsilon, \quad |\LagrML_{\psi}-\LagrMinf_{\Emininf}|\leq \varepsilon,
	\end{align}
	where $\Theta_L(\psi)$, $\Emininf_L$, $\LagrML_{\psi}$ and $\LagrMinf_{\Emininf}$ are defined in \eqref{eq:invariantsurfaceE}, \eqref{eq:PsiL}, \eqref{eulag} and \eqref{eq:fullspaceVmu}, respectively.
\end{prop}
\begin{proof}
	We first show that $\limsup_{L\to \infty} \eL\leq \einf$ by using $\Emininf_L$ as a trial state for $\EL$. Observe that ${\|\Emininf_L\|_{L^2(\TL)}\to 1}$ and $T_L(\Emininf_L)\to T(\Emininf)$ as $L\to \infty$. To estimate the difference of the interaction terms we note that $\Emininf_L(\Emininf-\Emininf_L)=0$ and therefore
	\begin{align}
	|W_L(\Emininf_L)-W(\Emininf)|\leq |W_L(\Emininf_L)-W(\Emininf_L)|+W(\Emininf-\Emininf_L)+2\bra{(\Emininf-\Emininf_L)^2}\ket{\Delta_{\Rtre}^{-1} \Emininf_L^2}.
	\end{align}
	By dominated convergence, the last two terms converge to zero as $L\to \infty$. On the other hand, by Lemma \ref{kernelbounds} and since $\Emininf$ is normalized
	\begin{align}
	&|W_L(\Emininf_L)-W(\Emininf_L)|\leq \frac C L +\frac 1 {4\pi}\int_{[-L/2,L/2]^6} \Emininf_L(x)^2\Emininf_L(y)^2 \left|\frac 1 {\dist_{\TL}(x,y)}-\frac 1 {|x-y|}\right|dxdy.
	\end{align}
	Moreover, since $\dist_{\TL}(x,y)=|x-y|$ for $x,y\in [-L/4,L/4]^3$ and using the symmetry and the positivity of the integral kernel and the fact that $\dist_{\TL}(x,y)\leq |x-y|$, we get
	\begin{align}
	\label{eq:InteractionBounds}
	&\int_{[-L/2,L/2]^6} \Emininf_L(x)^2\Emininf_L(y)^2 \left|\frac 1 {\dist_{\TL}(x,y)}-\frac 1 {|x-y|}\right|dxdy\notag\\
	&\leq 2\int_{[-L/2,L/2]^3} \Emininf_L^2(x)\left(\int_{[-L/2,L/2]^3}\frac{(\Emininf_L-\Emininf_{L/2})^2(y)} {\dist_{\TL}(x,y)}dy\right)dx. 
	\end{align}
	Finally, by splitting $\dist_{\TL}^{-1}(x,\cdot)$ in its $L^{\infty}$ and $L^1$ parts and using that $\Emininf$ is normalized, we can bound the r.h.s. of \eqref{eq:InteractionBounds} by $\left(C_1\|\Emininf_L-\Emininf_{L/2}\|_2^2+C_2\|\Emininf_L-\Emininf_{L/2}\|_{\infty}^2\right)$, which vanishes as $L\to \infty$, since $\Emininf(x)\xrightarrow{|x|\to\infty}0$. 
		Putting the pieces together, we  conclude that
	\begin{align}
	|W_L(\Emininf_L)-W(\Emininf_L)|=o_L(1).
	\end{align}
	This shows our first claim, since
	\begin{align}
	\eL\leq\EL(\Emininf_L/\|\Emininf_L\|_2)=\frac 1 {\|\Emininf_L\|_2^2}\left(T_L(\Emininf_L)-\frac 1 {\|\Emininf_L\|_2^2}W_L(\Emininf_L)\right)\to \einf.
	\end{align}
	
	\vspace{3mm}
	
	We now proceed to show that 
	\begin{align}
	\label{eq:liminf}
	\liminf_{L\to \infty} \eL\geq \einf
	\end{align}
	and the validity of \eqref{eq:convofmin} using IMS localization. We shall show that for any $L^2$-normalized sequence $\psi_n \in H^1(\TLn)$ with $L_n\to \infty$ such that 
	\begin{align}
	\label{eq:lowenergy}
	\ELn(\psi_n)-\eLn\to 0,
	\end{align} 
	we have
	\begin{align}
	\label{eq:sequenceclaims}
	\liminf_{n \to \infty}\ELn(\psi_n)\geq \einf, \quad
	\lim_{n\to \infty}\dist_{H^1}\left(\Theta_{L_n}(\psi_n),\Emininf_{L_n}\right)=0, \quad \lim_{n\to \infty}|\mu^{L_n}_{\psi_n}-\LagrMinf_{\Emininf}|=0,
	\end{align}
	 which implies the claim of the proposition.
	
	Pick $\eta\in C^{\infty}(\Rtre)$ with $\text{supp}(\eta)\subset B_1$ and $\|\eta\|_2=1$. We denote by $\eta_R$ the rescaled copy of $\eta$ supported on $B_R$ with $L^2$-norm equal to $1$. As long as $R\leq L/2$, $\eta_R \in C^{\infty}(\TL)$ and we then consider the translates $\eta_R^y$ for any $y\in \TL$. Given $\psi\in H^1(\TL)$, we also define 
	\begin{align}
 	\psi_R^y:=\psi \eta_R^y/\|\psi \eta_R^y\|_2.
	\end{align}

	By standard properties of IMS localization, for any $R\leq L/2$, we have
	\begin{align}
	\label{Test}
	\int_{\TL}  T_L(\psi_R^y) \|\psi\eta_R^y\|_2^2dy=\int_{\TL} T_L(\psi\eta_R^y) dy=T_L(\psi)+\frac {\int |\nabla \eta|^2} {R^2}.
	\end{align} 
	Moreover, by using that $|\psi|^2=\int_{\TL} |\psi \eta_R^y|^2 dy=\int_{\TL} |\psi_R^y|^2 \|\psi \eta_R^y\|^2 dy$ and completing the square
	\begin{align}
	\label{West}
	W_L(\psi)=\int_{\TL} \left[W_L(\psi_R^y)-\left\||\psi_R^y|^2-|\psi|^2\right\|_{\mathring{H}^{-1}(\TL)}^2\right] \|\psi \eta_R^y\|_2^2dy.
	\end{align}
	Combining \eqref{Test} and \eqref{West}, we therefore obtain 
	\begin{align}
	\EL(\psi)+\frac C {R^2}=\int_{\TL} \left[\EL(\psi_R^y)+\left\||\psi_R^y|^2-|\psi|^2\right\|_{\mathring{H}^{-1}(\TL)}^2\right]\|\psi \eta_R^y\|_2^2dy.
	\end{align}
	Since the integrand on the r.h.s. is equal to the l.h.s. on average (indeed $\|\psi \eta_R^y\|_2^2dy$ is a probability measure) there exists $\bar y\in \TL$ such that
	\begin{align}
	\EL(\psi_R^{\bar y})+\left\||\psi_R^{\bar y}|^2-|\psi|^2\right\|_{\mathring{H}^{-1}(\TL)}^2\leq\EL(\psi)+\frac C {R^2}.
	\end{align}
	This fact has several consequences and it is particularly useful if we apply it to our sequence $\psi_n$ with a radius $R=R_n\leq L_n/2$ (we take for simplicity $R=L_n/4$). Indeed, by the above discussion and \eqref{eq:lowenergy}, we obtain that there exists $\bary_n\in \TLn$ such that the $L^2$-normalized functions
	\begin{align}
	\bar\psi_n:=\frac {\psi_n\eta^{\bary_n}_{L_n/4}}{\|\psi_n\eta^{\bary_n}_{L_n/4}\|_2}
	\end{align}
	are competitors both for the minimization of $\ELn$ and $\Einf$ (indeed, $\bar\psi_n$ can then be thought of as a function in $C^{\infty}_c(\Rtre)$, supported on $B_{L_n/4}$) and satisfy
	\begin{align}
	\label{phinprop}
	\ELn(\bar\psi_n)&\leq \ELn(\psi_n)+\frac C {L_n^2}\leq  \eLn+o_{L_n}(1), \nonumber\\
	&\|\rho_{\bar\psi_n}-\rho_{\psi_n}\|^2_{\mathring{H}^{-1}(\mathbb{T}^3_{L_n})}\leq \frac C {L_n^2}. 
	\end{align}
	In other words, we can localize any element of our sequence $\psi_n$ to a ball of radius $R= L_n/4$ with an energy expense of order $L_n^{-2}$, and  the localized function is close (in the sense of the second line of \eqref{phinprop}) to $\psi_n$ itself, up to an error again of order $L_n^{-2}$.
	
	 Moreover $T_{L_n}(\bar \psi_n)=T(\bar\psi_n)$ and, using Lemma \ref{kernelbounds} and the fact that $\dist_{\TLn}(x,y)=|x-y|$ for all $x,y \in B_{L_n/4}$, we have
	\begin{align}
	\label{eq:ApproxInteraction}
	|W_{L_n}(\bar\psi_n)-W(\bar\psi_n)|\lesssim \frac 1 {L_n}.
	\end{align}
	Therefore, using \eqref{phinprop}
	\begin{align}
	\label{energysequence}
	\einf\leq \Einf(\bar\psi_n)\leq\ELn(\bar\psi)+\frac C {L_n}\leq \eLn+o_{L_n}(1),
	\end{align}
	which shows the first claim in \eqref{eq:sequenceclaims}. By Theorem \ref{h1coerc} and \eqref{energysequence}, it also follows that
	\begin{align}
	\dist_{H^1} \left( \Theta(\Emininf), \bar\psi_n\right)\xrightarrow{n\to\infty} 0.
	\end{align} 
	Hence, up to an $n$-dependent translation and change of phase (which we can both assume to be zero without loss of generality  by suitably redefining $\psi_n$),  $\bar\psi_n\xrightarrow{H^1(\Rtre)} \Emininf$, and the convergence also holds in $L^p(\Rtre)$ for any $2\leq p \leq6$. From this and the second line of \eqref{phinprop}, we would like to deduce that also $\psi_n$ and $\Emininf_{L_n}$ are close. We first note that, by a simple application of H\"older's inequality, it follows that for any $f\in L^2(\TL)$ with zero mean 
	\begin{align}
	\label{l2tofractsob}
	\| f\|_{L^2(\TL)}^2&\leq \left(\sum_{0\neq k\in \frac {2\pi}{L} \mathbb{Z}^3} |k|^{1/4} |f_k|^2\right)^{8/9}\left(\sum_{0\neq k\in \frac {2\pi}{L} \mathbb{Z}^3} |k|^{-2} |f_k|^2\right)^{1/9}\nonumber\\
	&=\|f\|_{\mathring{H}^{1/8}(\TL)}^{16/9}\|f\|_{\mathring{H}^{-1}(\TL)}^{2/9}.
	\end{align} 
	We combine this with \eqref{phinprop} and apply it to the zero mean function $(\rho_{\psi_n}-\rho_{\bar\psi_n})$ , obtaining
	\begin{align}
	\|\rho_{\bar\psi_n}-\rho_{\psi_n}\|_{L^2(\mathbb{T}^3_{L_n})}^2\lesssim \left(\frac{\|\rho_{\psi_n}\|^2_{\mathring{H}^{1/8}(\TLn)}+\|\rho_{\bar\psi_n}\|^2_{\mathring{H}^{1/8}(\mathbb{T}^3_{L_n})}}{L_n^{1/4}}\right)^{8/9}.
	\end{align}
	Applying Lemma \ref{lem:Hsrho} to $\psi_n$ and $\bar\psi_n$ (which are uniformly bounded in $H^1$ by Lemma \ref{minprop}) we  conclude that $(\rho_{\psi_n}-\rho_{\bar\psi_n})\xrightarrow{L^2}0$.
	
	As a consequence, since $\psi_n$ and $\bar\psi_n$ have the same phase, $\psi_n$ and $\bar\psi_n$ are arbitrarily close in $L^4$. Indeed,
	\begin{align}
	\|\psi_n-\bar\psi_n\|_{L^4(\mathbb{T}^3_{L_n})}^4= \int_{\mathbb{T}^3_{L_n}} \left||\psi_n|-|\bar\psi_n|\right|^4dx\leq\int_{\mathbb{T}^3_{L_n}} (\rho_{\psi_n}-\rho_{\bar\psi_n})^2 dx\xrightarrow{n\to \infty} 0. 
	\end{align}
	By the identification of $\mathbb{T}^3_{L_n}$ with $[-L_n/2,L_n/2]^3$, we finally get $\|\psi_n-\Emininf\|_{L^4(\Rtre)}\to 0$, if $\psi_n$ is set to be $0$ outside $[-L_n/2,L_n/2]^3$. Moreover, $\psi_n$ converges to $\Emininf$ in $L^p(\Rtre)$ for any $2\leq p <6$, since $\|\psi_n\|_2=1$, $\psi_n\xrightarrow{L^4}\Emininf$, $\|\Emininf\|_2=1$ and $\|\psi_n\|_p$ is uniformly bounded for any $2\leq p \leq 6$. 
	
	To show the second claim in \eqref{eq:sequenceclaims}, we need to show that the convergence actually holds in $H^1(\TLn)$, i.e., that $\|\psi_n-\Emininf_{L_n}\|_{H^1(\mathbb{T}^3_{L_n})}\to 0$. First, we show convergence in $H^1(B_R)$ for fixed $R$. Note that 
	\begin{align}
	\label{eq:convergenceofnorms}
	\left(\|\psi_n\|_{H^1(\TLn)}-\|\Emininf\|_{H^1(\Rtre)}\right)\to 0,
	\end{align} 
	since 
	\begin{align}
	|T_{L_n}(\psi_n)-T_{L_n}(\bar\psi_n)|&=|\ELn(\psi_n)+W_{L_n}(\psi_n)-\ELn(\bar\psi_n)+W_{L_n}(\bar\psi_n)|\nonumber\\
	&\leq|\ELn(\psi_n)-\ELn(\bar\psi_n)|+|W_{L_n}(\psi_n)-W_{L_n}(\bar\psi_n)|\to 0,
	\end{align}
	and $T_{L_n}(\bar\psi_n)=T(\bar\psi_n) \to T(\Emininf)$ by $H^1$ convergence. Moreover, given that $\psi_n$ is uniformly bounded in $H^1(B_R)$ and $\psi_n \to \Emininf$ in $L^2(B_R)$, we have $\psi_n\rightharpoonup \Emininf$ in $H^1(B_R)$ for any $R$ and this, together with \eqref{eq:convergenceofnorms} and weak lower semicontinuity of the norms, implies $\psi_n\to \Psi$ in $H^1(B_R)$ for any $R$. 
	Finally, for any $\varepsilon>0$ there exists $R=R(\varepsilon)$ such that $\|\Emininf\|_{H^1(B_R^c)}\leq \varepsilon$ and, using strong $H^1$-convergence on balls and again \eqref{eq:convergenceofnorms}, we obtain
	\begin{align}
	\|\psi_n-\Emininf_{L_n}\|_{H^1(\mathbb{T}^3_{L_n})}&\leq \|\psi_n-\Emininf\|_{H^1(B_R)}+\|\psi_n-\Emininf\|_{H^1([-L_n/2,L_n/2]^3\setminus B_R)}\nonumber\\
	&\leq \|\psi_n-\Emininf\|_{H^1(B_R)}+\|\psi_n\|_{H^1([-L_n/2,L_n/2]^3\setminus B_R)}+\|\Emininf\|_{H^1([-L_n/2,L_n/2]^3\setminus B_R)}\nonumber\\
	&\leq \|\psi_n-\Emininf\|_{H^1(B_R)}+2\varepsilon+o_n(1)\to 2\varepsilon, 
	\end{align}
 	which concludes the proof of the second claim in \eqref{eq:sequenceclaims}. 
	
	Finally, we show the third claim in \eqref{eq:sequenceclaims}. This simply follows from the previous bounds, which guarantee that $\ELn(\psi_n)\to \einf$ and $T_{L_n}(\psi_n)\to T(\Psi)$ and hence 
 	\begin{align}
 	\LagrML_{\psi_n}=T_{L_n}(\psi_n)-2W_{L_n}(\psi_n)=2\ELn(\psi_n)-T_{L_n}(\psi_n)\to 2\einf-T(\Psi)=\LagrMinf_{\Emininf}.
 	\end{align}
\end{proof}

We conclude this section with a simple corollary of Proposition \ref{prop:En&MinConv}.
\begin{cor}
	\label{cor:aperiodicity}
There exists $L^*$ such that for $L>L^*$ and any $\psi\in\MinLe$ we have $\psi\neq \psi^y$ for  $0\neq y\in \TL$.
\end{cor}
\begin{proof}
It is clearly sufficient to show the claim for $\psi \in \MinLe$ such that 
\begin{align}
\dist_{H^1}(\Theta_L(\psi),\Psi_L)=\|\psi-\Psi_L\|_{H^1(\TL)}
\end{align}
and for $y \in \TL$ such that $|y|\geq L/4$ (indeed, if the claim fails for some $y'$ such that $|y'|<L/4$ it also fails for some $y$ such that $|y|\geq L/4$). For any such $\psi$ and $y$, Proposition \ref{prop:En&MinConv} and the fact that $\Psi\neq \Psi^y$ for any $y\in \Rtre$ guarantee the existence of $L^*$ such that for any $L>L^*$ we have
\begin{align}
\|\psi-\psi^y\|_{H^1(\TL)}\geq \|\Psi_L^y-\Psi_L\|_{H^1(\TL)}-2\|\psi-\Psi_L\|_{H^1(\TL)}\geq C>0
\end{align}
and this completes the proof.
\end{proof}

\subsubsection{Study of the Hessian of $\EL$} \label{subsubsec:HessianEL} 
In this section we study the Hessian of $\EL$ at its minimizers, showing that it is strictly positive, universally, for $L$ big enough. Positivity is of course understood up to the trivial zero modes resulting from the symmetries of the problem (translations and changes of phase). This is obtained by comparing $\EL$ with $\mathcal{E}$ and exploiting Theorem \ref{h1coerc}. 

For any minimizer $\psi\in \MinLe$, the Hessian of $\EL$ at $\psi$ is defined by 
\begin{align}
\lim_{\varepsilon\to 0} \frac 1 {\varepsilon^2} \left(\EL\left(\frac{\psi+\varepsilon f}{\|\psi+\varepsilon f\|_2}\right)-\eL\right)=H^{\EL}_{\psi}(f)\quad \forall f\in H^1(\TL).
\end{align}
An explicit computation gives
\begin{align}
\label{Hessian}
H^{\EL}_{\psi}(f)=\expval{\LL_{\psi}}{\Im f}+\expval{Q_{\psi}(\LL_{\psi}-4\XL_{\psi}) Q_{\psi}}{\Re f}, 
\end{align}
with $Q_{\psi}=\unit- |\psi \rangle\langle\psi|$ and
\begin{align}
\LL_{\psi}:= - \Delta_{L} +V_{\sigma_{\psi}}-\LagrML_{\psi} \ , \quad \XL_{\psi}(x,y):=  \psi(x)\Tker \psi(y) \label{def:lpm}.
\end{align}
(We use the same notation for the operator $\XL_\psi$ and its integral kernel for simplicity.) 
We recall that $\LagrML_{\psi}=T_L(\psi)-2W_L(\psi)$ (see \eqref{eulag})  and that $V_{\sigma_{\psi}}= -2(-\Delta_L)^{-1} \rho_{\psi}$ (see \eqref{eq:Vsigmatorus})  and we note that $\LL_{\psi} \psi=0$ is exactly the Euler--Lagrange equation derived in Lemma \ref{minprop}.  

By minimality of $\psi$, we know that $\infspec L = \infspec Q (L-4X) Q =0$, since both operators are clearly nonnegative and $\psi$ is in the kernel of both of them. Moreover, $\ker \LL_{\psi}= \spn\{\psi\}$, since it is a Schr\"odinger operator of least (simple) eigenvalue $0$. The situation is more complicated for $Q_{\psi}(\LL_{\psi}-4\XL_{\psi}) Q_{\psi}$, whose kernel contains at least $\psi$ and $\partial_i \psi$ (by the translation invariance of the problem). Since both $\LL_{\psi}$ and $Q_{\psi}(\LL_{\psi}-4\XL_{\psi})Q_{\psi}$ have compact resolvents (they are given by bounded perturbations of $-\Delta_L$), they both have discrete spectrum. Our aim is two-fold: first we need to show that the kernel of $Q_{\psi}(\LL_{\psi}-4\XL_{\psi})Q_{\psi}$ is exactly spanned by $\psi$ and its partial derivatives, secondly we want to show that the spectral gap (above the trivial zero modes) of both operators is bounded by a universal positive constant. 

Before stating the main result of this section, we introduce the relevant full-space objects: let again $\Emininf$ be the unique positive and radial full-space minimizer of the Pekar functional~\eqref{eq:FullSpaceE} and, analogously to \eqref{def:lpm}, define
\begin{align}
L_{\Emininf}:= - \Delta_{\R^{3}} +V_{\sigma_{\Emininf}}-\LagrMinf_{\Emininf} \ , \quad X_{\Emininf}(x,y):= \Emininf(x)(-\Delta_{\Rtre})^{-1}(x,y)\Emininf(y).
\end{align}
We  introduce
\begin{align}
\label{eq:hinfty}
&h_{\infty}':=\inf_{f\in H_{\mathbb{R}}^1(\Rtre), \|f\|_2=1\atop f \in (\spn\{\Emininf\})^{\perp}} \expval{L_{\Emininf}}{f},\nonumber\\
&h_{\infty}'':=\inf_{f\in H_{\mathbb{R}}^1(\Rtre), \|f\|_2=1\atop f \in (\spn\{\Emininf, \partial_1 \Emininf, \partial_2 \Emininf, \partial_3 \Emininf\})^{\perp}} \expval{L_{\Emininf}-4X_{\Emininf}}{f}.
\end{align}  
We emphasize that the results contained in \cite{lenzmann2009uniqueness} imply that $\min \{h_{\infty}',h_{\infty}''\}>0$. Moreover, it is easy to see, using that $V_{\sigma_{\Emininf}}(x)\lesssim -|x|^{-1}$ for large $x$, that $L_{\Emininf}$ has infinitely many eigenvalues between $0$, its least and simple eigenvalue with eigenfunction given by $\Emininf$, and $-\LagrMinf_{\Emininf}$, the bottom of its continuous spectrum. Since furthermore $X_{\Emininf}$ is positive, this implies, in particular, that
\begin{align}
\label{eq:LowerBoundh2}
h_{\infty}'', h_{\infty}'< -\LagrMinf_{\Emininf},
\end{align}
which we shall use later. 

\begin{prop}
	\label{coercivity}
	For any $L>0$, we define
	\begin{align}
	\label{Lminusineq}
	&\quad\,h_L':=\inf_{\psi\in \MinLe}\inf_{f \in H_{\mathbb{R}}^1(\TL), \|f\|_2=1 \atop f \in(\spn\{\psi\})^{\perp}} \quad\expval{\LL_{\psi}}{f}, \\
	\label{Lplusineq}
	&h_L'':=\inf_{\psi\in\MinLe}\inf_{f \in H^1_{\mathbb{R}}(\TL), \|f\|_2=1 \atop f \in (\spn\{\psi, \partial_1 \psi, \partial_2 \psi, \partial_3 \psi\})^{\perp}} \expval{\LL_{\psi}-4\XL_{\psi}}{f}.
	\end{align}
	Then 
	\begin{align}
	\label{eq:claimliminf}
	&\liminf_{L\to \infty} h_L'\geq h_{\infty}', \quad\liminf_{L\to \infty} h_L''\geq h_{\infty}''.
	\end{align}
\end{prop}

It is not difficult to show that 
\begin{align}
\label{eq:ConvergenceToFullSpace}
&\limsup_{L\to \infty} h_L'\leq h_{\infty}', \quad\limsup_{L\to \infty} h_L''\leq h_{\infty}'',
\end{align}
simply by considering localizations of the full-space optimizers and using Proposition \ref{prop:En&MinConv}. Hence there is actually equality in \eqref{eq:claimliminf}. 

To prove Proposition \ref{coercivity} we need the following two Lemmas.

\begin{lem}
	\label{Xprop}
	For $\psi\in \MinLe$, the operator $Y^L_{\psi}$ with integral kernel $Y^L_{\psi}(x,y):= \Tker\psi(y)$ is universally bounded from $L^2(\TL)$ to $L^{\infty}(\TL)$. This in particular implies that the operators $\XL_{\psi}$, defined in \eqref{def:lpm}, are universally bounded from $L^2(\TL)$ to $L^2(\TL)$. 
\end{lem}
\begin{proof} 
	Using Lemma \ref{kernelbounds} and the normalization of $\psi$, we have
	\begin{align}
	|Y^L_{\psi}(f)(x)|&=\left|\int_{\TL} \Tker \psi(y) f(y) dy\right|\lesssim \|f\|_2+\int_{\TL} \frac {|\psi(y)f(y)|}{4\pi \dist_{\TL}(x,y)} dy\nonumber\\
	&\lesssim \|f\|_2+\int_{B_{1}(x)}\frac {|\psi(y)f(y)|}{\dist_{\TL}(x,y)} dy\leq (1+C\|\psi\|_{\infty})\|f\|_2\lesssim \|f\|_2.
	\end{align}
	To conclude, we also made use of the fact that the minimizers are universally bounded in  $L^{\infty}$ by Lemma \ref{minprop}. 
\end{proof}

Recall the definition of $\Emininf_L$ in \eqref{eq:PsiL}. 

\begin{lem}
	\label{lochess}
	For any $\varepsilon>0$, there exists $R'_{\varepsilon}$ and $L'_{\varepsilon}$ (with $R_{\varepsilon}'\leq L_{\varepsilon}'/2$) such that for any $L>L'_{\varepsilon}$, any normalized $f$ in $L^2(\TL)$ supported on $B_{R'_{\varepsilon}}^c:=[-L/2,L/2]^3\setminus B_{R'_{\varepsilon}}$, and any $\psi\in \MinLe$ such that 
	\begin{align}
	\|\psi-\Emininf_L\|_{H^1(\TL)}=\dist_{H^1}(\Theta_L(\psi),\Emininf_L) 
	\end{align}
	we have
	\begin{align}
	\expval{\LL_{\psi}-4\XL_{\psi}}{f}\geq -\LagrMinf_{\Emininf}-\varepsilon.
	\end{align} 
\end{lem}
\begin{proof}
	By definition of $\LL_{\psi}$ and $\XL_{\psi}$, we have 
	\begin{align}
	\expval{\LL_{\psi}-4\XL_{\psi}}{f}&=T_L(f)-\LagrML_{\psi} +\expval{V_{\sigma_{\psi}}}{f}-4\expval{\XL_{\psi}}{f}\nonumber\\
	&\geq -\LagrML_{\psi}+\expval{V_{\sigma_{\psi}}}{f}-4\expval{\XL_{\psi}}{f}.
	\end{align}
	By Proposition \ref{prop:En&MinConv}, taking $L_{\varepsilon}'$ sufficiently large guarantees that
	\begin{align}	
	|\LagrML_{\psi}-\LagrMinf_{\Emininf}|\leq \varepsilon/2.
	\end{align}  
	Thus we only need to show that $\expval{V_{\sigma_{\psi}}}{f}$ and $\expval{\XL_{\psi}}{f}$ can be made arbitrary small by taking $L_{\varepsilon}'$ and $R_{\varepsilon}'$ sufficiently large. Since $f$ is normalized and supported on $B_{R'_{\varepsilon}}^c$, 
	\begin{align}
	|\expval{V_{\sigma_{\psi}}}{f}|\leq \|V_{\sigma_{\psi}}\|_{L^{\infty}\left(B_{R'_{\varepsilon}}^c\right)}.
	\end{align}
	Moreover, using Lemma \ref{kernelbounds}, splitting the integral over $B_t(x)$ and $B_t^c(x)$ (for some $t>0$), and assuming $x\in B_{R'_{\varepsilon}}^c$, we find
	\begin{align}
	|V_{\sigma_{\psi}}(x)|\leq \frac C L +C\int_{\TL} \frac {|\psi(y)|^2}{\dist_{\TL}(x,y)} dy\leq \frac C L+C t \|\psi\|^2_{L^6\left(B_{R'_{\varepsilon}-t}^c\right)} +1/t.
	\end{align} 
	On the other hand, by Lemma \ref{Xprop},
	\begin{align}
	|\expval{\XL_{\psi}}{f}|\leq C\|f\|_{2}\int_{\TL} \psi(y)|f(y)|dy\leq C \|\chi_{B_{R'_{\varepsilon}}^c} \psi\|_2.
	\end{align}
	Therefore, by applying Proposition \ref{prop:En&MinConv}, we can conclude that there exists $L'_{\varepsilon}$ and $R'_{\varepsilon}$ such that, for any $L>L'_{\varepsilon}$ and any $L^2$-normalized $f$ supported on $B_{R'_{\varepsilon}}^c$, we have
	\begin{align}
	\expval{V_{\sigma_{\psi}}}{f}-4\expval{\XL_{\psi}}{f}\geq -\varepsilon/2,
	\end{align}
	which concludes our proof. 
\end{proof}

\begin{proof}[Proof of Proposition \ref{coercivity}]
	We only show the second inequality in \eqref{eq:claimliminf}, as its proof can easily be modified to also show the first. Moreover, we observe that the second inequality in \eqref{eq:claimliminf} is equivalent to the statement that for any sequence $\psi_n\in \mathcal{M}_{L_n}$ with $L_n\to \infty$,
	\begin{align}
	\liminf_{n} \inf_{f \in H^1(\mathbb{T}^3_{L_n}), \|f\|_2=1 \atop f \in \spn\{\psi_n, \partial_1 \psi_n, \partial_2 \psi_n, \partial_3 \psi_n\}^{\perp}} \expval{\LLn_{\psi_n}-4\XLn_{\psi_n}}{f}\geq h_{\infty}'',
	\end{align}
	which we shall prove in the following.
	 
	We consider $\psi_n\in \mathcal{M}_{L_n}$, $L_n\to \infty$, and define
	\begin{align}
	h_n:=\inf_{f \in H^1(\TLn), \|f\|_2=1 \atop f \in \spn\{\psi_n, \partial_1 \psi_n, \partial_2 \psi_n, \partial_3 \psi_n\}^{\perp}} \expval{\LLn_{\psi_n}-4\XLn_{\psi_n}}{f}.
	\end{align}
	By translation invariance of $\mathcal{E}_{L_n}$ and by Proposition \ref{prop:En&MinConv}, we can also restrict to sequences $\psi_n$ converging to $\Emininf$ in $L^2(\Rtre)$ and such that $\|\psi_n-\Emininf_{L_n}\|_{H^1\left(\TLn\right)}\to 0$, where $\Emininf_{L_n}$ is defined  in \eqref{eq:PsiL}.
	
	Let now $g_n$ be a normalized function in $L^2(\TLn)$, orthogonal to $\psi_n$ and its partial derivatives, realizing $h_n$ (which exists by compactness, and can be taken to be a real-valued function). We define the following partition of unity $0\leq\eta^1_R,\eta^2_R\leq1$, with $\eta^i_R\in C^{\infty}(\Rtre)$, $\eta^i_R(x)=\eta_i(x/R)$ and
	\begin{align}
	\eta_1(x)=
	\begin{cases}
	&1 \quad x\in B_1,\\
	&0 \quad x\in B_{2}^c
	\end{cases}
	\quad \quad \quad
	\eta_2=\sqrt{1-|\eta_1|^2}.
	\end{align}
	We define $\eta^i_n:=\eta^i_{L_n/8}$ and
	\begin{align}
	g_n^i:=\eta^i_ng_n/\|\eta^i_ng_n\|_2.
	\end{align}
	Standard properties of IMS localization imply that 
	\begin{align}
	h_n & = \expval{\LLn_{\psi_n}-4\XLn_{\psi_n}}{g_n}\notag\\
	&=\sum_{i=1,2} \|\eta^i_ng_n\|_2^2\expval{\LLn_{\psi_n}-4\XLn_{\psi_n}}{g_n^i} \notag\\ & \quad -\sum_{i=1,2} \Big(\expval{|\nabla \eta^i_n|^2}{g_n}+2\expval{[\eta^i_n,[\eta^i_n,\XLn_{\psi_n}]]}{g_n}\Big). \label{llo}
	\end{align}
	Clearly, the first summand in the second sum is of order $O(L_n^{-2})$, by the scaling of $\eta^i_n$. For the second summand, we observe that
	\begin{align}
	[\eta^i_n,[\eta^i_n,\XLn_{\psi_n}]](x,y)=\psi_n(x)(-\Delta_{L_n})^{-1}(x,y)\psi_n(y)\left(\eta^i_n(x)-\eta^i_n(y)\right)^2,
	\end{align}
	and proceed to bound the Hilbert-Schmidt norm of both operators ($i=1,2$), which will then bound the last line of \eqref{llo}. We make use of Lemma \ref{kernelbounds} to obtain
	\begin{align}
	&\int_{\TLn\times\TLn} |\Delta_{L_n}^{-1}(x,y)|^2 \psi_n(x)^2\psi_n(y)^2 \left(\eta^R_i(x)-\eta^i_n(y)\right)^4 dx dy\nonumber\\
	&\lesssim \frac 1 {L_n^2} +\int_{\TLn\times \TLn} \frac{ \left(\eta^i_n(x)-\eta^i_n(y)\right)^4} {d^2_{\TLn}(x,y)}\psi_n(x)^2\psi_n(y)^2 dx dy\leq \frac 1 {L_n^2}+ \|\nabla \eta^i_n\|^2_{\infty}.
	\end{align} 
	Therefore, also the second summand in the error terms is order $L_n^{-2}$, which allows us to conclude that
	\begin{align}
	\label{approxest}
	\sum_{i=1,2} \|\eta^i_ng_n\|_2^2\expval{\LLn_{\psi_n}-4\XLn_{\psi_n}}{g_n^i}=h_n+O(L_n^{-2}).
	\end{align}
	By Lemma \ref{lochess} applied to $g_n^2$ (which is supported on $B^c_{L_n/4}$) and \eqref{eq:LowerBoundh2}, we find
	\begin{align}
	\label{eq:boundong2}
	\expval{\LLn_{\psi_n}-4\XLn_{\psi_n}}{g_n^2}\geq -\LagrMinf_{\Emininf}+o_n(1)> h_{\infty}''+o_n(1).
	\end{align}
	Since the l.h.s. of \eqref{approxest} is a convex combination and $(\LLn_{\psi_n}-4\XLn_{\psi_n})$ is uniformly bounded from below, \eqref{eq:boundong2} allows to restrict to sequences $\psi_n$ such that
	\begin{align}
	\label{eq:LowerBoundNorm}
	\|\eta^1_n g_n\|_2\geq C
	\end{align} 
	uniformly in $n$ and  
	\begin{align}
	\label{g1est}
	\expval{\LLn_{\psi_n}-4\XLn_{\psi_n}}{g_n^1}\leq h_n+o_n(1),
	\end{align}
	since our claim holds on any sequence for which \eqref{eq:LowerBoundNorm} and \eqref{g1est} are not simultaneously satisfied.	
	Using \eqref{eq:LowerBoundNorm} it is easy to see that $g_n^1$ is  almost orthogonal to $\psi_n$, in the sense that 
	\begin{align}
	|\bra{g_n^1}\ket{\psi_n}|=\frac 1 {\|g_n\eta^1_n\|_2}|\bra{g_n(\eta^1_n-1)}\ket{\psi_n}|\leq \frac 1 C \|(1-\eta^1_n)\psi_n\|_2\leq\frac 1 C \|\chi_{B^c_{L_n/8}}\psi_n\|_2 \xrightarrow{n\to\infty} 0.
	\end{align}
	Here we used the $L^2$-convergence of $\psi_n$ to $\Emininf$. Clearly, the same computation (together with the $H^1$-convergence of $\psi_n$ to $\Emininf$) shows that $g_n^1$ is also almost orthogonal to the partial derivatives of $\psi_n$.
	
	To conclude, we wish to modify $g_n^1$ in order to obtain a function $\tilde{g}_n$ which satisfies the constraints (i.e., is a competitor)  of the full-space variational problem introduced in \eqref{eq:hinfty}. We also wish to have
	\begin{align}
	\label{ultimatewish}
	\expval{L_{\Emininf}-4X_{\Emininf}}{\tilde{g}_n}=\expval{\LLn_{\psi_n}-4\XLn_{\psi_n}}{g_n^1}+o_n(1).
	\end{align}
	Indeed, \eqref{ultimatewish} together with \eqref{g1est} and the fact that $\tilde{g}_n$ is a competitor on $\Rtre$, would imply that
	\begin{align}
	h_n\geq \expval{\LLn_{\psi_n}-4\XLn_{\psi_n}}{g_n^1}-o_n(1)=\expval{L_{\Emininf}-4X_{\Emininf}}{\tilde{g}_n}-o_n(1)\geq h_{\infty}''-o_n(1),
	\end{align}
	which finally yields the proof of the Proposition also for sequences $\psi_n$ satisfying \eqref{eq:LowerBoundNorm} and \eqref{g1est}.
	
	We have a natural candidate for $\tilde{g}_n$, which is simply
	\begin{align}
	\tilde{g}_n:= \frac{(\unit-\mathcal{P})g_n^1}
	{\|(\unit-\mathcal{P})g_n^1\|_2},
	\end{align}
	with $\mathcal{P}(g_n^1):=\Emininf \bra{\Emininf}\ket{g_n^1}+\sum_{i=1,2,3}\frac{\partial_i\Emininf}{\|\partial_i \Emininf\|_2} \bra{\frac{\partial_i\Emininf}{\|\partial_i \Emininf\|_2}}\ket{g_n^1}$. Clearly $\tilde{g}_n$ is a competitor for the full space minimization and we are only left with the task of proving that $\tilde{g}_n$ satisfies \eqref{ultimatewish}. 
	We observe that, since $g_n^1$ is  almost orthogonal to $\psi_n$ and its partial derivatives, and using Proposition \ref{prop:En&MinConv}, 
	\begin{align}
	\label{scalarbound}
	|\bra{\Emininf}\ket{g_n^1}|&\leq \|\Emininf-\psi_n\|_{L^2(B_{L_n/4})}+|\bra{\psi_n}\ket{g_n^1}|=o_n(1),\nonumber\\
	|\bra{\partial_i\Emininf}\ket{g_n^1}|&\leq \|\Emininf-\psi_L\|_{H^1(B_{L_n/4})}+|\bra{\partial_i\psi_n}\ket{g_n^1}|=o_n(1).
	\end{align}
	Therefore
	\begin{align}
	\|\mathcal{P}(g_n^1)\|_2\to 0 \quad \text{and} \quad \|(\unit-\mathcal{P})g_n^1\|_2\to 1.
	\end{align}
	Hence, the normalization factor does not play any role in the proof of \eqref{ultimatewish}. Moreover
	\begin{align}
	&\expval{(L_{\Emininf}-4X_{\Emininf})}{(\unit-\mathcal{P})g_n^1}\notag\\
	&=\expval{(L_{\Emininf}-4X_{\Emininf})}{g_n^1}+\expval{(L_{\Emininf}-4X_{\Emininf})}{\mathcal{P}(g_n^1)}-2\bra{g_n^1}(L_{\Emininf}-4X_{\Emininf})\ket{\mathcal{P}(g_n^1)},
	\end{align}
	and thus we can conclude that also $\mathcal{P}(g_n^1)$ does not play any role in the proof of \eqref{ultimatewish}, since $(L_{\Emininf}-4X_{\Emininf})\mathcal{P}$ is a bounded operator ($\mathcal{P}$ has finite dimensional range contained in the domain of $(L_{\Emininf}-4X_{\Emininf})$), $\mathcal{P}$ is a projection and $\|\mathcal{P}(g_n^1)\|_2\to 0$. With this discussion, we reduced our problem to showing that
	\begin{align}
	\label{veryultimatewish}
	\expval{(L_{\Emininf}-4X_{\Emininf})}{g_n^1}=\expval{(\LLn_{\psi_n}-4\XLn_{\psi_n})}{g_n^1}+o_n(1).
	\end{align}
	Clearly the kinetic energy terms coincide for every $n$ and $\LagrMLn_{\psi_n}\to \LagrMinf$, by Proposition \ref{prop:En&MinConv}. Therefore we only need to prove that
	\begin{align}
	|\expval{V_{\sigma_{\psi_n}}-V_{\sigma_{\Emininf}}}{g_n^1}|,|\expval{\XLn_{\psi_n}-X_{\Emininf}}{g_n^1}|\to 0.
	\end{align}
	For the first term, using that $g_n^1$ is supported on $B_{L_n/4}$, we have
	\begin{align}
	|\expval{V_{\sigma_{\psi_n}}-V_{\sigma_{\Emininf}}}{g_n^1}|\leq \|V_{\sigma_{\Emininf}}-V_{\sigma_{\psi_n}}\|_{L^{\infty}(B_{L_n/4})}.
	\end{align}
	If we define $\Emininf_R:=\chi_{B_R} \Emininf$ and $(\psi_n)_R:=\chi_{B_R} \psi_n$ we have $V_{\sigma_{\Emininf}}=V_{\sigma_{\Emininf_R}}+V_{\sigma_{[\Emininf-\Emininf_R]}}$ and $V_{\sigma_{\psi_n}}=V_{\sigma_{(\psi_n)_R}}+V_{\sigma_{[\psi_n-(\psi_n)_R]}}$. We consider $R=R(n)=L_n/8$ and observe that
	\begin{align}
	|V_{\sigma_{[\Emininf-\Emininf_R]}}(x)|=2\int_{\Rtre} \Gkerspace (\Emininf-\Emininf_R)^2dy\lesssim\|\Emininf-\Emininf_R\|_6^2+\|\Emininf-\Emininf_R\|_2^2\to 0.
	\end{align}
	Similar computations, together with Lemma \ref{kernelbounds}, yield similar estimates for $|V_{\sigma_{[\psi_n-(\psi_n)_R]}}(x)|$. Moreover, since $\dist_{\TLn}(x,y)=|x-y|$ for $x,y\in B_{L_n/8}$, we have, for any $x\in B_{L_n/8}$
	\begin{align}
	|(V_{\sigma_{\Emininf_R}}-V_{\sigma_{(\psi_n)_R}})(x)|&\lesssim \left|\int_{B_{L_n/4}} \frac 1 {|x-y|} (\Emininf(y)-\psi_n(y))(\Emininf(y)+\psi_n(y))dy\right| + \frac 1 {L_n}\nonumber\\
	&\lesssim  \|\Emininf+\psi_n\|_{\infty}\|\Emininf-\psi_n\|_6+\|\Emininf-\psi_n\|_2\|\Emininf+\psi_n\|_2+\frac 1 {L_n}\to 0.
	\end{align}
	Here we used again Lemma \ref{kernelbounds}, the convergence of $\psi_n$ to $\Emininf$ and the universal $L^{\infty}$-bounded\-ness of minimizers. Putting the pieces together we obtain
	\begin{align}\nonumber
	\|V_{\sigma_{\Emininf}}-V_{\sigma_{\psi_n}}\|_{L^{\infty}(B_{L_n/4})} & \leq \|V_{\sigma_{[\Emininf-\Emininf_R]}}\|_{\infty}+\|V_{\sigma_{[\psi_n-(\psi_n)_R]}}\|_{\infty} \\ & \quad +\|V_{\sigma_{\Emininf_R}}-V_{\sigma_{(\psi_n)_R}}\|_{L^{\infty}(B_{R(n)})}\to 0,
	\end{align}
	as desired. The study is similar for $\expval{\XLn_{\psi_n}-X_{\Emininf}}{g_n^1}$, hence we shall not write it down explicitly. 
	
	We  conclude that \eqref{veryultimatewish} holds and, by the discussion above, the proof is complete. 
\end{proof}

\subsubsection{Proof of Theorem \ref{uniquenessANDcoercivity}} \label{subsubsec:MainResEL} In this section we first prove universal local bounds for $\EL$ around minimizers. These are a direct consequence of the results on the Hessian in the previous subsection, the proof follows along the lines of \cite{frank2013symmetry}, \cite[Appendix A]{frank2019quantum} and \cite[Appendix A]{feliciangeli2020persistence}. Such universal local bounds yield universal local uniqueness of minimizers, i.e., the statement that minimizers that are not equivalent (i.e., not obtained one from the other by translations and changes of phase) must be universally apart (in $H^1(\TL)$). Together with Proposition \ref{prop:En&MinConv}, this clearly implies uniqueness of minimizers for $L$ big enough, which is the first part of Theorem \ref{uniquenessANDcoercivity}. A little extra effort will then complete the proof of Theorem \ref{uniquenessANDcoercivity}.

In this section, for any $\psi \in \MinLe$ and any $f\in L^2(\TL)$, we write $e^{i\theta}\psi^y=P^{L^2}_{\Theta_L(\psi)}(f)$, respectively $e^{i\theta}\psi^y=P^{H^1}_{\Theta_L(\psi)}(f)$, to mean that $e^{i\theta} \psi^y$ realizes the $L^2$-distance, respectively the $H^1$-distance, between $f$ and $\Theta_L(\psi)$. Note that by compactness these always exist, 
but they might not be unique. The possible lack of uniqueness is not a concern for our analysis, however.

\begin{prop}[Universal Local Bounds]
	\label{prop:Univlocbounds}
	There exist universal constants $K_1>0$ and $K_2>0$ and $L^{**} >0$ such that, for any $L>L^{**}$, any $\psi \in \MinLe$ and any $L^2$-normalized $f\in H^1(\TL)$ with
	\begin{align}
	\label{localrequirement}
	\dist_{H^1}(\Theta_L(\psi), f)\leq K_1,
	\end{align}  
	we have
	\begin{align}
	\EL(f)-\eL\geq K_2\|P^{L^2}_{\Theta_L(\psi)}(f)-f\|_{H^1(\TL)}\geq K_2\dist_{H^1}^2\left(\Theta_L(\psi),f\right).
	\end{align}
\end{prop}
\begin{proof} 
	We can restrict to \emph{positive} $\psi\in \MinLe$ and normalized $f$ such that
	\begin{align}
	\label{l2proj}
	P^{L^2}_{\Theta_L(\psi)}(f)=\psi,
	\end{align}
	which clearly implies  
	\begin{align}
	\label{eq:minconseq}
	\bra{\psi}\ket{f}\geq 0, \quad \bra{\Re f}\ket{\partial_i\psi}=0.
	\end{align} 
	Under this assumption, we prove that if \eqref{localrequirement} holds then
	\begin{align}
	\label{localhessbound}
	\EL(\phi)-\eL\geq K_2\|\psi-f\|_{H^1(\TL)}^2\geq K_2\dist^2_{H^1}\left(\Theta_L(\psi),f\right).
	\end{align}
	The general result follows immediately by invariance of $\EL$ under translations and changes of phase. 
	
	We denote $\delta:=f-\psi$ and proceed to expand $\EL$ around $\psi$:
	\begin{align}
	\label{hessexpansion}
	\EL(f)=\EL(\psi+\delta)=\eL+H^{\EL}_{\psi}(\delta)+\Err_{\psi}(\delta).
	\end{align}
	We recall that $H^{\EL}_{\psi}$ is simply the quadratic form associated to the Hessian of $\EL$ at $\psi$ and it is defined in \eqref{Hessian}. We denote $P_{\psi}:=\ket{\psi}\bra{\psi}$. The last term, which we see as an error contribution, is explicitly given by
	\begin{align}
	\Err_{\psi}(\delta)=&-8\bra{\Re \delta} \XL_{\psi}\ket{P_{\psi} \Re \delta}+4\expval{\XL_{\psi}}{P_{\psi}\Re \delta}\nonumber\\
	&-4\bra{|\delta|^2} (-\Delta_L)^{-1}\ket{\psi \Re \delta}+W_L(\delta).
	\end{align}
	Our first goal is to estimate $|\Err_{\psi}(\delta)|$. By \eqref{eq:minconseq} and the normalization of both $\psi$ and $f$, we find 
	\begin{align}
	\label{trick}
	\|\delta\|_2^2=2-2\bra{\psi}\ket{f}.
	\end{align}
	Therefore, also using the positivity of $\psi$, we have
	\begin{align}
	P_{\psi} \Re \delta=\psi(\bra{\psi}\ket{f}-1)=-\frac 1 2 \psi \|\delta\|_2^2.
	\end{align}
	We now apply Lemma \ref{Xprop} to obtain
	\begin{align}
	\label{eq:FirstErrEst}
	&|\bra{\Re \delta} \XL_{\psi}\ket{P_{\psi} \Re \delta}|\lesssim\|\Re \delta\|_2 \|P_{\psi} \Re \delta\|_2\lesssim\|\delta\|_2^3,\nonumber\\
	&|\expval{\XL_{\psi}}{P_{\psi}\Re \delta}|\lesssim \|P_{\psi}\Re \delta\|_2^2\lesssim \|\delta\|_2^4,\nonumber\\
	&|\bra{|\delta|^2} (-\Delta_L)^{-1}\ket{\psi \Re \delta}|\lesssim\|\delta\|_2^2\|\Re\delta\|_2\leq\|\delta\|_2^3.
	\end{align}
	Finally, by \eqref{Tbounds},
	\begin{align}
	\label{eq:WLest}
	W_L(\delta)=\|\delta\|_2^4 W_L\left(\frac{\delta}{\|\delta\|_2}\right)\leq \|\delta\|_2^4\left(\frac 1 2 T_L\left(\frac{\delta}{\|\delta\|_2}\right)+C\right)\lesssim\|\delta\|_2^2\|\delta\|_{H^1(\TL)}^2.
	\end{align}
	Recalling \eqref{localrequirement}, we can estimate 
	\begin{align}
	\|\delta\|_2=\dist_{L^2}(f,\Theta_L(\psi))\leq\dist_{H^1}(f,\Theta_L(\psi))\leq K_1,
	\end{align} 
	and this implies, combined with \eqref{eq:FirstErrEst} and \eqref{eq:WLest}, that
	\begin{align}
	\label{remainderbound}
	|\Err_{\psi}(\delta)|\lesssim \|\delta\|_{H^1(\TL)}^3.
	\end{align} 
	
	We now want to bound $H^{\EL}_{\psi}(\delta)$. We fix $0<\tau<\min\{h_{\infty}',h_{\infty}''\}$, where $h_{\infty}'$ and $h_{\infty}''$ are defined in \eqref{eq:hinfty}. Proposition \ref{coercivity} implies that there exists $L^{**}$ such that for $L>L^{**}$ and $\psi \in \MinLe$, we have
	\begin{align}
	\LL_{\psi}\geq \tau Q_{\psi}, \quad Q_{\psi}(\LL_{\psi}-4\XL_{\psi})Q_{\psi}\geq \tau Q_{\psi}',
	\end{align}  
	where we define $Q_{\psi}=\unit-P_{\psi}$ and $Q_{\psi}':=\unit-P_{\psi}-\sum_{i=1,2,3} P_{\partial_i \psi/\|\partial_i \psi\|_2}$. We note that, by \eqref{eq:minconseq} and since $\psi$ is orthogonal in $L^2$ to its partial derivatives, we have 
	\begin{equation}
	Q_{\psi}(\Re f- \psi)=Q_{\psi}'(\Re f- \psi).
	\end{equation}
	Therefore, recalling the definition of $H^{\EL}_{\psi}$ given in \eqref{Hessian},
	\begin{align}
	H^{\EL}_{\psi}(\delta)&=\expval{\LL_{\psi}}{\Im f}+\expval{Q_{\psi}(\LL_{\psi}-4\XL_{\psi} )Q_{\psi}}{\Re f- \psi}\nonumber\\
	&\geq \tau (\|Q_{\psi}\Im f\|_{2}^2+\|Q_{\psi}'(\Re f -\psi)\|_{2}^2)=\tau\|Q_{\psi}\delta\|^2_{L^2(\TL)}.
	\end{align}
	Moreover, applying \eqref{trick},
	\begin{align}
	\|Q_{\psi}\delta\|_{L^2(\TL)}^2=\|\delta\|_2^2-\bra{\psi}\ket{\delta}^2=\|\delta\|_2^2\left(1-\frac 1 4 \|\delta\|_2^2\right)\geq \frac 1 2 \|\delta\|_2^2,
	\end{align}
	and we can thus conclude that
	\begin{align}
	\label{HessL2}
	H^{\EL}_{\psi}(\delta)\geq \frac {\tau} 2 \|\delta\|_{2}^2.
	\end{align}
	On the other hand, by the universal boundedness of $V_{\sigma_{\psi}}$ in $L^{\infty}(\TL)$ and the universal boundedness of $\LagrML_{\psi}$ (see Proposition \ref{prop:En&MinConv}), we have, for some universal $C_1>0$,
	\begin{align}
	\LL_{\psi}\geq -\Delta_L - C_1.
	\end{align}
	Similarly, also using Lemma \ref{Xprop}, for some universal $C_2>0$,
	\begin{align}
	Q(\LL_{\psi}-4\XL_{\psi})Q\geq -\Delta_L- C_2.
	\end{align}
	If we then define $C:=(\max\{C_1,C_2\}+1)$, we can conclude the validity of the universal bound
	\begin{align}
	\label{HessH1}
	H^{\EL}_{\psi}(\delta)\geq \|\delta\|_{H^1(\TL)}^2 - C \|\delta\|_{L^2(\TL)}^2.
	\end{align}
	By interpolating between \eqref{HessL2} and \eqref{HessH1}, we obtain
	\begin{align}
	\label{eq:HboundwrtH1}
	H^{\EL}_{\psi}(\delta)\geq \frac {\tau}{\tau+2C} \|\delta\|_{H^1(\TL)}^2.
	\end{align} 
	Using \eqref{remainderbound} and \eqref{eq:HboundwrtH1} in \eqref{hessexpansion}, we can conclude that there exists a universal constant $C$ such that for any $L>L^{**}$, any $0<\psi \in\MinLe$ and any normalized $f$ satisfying \eqref{l2proj}, 
	\begin{align}
	\EL(f)-\eL\geq \frac 1 C\|\delta\|_{H^1(\TL)}^2-C\|\delta\|_{H^1(\TL)}^3.
	\end{align}
	In particular, for $K_2$ sufficiently small, we can find a universal constant $c$ such that \eqref{localhessbound} holds, \emph{as long as}
	\begin{align}
	\label{wronglocalrequirement}
	\|\delta\|_{H^1(\TL)}=\|P^{L^2}_{\Theta_L(\psi)}(f)-f\|_{H^1(\TL)}\leq c.
	\end{align}
	
	To conclude the proof, it only remains to show that there exists a universal $K_1$ such that \eqref{wronglocalrequirement} holds as long as \eqref{localrequirement} holds. This can be achieved as follows. We have, using that both $\psi$ and $P^{H^1}_{\Theta_L(\psi)}(f)$ are in $\MinLe$ and thus are universally bounded in $H^2(\TL)$ (by Lemma \ref{minprop}) and recalling (see \eqref{l2proj}) that $\psi=P^{L^2}_{\Theta_L(\psi)}(f)$,
	\begin{align}
	\|\psi-P^{H^1}_{\Theta_L(\psi)}(f)\|_{\mathring{H}^1(\TL)}&\leq \|\psi-P^{H^1}_{\Theta_L(\psi)}(f)\|^{1/2}_{L^2(\TL)}\|(-\Delta_L)(\psi-P^{H^1}_{\Theta_L(\psi)}(f))\|^{1/2}_{L^2(\TL)}\nonumber\\
	&\lesssim \|\psi-P^{H^1}_{\Theta_L(\psi)}(f)\|_{L^2(\TL)}^{1/2}\nonumber\\
	&\leq \left(\dist_{L^2}\left(\Theta_L(\psi),f\right)+\|f-P^{H^1}_{\Theta_L(\psi)}(f)\|_{L^2(\TL)}\right)^{1/2}\nonumber\\
	&\lesssim \dist^{1/2}_{H^1}\left(\Theta_L(\psi),f\right).
	\end{align}
	Therefore, for some universal $C$
	\begin{align}
	\|f-\psi\|_{H^1(\TL)}\leq \dist_{H^1}\left(\Theta_L(\psi),f\right)+C\dist^{1/2}_{H^1}\left(\Theta_L(\psi),f\right),
	\end{align}
	and it suffices to take $K_1\leq \left[(-C+\sqrt{C^2+4c})/2\right]^2$ to conclude our discussion. 
\end{proof}

We are ready to prove Theorem \ref{uniquenessANDcoercivity}.

\begin{proof}[Proof of Theorem \ref{uniquenessANDcoercivity}]
	Fix $K_1$ as in Proposition \ref{prop:Univlocbounds}. Using Proposition \ref{prop:En&MinConv}, we know that there exists $L_{K_1/2}$ such that, for any $L>L_{K_1/2}$ and any $\psi\in \MinLe$, we have
	\begin{align}
	\dist_{H^1}\left(\Theta_L(\psi),\Emininf_L\right)\leq K_1/2.
	\end{align}
	We claim that \eqref{bigLregime} holds with $L_1:=\max\{L_{K_1/2}, L^*, L^{**}\}$, where $L^*$ is the same as in Corollary \ref{cor:aperiodicity} and $L^{**}$ is the same as in Proposition \ref{prop:Univlocbounds}. 
	
	Let $L>L_1$ and $\psi\in \MinLe$. Since $L>L_1\geq L^*$, we have $\psi^y\neq\psi$ for any $0\neq y\in \TL$. Moreover, since $L>L_1\geq L_{K_1/2}$ and using the triangle inequality, for any other $\psi_1\in \MinLe$ we have
	\begin{align}
	\dist_{H^1}\left(\Theta_L(\psi),\psi_1\right)\leq K_1.
	\end{align}
	Since $L>L_1\geq L^{**}$, we can apply Proposition \ref{prop:Univlocbounds}, finding
	\begin{align}
	K_2\dist_{H^1}^2(\Theta_L(\psi),\psi_1)\leq\EL(\psi_1)-e_L=0,
	\end{align}
	i.e., $\psi_1\in \Theta_L(\psi)$, and \eqref{bigLregime} holds for $L>L_1$.
	
	For $\psi\in \MinLe=\Theta_L(\psi)$, and $L>L_1$, we now show the quadratic lower bound \eqref{globalquadbound}, independently of $L$. Lemma \ref{minprop}, which guarantees universal $H^1$-boundedness of minimizers, and estimate \eqref{Tbounds} ensure, by straightforward computations, that there exists $0<\kappa^*<1/2$ such that, if $f\in L^2(\TL)$ is normalized and satisfies
	\begin{align}
	\label{noquadbound}
	\EL(f)-\eL< \kappa^*\dist_{H^1}^2\left(\Theta_L(\psi),f\right),
	\end{align}
	then $f$ is universally bounded in $H^1(\TL)$ and must satisfy
	\begin{align}
	\label{noquadboundI}
	\EL(f)-\eL< \delta_{K_1},
	\end{align} 
	where $\delta_{K_1}$ is the $\delta_{\varepsilon}$ from Proposition \ref{prop:En&MinConv} with $\varepsilon=K_1$. On the other hand, Proposition \ref{prop:En&MinConv} and Proposition \ref{prop:Univlocbounds} combined with the fact that we have taken $L_1\geq L_{K_1/2}$ (and that trivially  $L_{K_1/2}\geq L_{K_1}$), guarantee that any $L^2$-normalized $f$ satisfying \eqref{noquadboundI} must satisfy
	\begin{align}
	\EL(f)-\eL\geq K_2\dist_{H^1}^2(\Theta_L(\psi),f).
	\end{align}
	Therefore the bound \eqref{globalquadbound} from Theorem \ref{uniquenessANDcoercivity} holds with the universal constant $\kappa_1:=\min\{\kappa^*,K_2\}$ and our proof is complete.
\end{proof}

This concludes our study of $\EL$. We now move on to the study of the functional $\FL$.

\subsection{Study of $\FL$} 
\label{subsec:FL} 

This section is structured as follows. In Section \ref{sec:LBforFL} we  prove Corollary \ref{cor:uniquenessANDcoercivity}. In Section \ref{sec:HFL}, we compute the Hessian of $\FL$ at its minimizers, showing the validity of \eqref{eq:Hessianexpr}. This allows to obtain a more precise lower bound for $\FL$ (compared to the bounds \eqref{eq:Fglobalquadbound} and \eqref{eq:Fglobalquadbound2} from Corollary \ref{cor:uniquenessANDcoercivity}), which holds locally around the $3$-dimensional surface of minimizers $\MinLf=\Omega_L(\varphi_L)$. Finally, in Section \ref{sec:localstudyFL}, we investigate closer the surface of minimizers $\Omega_L(\varphi_L)$ and the behavior of the functional $\FL$ close to it. In particular, we show that the Hessian of $\FL$ at its minimizers is strictly positive above its trivial zero modes and derive some key technical tools, which we exploit in Section \ref{sec:ProofMainResult}.

\subsubsection{Proof of Corollary \ref{cor:uniquenessANDcoercivity}}
\label{sec:LBforFL} 
In this section, we show the validity of Corollary \ref{cor:uniquenessANDcoercivity}. 
We need the following Lemma. Recall that in our discussion constants are universal if they are independent of $L$ for $L\geq L_0>0$.

\begin{lem}
	\label{lem:HardyLem}
	For $\psi,\phi \in H^1(\TL)$, $\|\psi\|_2=\|\phi\|_2=1$, 
	\begin{align}
	\expval{(-\Delta_L)^{-1/2}}{\rho_{\psi}-\rho_{\phi}}\lesssim \||\psi|-|\phi|\|_{H^1(\TL)}^2.
	\end{align}
\end{lem}
\begin{proof}
	We define $f(x):=|\psi(x)|+|\phi(x)|$ and $g(x):=|\psi(x)|-|\phi(x)|$. By the Hardy-Littlewood-Sobolev and the Sobolev inequality (see for example \cite{benyi2013sobolev} for a comprehensive overview of such results on the torus), and using the normalization of $\phi$ and $\psi$ we have
	\begin{align}
	\expval{(-\Delta_L)^{-1/2}}{\rho_{\psi}-\rho_{\phi}}&=\|(-\Delta_L)^{-1/4} (fg)\|_2^2\leq C \|fg\|_{3/2}^2\leq C\|f\|_2^2\|g\|_6^2\nonumber\\
	&\leq C' \|g\|_{H^1(\TL)}^2 = C'\||\psi|-|\phi|\|_{H^1(\TL)}^2,
	\end{align}
	which proves the Lemma. 
\end{proof}

\begin{proof}[Proof of Corollary \ref{cor:uniquenessANDcoercivity}]
	With $\psi_L$ as in Theorem \ref{uniquenessANDcoercivity}, let $\varphi_L:=\sigma_{\psi_L}\in C^{\infty}(\TL)$.  
	Observing that
	\begin{align}
	\label{eq:GLalternate}
	\GL(\psi,\varphi)=\EL(\psi)+\|\sigma_{\psi}-\varphi\|_2^2,
	\end{align} 
	and using Theorem \ref{uniquenessANDcoercivity} we can immediately conclude that in the regime $L>L_1$
	\begin{align}
	\MinLf=\Omega_L(\varphi_L).
	\end{align}
	It is also immediate, recalling the definition of $\GL$ in \eqref{eq:Gfun} and that $\psi_L>0$ (as proven in Theorem \ref{uniquenessANDcoercivity}), to conclude that $\psi_L$ must be the unique positive ground state of $h_{\varphi_L}$. 
	
	To prove \eqref{eq:Fglobalquadbound}, we first of all observe that if $\varphi\in L^2(\TL)$, we have 
	\begin{align}
	\FL(\varphi)=|(\varphi)_0|^2+\FL(\hat{\varphi}).
	\end{align}
	Therefore, it is sufficient to restrict to $\varphi$ with zero-average and show that in this case
	\begin{align}
	\FL(\varphi)-\eL\geq \min_{y\in \TL} \expval{\unit -(\unit+\kappa'(-\Delta_L)^{1/2})^{-1}}{\varphi-\varphi_L^y}.
	\end{align}
	Using Theorem \ref{uniquenessANDcoercivity}, we obtain
	\begin{align}
	\GL(\psi,\varphi)-\eL&=\EL(\psi)-\eL+\|\varphi-\sigma_{\psi}\|_2^2\nonumber\geq\EL(|\psi|)-\eL+\|\varphi-\sigma_{\psi}\|_2^2\nonumber\\
	&\geq \kappa_1\dist_{H^1}^2(|\psi|,\Theta(\psi_L))+\|\varphi-\sigma_{\psi}\|_2^2\nonumber\\
	&=\kappa_1\||\psi|-\psi_L^y\|_{H^1(\TL)}^2+\|\varphi-\sigma_{\psi}\|_2^2,
	\end{align}
	for some $y\in \TL$. We now apply Lemma \ref{lem:HardyLem} and use that $\varphi_L^y=\sigma_{\psi_L^y}$ (see \eqref{eq:GLalternate}), obtaining with a simple completion of the square
	\begin{align}
	\GL(\psi,\varphi)-\eL&\geq\kappa'\expval{(-\Delta_L)^{-1/2}}{\rho_{\psi}-\rho_{\psi_L^y}}+\|\varphi-\sigma_{\psi}\|_2^2\nonumber\\
	&=\|F^{1/2}(\sigma_{\psi}-\varphi_L^y)+F^{-1/2}(\varphi_L^y-\varphi)\|_2^2\nonumber\\
	&\quad +\expval{\unit-F^{-1}}{\varphi-\varphi_L^y},
	\end{align}
	where $F=\unit+\kappa'(-\Delta_L)^{1/2}$. Dropping the first term and minimizing over $\psi$ yields our claim. Finally, \eqref{eq:Fglobalquadbound2} immediately follows from \eqref{eq:Fglobalquadbound} and the spectral gap of the Laplacian, using the fact that $\varphi_L$ and all its translates have zero average since $\varphi_L=\sigma_{\psi_L}$.
\end{proof}

\subsubsection{The Hessian of $\FL$}\label{sec:HFL}

For any $\varphi \in L^2_{\R}(\TL)$, we introduce the notation 
\begin{align}
e(\varphi):=\infspec h_{\varphi},
\end{align}
and observe that $\FL$, defined in \eqref{eq:EFfun}, can  equivalently be written as
\begin{align}
\label{eq:Ffundef}
\FL(\varphi)= \|\varphi\|_2^2+e(\varphi), \quad \varphi\in L^2_{\R}(\TL).
\end{align}

We compute the Hessian of $\FL$ at its minimizers using standard arguments in perturbation theory, showing the validity of expression \eqref{eq:Hessianexpr}. We need the following two Lemmas. 

\begin{lem}
	\label{lem:L2infinbddness}
	For $L\geq L_0>0$, any $\varphi\in L^2(\TL)$ and any $T>0$
	\begin{align}
	\label{eq:L2infinbddness}
	\|(-\Delta_L+T)^{-1}\varphi\|=\|\varphi(-\Delta_L+T)^{-1}\|\leq C_T\|\varphi\|_{L^2(\TL)+L^{\infty}(\TL)}
	\end{align} 
	for some constant $C_T>0$ with $\lim_{T\to \infty} C_T = 0$. 
	Here $\varphi$ is understood as a multiplication operator, $\|\cdot\|$ denotes the operator norm on $L^2(\TL)$, and
	\begin{align}
	\|\varphi\|_{L^2(\TL)+L^{\infty}(\TL)}:=\inf_{\varphi_1+\varphi_2=\varphi \atop \varphi_1\in L^2(\TL), \, \varphi_2\in L^{\infty}(\TL)} \left(\|\varphi_1\|_{L^2(\TL)}+\|\varphi_2\|_{L^{\infty}(\TL)}\right).
	\end{align}
\end{lem}

Note that 
\begin{align}
\|\varphi\|_{L^2(\TL)}\leq L^{3/2} \|\varphi\|_{L^2(\TL)+L^{\infty}(\TL)}\leq L^{3/2} \|\varphi\|_{L^2(\TL)},
\end{align}
which clearly makes the two norms equivalent. Nevertheless, we find it more natural to work with a bound of the form \eqref{eq:L2infinbddness}, where $C_T$ is independent of $L$.

Lemma \ref{lem:L2infinbddness} implies that, for any $\varphi\in L^2(\TL)+L^{\infty}(\TL)$, the multiplication operator associated with $\varphi$ is infinitesimally relatively bounded with respect to  $-\Delta_L$. More precisely, for any $\delta>0$, there exists $C\left(\delta, \|\varphi\|_{L^2(\TL)+L^{\infty}(\TL)}\right)$ depending on $\varphi$ only  through $\|\varphi\|_{L^2(\TL)+L^{\infty}(\TL)}$, such that for any $f \in \text{Dom}(-\Delta_L)$
\begin{align}
\|\varphi f\|\leq \delta \|\Delta_L f\|+C\left(\delta, \|\varphi\|_{L^2(\TL)+L^{\infty}(\TL)}\right)\|f\|.
\end{align}
Whenever infinitesimal relative boundedness holds with a constant $C(\delta)$ uniform over a class of operators, we will say that the class is uniformly infinitesimally relatively bounded. In this case, Lemma \ref{lem:L2infinbddness} ensures that multiplication operators associated to functions in $(L^2+L^{\infty})$-balls are uniformly infinitesimally relatively bounded with respect to  $-\Delta_L$. 
\begin{proof}
	We first observe that, by self-adjointness of $(-\Delta_L + T)^{-1}$, it is sufficient to show that the claimed bound holds for $\|\varphi(-\Delta_L+T)^{-1}\|$. For any $f, \varphi\in L^2(\TL)$ and any decomposition of the form $\varphi=\varphi_1+\varphi_2$ with $\varphi_1\in L^2(\TL)$ and $\varphi_2\in L^{\infty}(\TL)$ we have
	\begin{align}
	\|\varphi (-\Delta_L+T)^{-1} f\|_2&\leq \|\varphi_1\|_2 \|(-\Delta_L+T)^{-1}f\|_{\infty}+\|\varphi_2\|_{\infty} \|(-\Delta_L+T)^{-1} f\|_2\nonumber\\
	&\leq \|\varphi_1\|_2 \|(-\Delta_L+T)^{-1}f\|_{\infty}+T^{-1}\|\varphi_2\|_{\infty}\|f\|_2.
	\end{align}
	Moreover, 
	\begin{align}
	\|(-\Delta_L+T)^{-1} f\|_{\infty}&\leq \sum_{k\in \frac {2\pi}{L} \mathbb{Z}^3} \frac 1 {L^{3/2}(|k|^2+T)} |f_k|\leq \left(\frac 1 {L^3}\sum_{k\in \frac {2\pi}{L} \mathbb{Z}^3} \frac 1 {(|k|^2+T)^2}\right)^{1/2}\|f\|_2\nonumber\\
	&\leq C \left(\int_{\Rtre} \frac 1 {(|x|^2+T)^2}\right)^{1/2}\|f\|_2 = C T^{-1/2} \|f\|_2. 
	\end{align}
	Therefore, picking $C_T:=\max\left\{T^{-1}, C T^{-1/2} \right\}$ yields
	\begin{align}
	\|\varphi (-\Delta_L+T)^{-1} f\|_2\leq C_T\left(\|\varphi_1\|_2 +\|\varphi_2\|_{\infty}\right)\|f\|_2,
	\end{align}
	optimizing over $\varphi_1$ and $\varphi_2$ completes the proof. 
\end{proof}

\begin{lem}
	\label{lem:inversehalflaplbounds}
	For  $\varphi \in L^2(\TL)$ 
	\begin{align}
	\|(-\Delta_L)^{-1/2}\varphi\|_{L^{\infty}(\TL)+L^2(\TL)}\lesssim \|(-\Delta_L+1)^{-1/2}\varphi\|_{L^2(\TL)}.
	\end{align}
\end{lem}

\begin{proof}
	We write $f_1=\chi_{[0,1)}$ and $f_2=\chi_{[1,+\infty)}$ and 
	\begin{align}
	\varphi_1= f_1\left[(-\Delta_L)^{-1/2}\right] \varphi, \quad \varphi_2=f_2\left[(-\Delta_L)^{-1/2}\right]\varphi.
	\end{align}
	Clearly $(-\Delta_L)^{-1/2} \varphi=\varphi_1+\varphi_2$. 
	Moreover
	\begin{align}
	\|(-\Delta_L)^{-1/2} \varphi\|_{L^{\infty}+L^2}&\leq \|\varphi_1\|_{\infty}+\|\varphi_2\|_2\nonumber\\
	&\leq \left(\sum_{0\neq k\in \frac {2\pi}{L} \mathbb{Z}^3\atop |k|<1} \frac 1 {L^3|k|^2}\right)^{1/2}\left(\sum_{0\neq k\in \frac {2\pi}{L} \mathbb{Z}^3\atop |k|<1} |\varphi_k|^2\right)^{1/2}+\left(\sum_{k\in \frac {2\pi}{L} \mathbb{Z}^3\atop |k|\geq 1} \frac {|\varphi_k|^2}{|k|^2}\right)^{1/2}\nonumber\\
	&\lesssim \left(\sum_{0\neq k\in \frac {2\pi}{L} \mathbb{Z}^3\atop |k|<1} |\varphi_k|^2\right)^{1/2}+\left(\sum_{k\in \frac {2\pi}{L} \mathbb{Z}^3\atop |k|\geq 1} \frac {|\varphi_k|^2}{|k|^2}\right)^{1/2}\nonumber\\
	&\lesssim \left(\sum_{k\in \frac {2\pi}{L} \mathbb{Z}^3} \frac 1 {|k|^2+1}|\varphi_k|^2\right)^{1/2}=C\|(-\Delta_L+1)^{-1/2}\varphi\|_{L^2(\TL)}.
	\end{align}
	This concludes the proof.
\end{proof}

Lemmas \ref{lem:L2infinbddness} and \ref{lem:inversehalflaplbounds} together yield the following Corollary, whose proof is omitted as it is now straightforward.

\begin{cor}
	\label{cor:Vphiinfrelbdd}
	For any $\varphi$ such that $\|(-\Delta_L+1)^{-1/2}\varphi\|_2$ is finite, the multiplication operator $V_{\varphi}$ (defined in \eqref{eq:hVphi}) is infinitesimally relatively bounded with respect to  $(-\Delta_L)$. Moreover, for $T> 0$ there exists $C_T$ such that
	\begin{align}
	\|(-\Delta_L+T)^{-1} V_{\varphi}\|\leq C_T \|(-\Delta_L+1)^{-1/2}\varphi\|_2, \quad \text{and} \;\;\; C_T\searrow 0\,\,\, \text{as} \,\,\, T\to \infty.
	\end{align}
\end{cor}

In particular, Corollary \ref{cor:Vphiinfrelbdd} implies that the family of multiplication operators associated to $\{V_{\varphi} \,|\, \|(-\Delta_L+1)^{-1/2}\varphi\|_2 \leq M\}$ is uniformly infinitesimally relatively bounded with respect to  $-\Delta_L$ for any $M$.

With these tools at hand we now investigate $\FL$ close to its minimum and, in particular, compute the Hessian of $\FL$ at its minimizers. We follow very closely the analogous analysis carried out in \cite{frank2019quantum}. By translation invariance of the problem, it is clearly sufficient to perform the computation with respect to  $\varphi_L$, where $\varphi_L$ is the same as in Corollary \ref{cor:uniquenessANDcoercivity}. 

\begin{prop}
	\label{prop:FHessian}
	For $L> L_1$ let $\varphi\in L^2_{\mathbb{R}}(\TL)$ be such that
	\begin{align}
	\label{eq:FHessianregime}
	\|(-\Delta_L+1)^{-1/2}(\varphi-\varphi_L)\|_{L^2(\TL)}\leq \varepsilon_L
	\end{align}
	for some $\varepsilon_L>0$ small enough. Then
	\begin{align}
	\label{eq:claimedHessianF}
	&\left|\FL(\varphi)-\eL-\expval{\unit-K_L}{\varphi-\varphi_L}\right|\notag\\
	&\lesssim_L \|(-\Delta_L+1)^{-1/2}(\varphi-\varphi_L)\|_{2}\expval{J_L}{\varphi-\varphi_L},
	\end{align}
	where
	\begin{align}
	\label{eq:KJdef}
	&K_L:=4(-\Delta_L)^{-1/2} \psi_L \frac {\GSOrthProj} {h_{\varphi_L}-e(\varphi_L)}\psi_L (-\Delta_L)^{-1/2},\notag\\
	&J_L=4(-\Delta_L)^{-1/2}\psi_L (-\Delta_L+1)^{-1} \psi_L (-\Delta_L)^{-1/2},
	\end{align} 
	and $\psi_L$, which we recall (see \eqref{eq:psilphil}) is the (positive) ground state of $h_{\varphi_L}$, is understood, in the expressions for $K_L$ and $J_L$, as a multiplication operator. 
\end{prop}

Note that this implies that $H^{\FL}_{\varphi_L}=\unit-K_L$, as claimed in \eqref{eq:Hessianexpr}. In particular, $K_L\leq \unit$ by minimality of $\varphi_L$. It is also clear, by definition, that $K_L\geq0$. We emphasize that $J_L$ is trace class, being the square of $2(-\Delta_L+1)^{-1/2}\psi_L(-\Delta_L)^{-1/2}$, which is Hilbert-Schmidt since $\psi_L$ is in $L^2$, as a function of $x$, and $f(k):=(|k|^2+1)^{-1/2}|k|^{-1}$ is in $L^2$, as a function of $k$. From the trace class property of $J_L$, together with the boundedness of $(-\Delta_L+1)^{1/2}\frac {\GSOrthProj} {h_{\varphi_L}-e(\varphi_L)} (-\Delta_L+1)^{1/2}$ (which follows from Corollary \ref{cor:Vphiinfrelbdd}), we immediately infer the trace class property of $K_L$. We even show in Lemma \ref{lem:KJLaplBound} that $J_L,K_L\lesssim_L (-\Delta_L+1)^{-2}$.

We shall in the following denote by $K_L^y$, respectively $J_L^y$, the unitary equivalent operators obtained from $K_L$ and $J_L$  by a translation by $y$. Note that $K_L^y$ and $J_L^y$ appear if one expands $\FL$ with respect to $\varphi_L^y$ instead of $\varphi_L$. Moreover, the invariance under translations of $\FL$ implies that 
\begin{align}
\spn\{\partial_j \varphi_L\}_{j=1}^3 \subset \ker (\unit-K_L).
\end{align}
%where we recall that $\spn\{\nabla\varphi_L\}$ denotes $\spn\{\partial_j \varphi_L\}_{j=1}^3$ in our notation. 
We show in Section \ref{sec:localstudyFL} that these two sets coincide. Finally, even though both $\varepsilon_L$ and the estimate \eqref{eq:claimedHessianF} in Proposition \ref{prop:FHessian} depend on $L$, with a little extra work one can show that the bound is actually uniform in $L$ (for large $L$). For simplicity we opt for the current version of Proposition \ref{prop:FHessian}, as it is sufficient for the purpose of our investigation, which is  set on a torus of fixed linear size $L>L_1$.

\begin{proof}
	We shall denote $h_0:=h_{\varphi_L}$. By assumption \eqref{eq:FHessianregime} and since $\varphi_L \in L^2(\TL)$, we can apply Corollary \ref{cor:Vphiinfrelbdd} to $\varphi_L$ and to $(\varphi-\varphi_L)$. This way we see that $V_{\varphi-\varphi_L}$ is uniformly infinitesimally relatively bounded with respect to $h_{0}$ for any $\varphi$ satisfying \eqref{eq:FHessianregime}.
	
	It is clear that $h_{0}$ admits a simple and isolated least eigenvalue $e(\varphi_L)$. Standard results in perturbation theory then imply that there exist $\varepsilon_L>0$ and a contour $\gamma$ around $e(\varphi_L)$ such that for any $\varphi$ satisfying \eqref{eq:FHessianregime} $e(\varphi)$ is the only eigenvalue  of $h_{\varphi}=h_0+V_{\varphi-\varphi_L}$ inside $\gamma$. (For fixed $\varphi$, the statement above is a standard result in perturbation theory, see \cite[Theorem XII.8]{reed1978iv}; moreover it is also possible to get a $\varphi$-independent $\gamma$ encircling $e(\varphi)$ (see \cite[Theorem XII.11]{reed1978iv}) since  $V_{\varphi-\varphi_L}$ is \emph{uniformly} infinitesimally relatively bounded with respect to $h_0$.)
	We can thus write
	\begin{align}
	\label{eq:perturbedev}
	e(\varphi)=\Trace \int_{\gamma}\frac z {z-(h_0+V_{\varphi-\varphi_L})} \frac {dz}{2\pi i}.
	\end{align}
	Moreover, by the uniform infinitesimal relative boundedness of $V_{\varphi-\varphi_L}$ with respect to $h_0$, we have 
	\begin{align}
	\label{eq:resolventrelbddness}
	\sup_{z\in \gamma} \|V_{\varphi-\varphi_L}(z-h_0)^{-1}\|<1,
	\end{align}
	for $\varepsilon_L$ sufficiently small. For any $z\in \gamma$, we can thus use the resolvent identity in the form
	\begin{align}
	\label{eq:resexp}
	\frac 1 {z-h_{0}-V_{\varphi-\varphi_L}}=&\left(\unit-\frac {\GSOrthProj} {z-h_0} V_{\varphi-\varphi_L}\right)^{-1} \frac {\GSOrthProj} {z-h_0}\nonumber\\
	&+\left(\unit- \frac {\GSOrthProj} {z-h_0} V_{\varphi-\varphi_L}\right)^{-1} \frac {P_{\psi_L}} {z-h_0} \left(\unit- V_{\varphi-\varphi_L} \frac 1 {z-h_0}\right)^{-1}.
	\end{align}
	The first term is analytic inside the contour $\gamma$ and hence it gives zero after integration when inserted in \eqref{eq:perturbedev}. Inserting the second term of \eqref{eq:resexp}, which is rank one, in \eqref{eq:perturbedev} and using Fubini's Theorem to interchange the trace and the integral, we obtain
	\begin{align}
	\label{eq:ebvexpansion}
	e(\varphi)=\int_{\gamma} \frac{z}{z-e(\varphi_L)} \left \langle \psi_L \left|\left(\unit- V_{\varphi-\varphi_L} \frac 1 {z-h_0}\right)^{-1}\left(\unit- \frac {\GSOrthProj} {z-h_0} V_{\varphi-\varphi_L}\right)^{-1}\right| \psi_L \right \rangle \frac {dz}{2\pi i}.
	\end{align}
	For simplicity, we introduce the notation 
	\begin{align}
	A= V_{\varphi-\varphi_L} \frac 1 {z-h_0}, \quad B= \frac {\GSOrthProj} {z-h_0} V_{\varphi-\varphi_L}.
	\end{align}
	Because of \eqref{eq:resolventrelbddness}, both $A$ and $B$ are smaller than $1$ in norm, uniformly in $z\in \gamma$. We shall use the identity 
	\begin{align}
	\frac 1 {\unit-A} \frac 1 {\unit-B}=&\unit+A+A(A+B)+\frac B {\unit-B} \notag\\
	&+ \frac {A^3}{\unit-A}+\frac{A^2}{\unit-A} B+\frac A{\unit-A} \frac {B^2}{\unit-B}.
	\end{align} 
	We insert the various terms in \eqref{eq:ebvexpansion} and do the contour integration. The term $\unit$ gives $e(\varphi_L)$. The term $A$, recalling (see \eqref{eq:psilphil}) that $(-\Delta_L)^{-1/2} \rho_{\psi_L}=\varphi_L$, yields
	\begin{align}
	\expval{V_{\varphi-\varphi_L}}{\psi_L}=-2\bra{\varphi-\varphi_L}\ket{\varphi_L}.
	\end{align}
	A standard calculation shows that the term $A(A+B)$ gives
	\begin{align}
	\left\langle \psi_L \left|V_{\varphi-\varphi_L} \frac {\GSOrthProj} {e(\varphi_L)-h_0} V_{\varphi-\varphi_L}\right|\psi_L\right\rangle=-\expval{K_L}{\varphi-\varphi_L}.
	\end{align}
	Furthermore, since $\GSOrthProj \psi_L=0$, the term $B(\unit-B)^{-1}$ yields zero. Recalling that $\FL(\varphi)=\|\varphi\|^2+e(\varphi)$ we obtain from \eqref{eq:ebvexpansion}
	\begin{align}
	&\FL(\varphi)-\FL(\varphi_L)-\expval{\unit-K_L}{\varphi-\varphi_L}\notag\\
	&=\int_{\gamma} \frac z {z-e(\varphi_L)}\left \langle \psi_L \left| \frac {A^3}{\unit-A}+A\left(\frac{A}{\unit-A} +\frac 1{\unit-A} \frac {B}{\unit-B}\right) B\right| \psi_L \right\rangle \frac {dz}{2\pi i}.
	\end{align}
	We observe that, since $\gamma$ is uniformly bounded and uniformly bounded away from $e(\varphi_L)$, we can get rid of the integration, i.e., it suffices to bound 
	\begin{align}
	&(I):=\sup_{z\in \gamma} \left| \left \langle \psi_L \left| \frac {A^3}{\unit-A}\right| \psi_L \right\rangle\right|,\notag\\
	&(II):=\sup_{z\in \gamma}\left| \left \langle \psi_L \left|A\left(\frac{A}{\unit-A} +\frac 1{\unit-A} \frac {B}{\unit-B}\right) B\right| \psi_L \right\rangle\right|,
	\end{align}
	with the r.h.s. of \eqref{eq:claimedHessianF} to conclude the proof. We note that 
	\begin{align}
	\expval{J_L}{\varphi-\varphi_L}=\left\|(-\Delta_L+1)^{1/2} V_{\varphi-\varphi_L} \psi_L\right\|_2^2,
	\end{align}
	and that, by infinitesimal relative boundedness of $V_{\varphi_L}$ with respect to $(-\Delta_L)$ and since $\gamma$ is uniformly bounded away from $e(\varphi_L)$, there exists some constant $C_L>0$ such that 
	\begin{align}
	&\sup_{z\in \gamma} \left\|(-\Delta_L+1)^{1/2} (z-h_0)^{-k}(-\Delta_L+1)^{1/2}\right\|\leq C_L \quad \text{for} \,\, k=1,2.
	\end{align}
	Therefore, 
	\begin{align}
	(I)&=\sup_{z\in \gamma}\left|(z-e(\varphi_L))^{-1}\expval{(z-h_0)^{-1}A (\unit-A)^{-1}}{V_{\varphi-\varphi_L} \psi_L}\right|\notag\\
	&\lesssim_L \sup_{z\in\gamma}\left\|(-\Delta_L+1)^{1/2} (z-h_0)^{-1} \frac A {\unit-A} (-\Delta_L+1)^{1/2}\right\|\expval{J_L}{\varphi-\varphi_L}\notag\\
	&\lesssim_L \sup_{z\in\gamma}\left\|(-\Delta_L+1)^{-1/2} \frac A {\unit-A} (-\Delta_L+1)^{1/2}\right\|\expval{J_L}{\varphi-\varphi_L},\\  
	(II)&\leq \sup_{z\in \gamma}\left\|\frac A {\unit-A} + \frac 1 {\unit-A} \frac B {\unit-B}\right\|\expval{AA^{\dagger}}{\psi_L}^{1/2} \expval{BB^{\dagger}}{\psi_L}^{1/2}\notag\\
	&\lesssim_L \sup_{z\in \gamma}\left\|\frac A {\unit-A} + \frac 1 {\unit-A} \frac B {\unit-B}\right\| \expval{J_L}{\varphi-\varphi_L}.
	\end{align}
	Since
	\begin{align}
	A(\unit-A)^{-1}=V_{\varphi-\varphi_L}(z-h_{\varphi})^{-1},
	\end{align}
	it follows that 
	\begin{align}
	&\left\|(-\Delta_L+1)^{-1/2} \frac A {\unit-A}(-\Delta_L+1)^{1/2}\right\|\notag\\	
	&\leq \|(-\Delta_L+1)^{-1/2}V_{\varphi-\varphi_L}(-\Delta_L)^{-1/2}\|\|(-\Delta_L)^{1/2}(z-h_{\varphi})^{-1}(-\Delta_L)^{1/2}\|\notag\\
	&\lesssim_L \|(-\Delta_L+1)^{-1} (\varphi-\varphi_L)\|,
	\end{align}
	where we used the relative boundedness of $h_{\varphi}$ w.r.t to $-\Delta_L$ and Corollary \ref{cor:Vphiinfrelbdd}. This yields the right bound for $(I)$. Similar estimates yield the right bounds for $\|A(\unit-A)^{-1}\|$ and $\|(\unit-A)^{-1}B(\unit-B)^{-1}\|\lesssim_L\|B\|$, concluding the proof.	
\end{proof}

As a final result of this subsection,  we prove the following Lemma about the operators $K_L$ and $J_L$. 

\begin{lem}
	\label{lem:KJLaplBound}
	Let $K_L$ and $J_L$ be the operators defined in \eqref{eq:KJdef}. We have 
	\begin{align}
	K_L, J_L \lesssim_L (-\Delta_L+1)^{-2}.
	\end{align}
\end{lem}
\begin{proof}
	We prove the result for $J_L$. By the relative boundedness of $h_{\varphi_L}$ with respect to $-\Delta_L$ the same proof applies to $K_L$. We shall show that $(-\Delta_L+1)(-\Delta_L)^{-1/2}\psi_L (-\Delta_L+1)^{-1/2}$ is bounded as an operator on $L^2(\TL)$. In fact, for $f\in L^2(\TL)$,
	\begin{align}
	&\|(-\Delta_L+1)(-\Delta_L)^{-1/2}\psi_L (-\Delta_L+1)^{-1/2} f\|_2^2\nonumber\\
	&=\sum_{0\neq k\in \frac {2\pi} L \mathbb{Z}^3} \left(\frac{|k|^2+1}{|k|}\right)^2 \left|\sum_{\xi \in \frac {2\pi} L \mathbb{Z}^3} (\psi_L)_{k-\xi} \frac{f_{\xi}}{(|\xi|^2+1)^{1/2}}\right|^2\nonumber\\
	&\leq \|(-\Delta_L+1)^{3/2}\psi_L\|_2^2\sum_{0\neq k\in \frac {2\pi} L \mathbb{Z}^3} \left(\frac{|k|^2+1}{|k|}\right)^2 \sum_{\xi \in \frac {2\pi} L \mathbb{Z}^3} \frac{|f_{\xi}|^2}{(|k-\xi|^2+1)^{3}(|\xi|^2+1)}\notag\\
	&\lesssim_L \sum_{ \xi \in \frac {2\pi} L \mathbb{Z}^3} \frac{|f_{\xi}|^2}{|\xi|^2+1} \sum_{0\neq k\in \frac {2\pi} L \mathbb{Z}^3} \frac{(|k|^2+1)^2}{|k|^2(|k-\xi|^2+1)^{3}}\lesssim_L \|f\|_2^2,
	\end{align}
	where we used that $\psi_L\in C^{\infty}(\TL)$ and that $\sum_{0\neq k\in \frac {2\pi} L \mathbb{Z}^3}\frac{(|k|^2+1)^2}{|k|^2(|k-\xi|^2+1)^{3}}\lesssim |\xi|^2+1$. Therefore
	\begin{equation}
		J_L \leq \|(-\Delta_L+1)(-\Delta_L)^{-1/2} \psi_L (-\Delta_L+1)^{-1/2}\|^2(-\Delta_L+1)^{-2}\lesssim_L (-\Delta_L+1)^{-2},
	\end{equation}	
	as claimed. 
\end{proof}

\subsubsection{Local Properties of $\MinLf$ and $\FL$}
\label{sec:localstudyFL}

For $L>L_1$ we introduce the notation
\begin{align}
\label{eq:projdef}
\GradProj:= \text{$L^2$-projection onto} \,\, \spn\{\partial_j \varphi_L\}_{j=1}^3,
\end{align}   
which is going to be used throughout this section and Section \ref{sec:ProofMainResult}.
According to Theorem~\ref{uniquenessANDcoercivity}, the condition  $L> L_1$ guarantees that $\psi_L^y\neq \psi_L$ for any $\psi_L \in \MinLe$ and any $y\neq 0$, which implies that $\ran \GradProj$ is three dimensional (i.e that the partial derivatives of $\varphi_L$ are linearly independent); if not, there would exist $\nu \in \mathbb{S}^2$ such that $\partial_{\nu} \psi_L=0$ and this would imply $\psi_L=\psi_L^y$ for any $y$ parallel to $\nu$.

For technical reasons, we also introduce a family of weighted norms which will be needed in Section \ref{sec:ProofMainResult}. For $T\geq0$, we define
\begin{align}
\label{eq:weightednorm}
\|\varphi\|_{W_T}:=\expval{W_T}{\varphi}^{1/2},
\end{align}
where $W_T$ acts in $k$-space as multiplication by 
\begin{align}
\label{eq:WTdef}
W_T(k)=
\begin{cases}
1 & |k|\leq T\\
(|k|^2+1)^{-1} & |k|> T.
\end{cases}
\end{align} 
Note that $\|\varphi\|^2_{W_0}=\expval{(-\Delta_L+1)^{-1}}{\varphi}$ and $\|\varphi\|_{W_{\infty}}=\|\varphi\|_2$. 

For the purpose of this section we could  formulate the following Lemma only with respect to $\|\cdot\|_2=\|\cdot\|_{W_{\infty}}$, but we opt for this more general version since we shall need it in Section~\ref{sec:ProofMainResult}.

\begin{lem}
	\label{lemma:WTneighbourhood}
	For any $L>L_1$, there exists $\varepsilon'_L$ (independent of $T$) such that for any $\varphi\in L^2_{\R}(\TL)$ with $\dist_{W_T}(\varphi,\Omega_L(\varphi_L))\leq \varepsilon'_L$ there exist a \emph{unique} couple $(y_{\varphi},v_{\varphi})$, depending on $T$, with $y_{\varphi}\in \TL$ and $v_{\varphi}\in (\spn_{i=1,2,3} \{W_T\partial_i \varphi_L\})^{\perp}$, such that
	\begin{align}
	\label{eq:diffreq1}
	\varphi=\varphi_L^{y_{\varphi}}+(v_{\varphi})^{y_{\varphi}} \quad \text{and} \quad \|v_{\varphi}\|_{W_T}\leq \varepsilon'_L.
	\end{align} 
\end{lem}

As Proposition \ref{prop:FHessian} above, we opt for an $L$-dependent version of Lemma \ref{lemma:WTneighbourhood} for simplicity, as it is sufficient for our purposes. We nevertheless believe it is possible to prove a corresponding statement that is uniform in $L$.  Note that Lemma \ref{lemma:WTneighbourhood} is equivalent to the statement that there exists a $T$-independent $\varepsilon'_L$ such that the $W_T$-projection onto $\Omega_L(\varphi_L)$ is uniquely defined in an $\varepsilon_L'$-neighborhood of $\Omega_L(\varphi_L)$ with respect to the $W_T$-norm, and that, for any $\varphi$ therein,  $\varphi_L^{y_{\varphi}}$ characterizes the $W_T$-projection of $\varphi$ onto $\Omega_L(\varphi_L)$, so that
\begin{align}
\dist_{W_T}(\varphi,\Omega_L(\varphi_L))=\|\varphi-\varphi_L^{y_{\varphi}}\|_{W_T}=\|v_{\varphi}\|_{W_T}.
\end{align}

\begin{proof}
	We begin by observing that the Lemma is equivalent to showing that for any $\|\cdot\|_{W_T}$-normalized ${v \in (\spn_{i=1,2,3} \{W_T\partial_i \varphi_L\})^{\perp}}$, any $\varepsilon \leq \varepsilon_L'$ and any $0\neq y\in \TL$ we have 
	\begin{align}
	\label{eq:diffreq2}
	\varepsilon <  \|\varphi_L+\varepsilon v-\varphi_L^y\|_{W_T}.
	\end{align}
	Indeed, if the Lemma holds then $\varphi=\varphi_L+\varepsilon v $ does not admit other decompositions of the form \eqref{eq:diffreq1}, which implies that, for any $y\neq 0$, \eqref{eq:diffreq2} holds (otherwise there would exist $y\neq 0$ minimizing the $W_T$-distance of $\varphi$ from $\Omega_L(\varphi_L)$ and such $y$ would necessarily yield a second decomposition of the form \eqref{eq:diffreq1}). On the other hand, if the statement \eqref{eq:diffreq2} holds and the Lemma does not, then there exists $\varphi$ such that $\dist_{W_T}(\varphi,\Omega_L(\varphi_L))\leq \varepsilon'_L$ and also such that $(y_1,v_1)$ and $(y_2,v_2)$ yield two different decompositions of the form \eqref{eq:diffreq1} for $\varphi$ (note that at least one decomposition of the form \eqref{eq:diffreq1} always exist, as there exist at least one element of $\Omega_L(\varphi_L)$ realizing the $W_T$-distance of $\varphi$ from $\Omega_L(\varphi_L)$). By considering $\varphi^{-y_1}$ (respectively $\varphi^{-y_2}$) we find $\|v_1\|_{W_T}>\|v_2\|_{W_T}$ (respectively $\|v_2\|_{W_T}>\|v_1\|_{W_T}$), which is clearly a contradiction. We shall hence proceed to prove the statement \eqref{eq:diffreq2}.
	
	Taylor's formula and the regularity of $\varphi_L$ imply the existence of a $T$-independent constant $C^1_L$  such that
	\begin{align}
	\label{eq:hessianbound}
	\varphi_L^y=\varphi_L+y\cdot (\nabla \varphi_L)+ g_y, \quad \text{with} \quad \|g_y\|_{W_T}\leq \|g_y\|_2\leq C^1_L |y|^2.
	\end{align}
	As remarked after \eqref{eq:projdef}, the kernel of $\GradProj$ is three-dimensional, hence there exists a constant $C^2_L$ independent of $T$ such that
	\begin{align}
	\label{eq:gradlb}
	\min_{\nu \in \mathbb{S}^2} \|\nu \cdot \nabla \varphi_L\|_{W_T}\geq \min_{\nu \in \mathbb{S}^2} \|\nu \cdot \nabla \varphi_L\|_{W_0}\geq C^2_L.
	\end{align}
	Therefore, using that $v\perp_{W_T} \nabla \varphi_L$ in combination with \eqref{eq:hessianbound} and \eqref{eq:gradlb}, we find, for 
	\begin{align}
	\label{eq:firstybound}
	|y|<(C^2_L-2\varepsilon C^1_L)^{1/2}(C^1_L)^{-1},
	\end{align}
	that
	\begin{align}
	\|\varphi_L+\varepsilon v-\varphi_L^y\|_{W_T}&=\|\varepsilon v-y\cdot (\nabla \varphi_L)-g_y\|_{W_T}\geq \left(\varepsilon^2+|y|^2 C_L^2\right)^{1/2}-C_L^1 |y|^2> \varepsilon,
	\end{align}
	i.e., that \eqref{eq:diffreq2} holds for $y$ satisfying \eqref{eq:firstybound}. Furthermore, we have
	\begin{align}
	\label{eq:initialest}
	\|\varphi_L+\varepsilon v-\varphi_L^y\|_{W_T}^2\geq \varepsilon^2+\|\varphi_L-\varphi_L^y\|_{W_T}\left(\|\varphi_L-\varphi_L^y\|_{W_T}-2\varepsilon\right),
	\end{align}
	and this implies that \eqref{eq:diffreq2} holds for any $y$ such that
	\begin{align}
	\label{eq:firstscan}
	\|\varphi_L-\varphi_L^y\|_{W_T}> 2 \varepsilon.
	\end{align}
	Using again \eqref{eq:gradlb} and \eqref{eq:hessianbound}, there exist $C^3_L,c^1_L,c^4_L>0$ independent of $T$ such that 
	\begin{align}
	\label{eq:localexp}
	\|\varphi_L-\varphi_L^y\|_{W_T}=\|y\cdot (\nabla \varphi_L)&+g_y\|_{W_T}\geq C^2_L|y|-C^1_L|y|^2\geq C^3_L |y|, \quad \text{for} \quad |y|\leq c^1_L,\nonumber\\
	&\|\varphi_L-\varphi_L^y\|_{W_T}>c^4_L \quad \text{for}\quad |y|>c^1_L,
	\end{align}
	where the second line simply follows from $\|\cdot\|_{W_T}\geq\|\cdot\|_{W_0}$, the fact that $\varphi_L\neq \varphi_L^y$ for any $0\neq y \in [-L/2,L/2]^3$ and the continuity of $\varphi_L$.
	Combining \eqref{eq:firstscan} and \eqref{eq:localexp}, we conclude that \eqref{eq:diffreq2} holds if either $|y|>c^1_L$ or 
	\begin{align}
	\label{eq:secondybound}
	|y|>2\varepsilon(C^3_L)^{-1}.
	\end{align}
	Picking $\varepsilon_L'$ sufficiently small, the fact that \eqref{eq:diffreq2} holds both under the conditions \eqref{eq:firstybound} and \eqref{eq:secondybound} shows that it holds for any $y\in \TL$, and  this completes the proof. 
\end{proof}

We conclude our study of the Pekar functional $\FL$ by showing that $\ker (\unit-K_L)=\spn\{\partial_j \varphi_L\}_{j=1}^3=\ran \GradProj$. Since clearly $\ran \GradProj\subset \ker(\unit-K_L)$, this is a consequence of the following Proposition. 

\begin{prop}
	\label{prop:Hessianstrictpos}
	Recalling the definition of $\tau_L$ from Corollary \ref{cor:uniquenessANDcoercivity}, we have
	\begin{align}
	\unit-K_L\geq \tau_L (\unit-\GradProj).
	\end{align}
\end{prop} 
\begin{proof}
	We need to show that for all normalized $v\in \ran  (\unit-\GradProj)$ the bound
	\begin{align}
	\expval{\unit-K_L}{v}\geq \tau_L
	\end{align}
	holds. Using Lemma \ref{lemma:WTneighbourhood} in the case $T=\infty$, for any such $v$ and $\varepsilon$ small enough, denoting $\varphi=\varphi_L+\varepsilon v$, we obtain
	\begin{align}
	\dist_{L^2}^2(\varphi,\Omega_L(\varphi_L))= \varepsilon^2.
	\end{align}
	Moreover, since $\|(-\Delta_L+1)^{-1}(\varphi-\varphi_L)\|\leq \varepsilon\|v\|_2=\varepsilon$, for $\varepsilon$ small enough we can expand $\FL(\varphi)$ with respect to $\varphi_L$ using Proposition \ref{prop:FHessian}. Combining this with \eqref{eq:Fglobalquadbound2}, we arrive at
	\begin{align}
	\tau_L \varepsilon^2\leq\FL(\varphi_L+\varepsilon v)-\eL \leq \varepsilon^2\expval{\unit-K_L}{v}+\varepsilon^3\expval{J_L}{v}.
	\end{align}
	Since $\varepsilon$ can be taken arbitrarily small, the proof is complete.
\end{proof}

\section{Proof of Main Results} \label{sec:ProofMainResult}

In this Section we give the proof  of Theorem \ref{theo:GSEexpansion}. In Section \ref{Sec:UpperBound} we prove the upper bound in \eqref{eq:infspecHsharp}. In Section \ref{Sec:LowerBoundI} we estimate the cost of substituting the full Hamiltonian $\HL$ with a cut-off Hamiltonian depending only on finitely many phonon modes, a key step in providing a lower bound for $\infspec \HL$. Finally, in Section \ref{Sec:LowerBoundII}, we show the validity of the lower bound in \eqref{eq:infspecHsharp}. 

The approach used in Sections \ref{Sec:UpperBound} and \ref{Sec:LowerBoundI} follows closely the one used in \cite{frank2019quantum}, even if, in our setting, minor complications arise in the proof of the upper bound due the presence of the zero modes of the Hessian. For the lower bound in Section \ref{Sec:LowerBoundII}, however, a substantial amount of additional work is needed to deal with the translation invariance, which complicates the proof significantly. 

\subsection{Upper Bound} \label{Sec:UpperBound} In this section we construct a trial state, which will be used to obtain an upper bound on the ground state energy of $\HL$ for fixed $L>L_1$. This is carried out using the $Q$-space representation of the bosonic Fock space $\mathcal{F}(L^2(\TL))$ (see \cite{reed1975ii}). Even though the estimates contained in this section are $L$-dependent, we believe it is possible, with little modifications to the proof, to obtain the same upper bound with the same error estimates uniformly in $L$.

Our trial state depends non-trivially only  on finitely many phonon variables,  and we proceed to describe it more in detail. If one picks $\Pi$ to be a \emph{real} finite rank projection on $L^2(\TL)$, then
\begin{align}
\mathcal{F}(L^2(\TL))\cong \mathcal{F}(\Pi L^2(\TL))\otimes \mathcal{F}((\unit-\Pi)L^2(\TL)).
\end{align}
The first factor $\mathcal{F}(\Pi L^2(\TL))$ can isomorphically be identified with $L^2(\mathbb{R}^N)$, where $N$ is the complex dimension of $\ran \Pi$. In particular, there is a one-to-one correspondence between any real $\varphi\in \ran \Pi$ and $\lambda=(\lambda_1,\dots,\lambda_N)\in \mathbb{R}^N$, explicitly given by
\begin{align}
\label{eq:rangeRniso}
\varphi=\sum_{i=1}^N \lambda_i \varphi_i\cong(\lambda_1,\dots,\lambda_N)=\lambda,
\end{align}
where $\{\varphi_i\}_{i=1}^N$ denotes an orthonormal basis of $\ran \Pi$ consisting of real-valued functions. The trial state we use corresponds to the vacuum in the second factor ${\mathcal{F}((\unit-\Pi)L^2(\TL))}$ and shall hence be written only as a function of $x$ (the electron variable) and $\lambda$ (the finitely many phonon variables selected by $\Pi$). We begin by specifying some properties we wish $\Pi$ to satisfy. Consider $\varphi_L$ from Corollary \ref{cor:uniquenessANDcoercivity} and define $\Pi$ to be a projection of the form $\Pi=\Pi'+\GradProj$, where $\GradProj$ is defined in \eqref{eq:projdef} and $\Pi'$ is an $(N-3)$-dimensional projection onto $(\spn\{\partial_j \varphi_L\}_{j=1}^3)^{\perp}=\ran (\unit-\GradProj)$ that will be further specified later but will always be taken so that $\varphi_L \in \ran \Pi$. Our trial state is of the form  
\begin{align}
\Psi(x,\varphi)=G(\varphi)\eta(\varphi) \psi_{\varphi}(x),
\end{align}
where
\begin{itemize}
	\item $x\in \TL$ and $\varphi$ is a real element of $\ran \Pi$ (identified with $\lambda\in \mathbb{R}^N$ as in \eqref{eq:rangeRniso}),
 	\item $G(\varphi)$ is a Gaussian factor explicitly given by
	\begin{align}
	G(\varphi)=\exp(-\alpha^2\expval{\left[\Pi(\unit-K_L)\Pi \right]^{1/2}}{\varphi-\varphi_L}),
	\end{align} 
	\item $\eta$ is a `localization factor' given by
	\begin{align}
	\label{eq:eta}
	\eta(\varphi)=\chi\left(\varepsilon^{-1}\|(-\Delta_L+1)^{-1/2}(\varphi-\varphi_L)\|_{L^2(\TL)}\right),
	\end{align} 
	for some $0<\varepsilon<\varepsilon_L$ (with $\varepsilon_L$  as in Proposition \ref{prop:FHessian}), where $0\leq \chi\leq 1$ is a smooth cut-off function such that $\chi(t)=1$ for $t\leq 1/2$ and $\chi(t)=0$ for $t\geq 1$,
	\item $\psi_{\varphi}$ is the unique positive ground state of $h_{\varphi}=-\Delta_L+V_{\varphi}$.
\end{itemize}
We note that our state actually depends on two parameters ($N$ and $\varepsilon$) and, of course, on the specific choice of $\Pi'$. We choose $\{\varphi_i\}_{i=1,\dots, N}$ to be a real orthonormal basis of  eigenfunctions of $[\Pi(\unit-K_L)\Pi]$ corresponding to eigenvalues $\mu_i=0$ for $i=1,2,3$ and $\mu_i\geq \tau_L>0$ for $i=4,\dots,N$. Recalling Proposition \ref{prop:Hessianstrictpos}, this amounts to choosing $\{\varphi_i\}_{i=1,2,3}$ to be a real orthonormal basis of $\ran \GradProj$ and $\{\varphi_i\}_{i=4,\dots,N}$ to be a real orthonormal basis of eigenfunctions of $[\Pi'(\unit-K_L)\Pi']$. With this choice,
we have (with a slight abuse of notation)
\begin{align}
G(\varphi)=
G(\lambda_4,\dots,\lambda_N)=\exp(-\alpha^2\sum_{i=4}^{N} \mu_i^{1/2}(\lambda_i-\lambda^L_i)^2),
\end{align}
where $\varphi_L\cong\lambda_L=(0,0,0,\lambda^L_4,\dots,\lambda_N^L)$, since $\varphi_L \in \ran \Pi$ by construction, and the first three coordinates are $0$ since $\varphi_L \in \left(\ran \GradProj\right)^{\perp}$. 

We first show that even if $G$ does not have finite $L^2(\mathbb{R}^N)$-norm, $\Psi$ does due to the presence of $\eta$. We define 
\begin{align}
\label{eq:Teps}
T_{\varepsilon}:=\{\|(-\Delta_L+1)^{-1/2}(\varphi-\varphi_L)\|\leq \varepsilon\}\subset \R^{N}
\end{align}
and
\begin{align}
\gamma_L:=\inf_{\varphi\in \ran \GradProj\atop \|\varphi\|_2=1} \expval{(-\Delta_L+1)^{-1}}{\varphi}>0.
\end{align} 
Then, on $T_{\varepsilon}$, noting that $\GradProj \varphi_L =0$, we have
\begin{align}
&\gamma_L^{1/2}\sqrt{\lambda_1^2+\lambda_2^2+\lambda_3^2}=\gamma_L^{1/2}\|\GradProj\varphi\|\leq \|(-\Delta_L+1)^{-1/2} \GradProj (\varphi-\varphi_L)\|_2\nonumber\\
&\leq \|(-\Delta_L+1)^{-1/2} \Pi' (\varphi-\varphi_L)\|_2+\varepsilon \leq \|\Pi'(\varphi-\varphi_L)\|+\varepsilon=\left(\sum_{i=4}^N (\lambda_i-\lambda_i^L)^2\right)^{1/2}+\varepsilon
\end{align}
and this implies, using the normalization of $\psi_{\varphi}$, that
\begin{align}
\|\Psi\|^2&=\int_{\mathbb{R}^N} G(\lambda_4,\dots, \lambda_N)^2 \eta(\lambda)^2 d \lambda_1 \dots d \lambda_N\leq \int_{\mathbb{R}^N} G(\lambda_4,\dots, \lambda_N)^2\unit_{T_{\varepsilon}}(\lambda) d \lambda_1 \cdots d \lambda_N\notag\\
&\leq \frac {4\pi} 3 \int_{\mathbb{R}^{N-3}} G(\lambda_4,\dots,\lambda_N)^2 \gamma_L^{-3/2}\left[\left(\sum_{i=4}^N (\lambda_i-\lambda_i^L)^2\right)^{1/2}+\varepsilon\right]^3 d\lambda_4 \cdots d\lambda_N<\infty.
\end{align}

We spend a few words to motivate our choice of $\Psi$. The absolute value squared of $\Psi$ has to be interpreted as a probability density over the couples $(\varphi,x)$, with $\varphi$ being a classical state for the phonon field and $x$ the position of the electron. In the electron coordinate, our $\Psi$ corresponds to the ground state of $h_{\varphi}$ for any value of $\varphi$. This implies, by straightforward computations, that  the expectation value of the Fr\"ohlich Hamiltonian in $\Psi$ equals the one of $e(\varphi)+\numero$, $e(\varphi)$ being the ground state energy of $h_{\varphi}$ and $\mathbb{N}$  the number operator. Moreover, because of the factor $\eta$, we are localizing our state to the regime where the Hessian expansion of $e(\varphi)$ from Proposition \ref{prop:FHessian} holds. To leading order, this effectively makes our system formally correspond to a system of infinitely many harmonic oscillators with frequencies given by the eigenvalues of $(\unit-K_L)^{1/2}$, with a Gaussian ground state. To carry out this analysis out rigorously, we need to choose a suitable finite rank projection $\Pi$, as detailed in the remainder of this section.

\medskip

We are now ready to delve into the details of the proof. It is easy to see that the interaction term appearing in the Fr\"ohlich Hamiltonian acts in the $Q$-space representation as the multiplication by $V_{\varphi}(x)$. Therefore, since $\Psi$ corresponds to the vacuum on $(\unit-\Pi)L^2(\TL)$ and only depends on $x$ through the factor $\psi_{\varphi}(x)$, the g.s. of $h_{\varphi}$, it follows that
\begin{align}
\expval{\HL}{\Psi}=\expval{e(\varphi)+\numero}{\Psi} 
\end{align}
where $\varphi = \Pi \varphi \cong \lambda\in \R^N$ and  the inner product on the r.h.s. is naturally interpreted as the one on $L^2(\TL)\otimes L^2(\mathbb{R}^{N})$. In  the $Q$-space representation, the number operator 
takes the form
\begin{align}
\numero=\sum_{n=1}^N \left(-\frac 1 {4\alpha^4} \partial_{\lambda_n}^2+\lambda_n^2-\frac 1 {2\alpha^2}\right)= \frac 1 {4\alpha^4} (-\Delta_{\lambda})+|\lambda|^2-\frac N {2\alpha^2}.
\end{align}
Using the fact that $\eta$ is supported on the set $T_{\varepsilon}$ defined in \eqref{eq:Teps}, we can use the Hessian expansion from Proposition \ref{prop:FHessian} to obtain bounds on $e(\lambda)$. Consequently, for a suitable positive constant $C_L$,
\begin{align}
\label{eq:firstevaluationtrial}
\expval{\HL}{\Psi}&\leq \expval{\eL+\expval{\unit-K_L+\varepsilon C_L J_L}{\varphi-\varphi_L}}{\Psi}\nonumber\\
&\quad +\left \langle \Psi \left|\frac 1 {4\alpha^4}(-\Delta_{\lambda})-\frac N{2\alpha^2} \right| \Psi\right \rangle\nonumber\\
& =\left(\eL-\frac 1 {2\alpha^2}\Tr(\Pi)\right) \|\Psi\|^2+A+B,
\end{align}
with
\begin{align}
&A=\left \langle \Psi \left| \frac 1 {4\alpha^4}(-\Delta_{\lambda})+\sum_{i=4}^{N} \mu_i(\lambda_i-\lambda^L_i)^2 \right| \Psi \right\rangle 
,\\
&B=\varepsilon C_L\expval{\expval{J_L}{\varphi-\varphi_L}}{\Psi}.
\end{align}
We shall now proceed to first show that $B$ only contributes as an error term and then to rewrite $A$ as the sum of a leading order energy correction term and an error term. We recall that by Lemma \ref{lem:KJLaplBound} 
\begin{align}
J_L\lesssim_L (-\Delta_L+1)^{-2}.
\end{align} 
Therefore, since $\eta$ is supported on $T_{\varepsilon}$, we have
\begin{align}
\label{eq:orderofB}
B\lesssim_L \varepsilon^3\|\Psi\|^2.
\end{align}
To treat $A$ a bit more work is required. A direct calculation shows that
\begin{align}
\left[\frac 1 {4\alpha^4}(-\Delta_{\lambda})+\sum_{i=4}^{N} \mu_i(\lambda_i-\lambda^L_i)^2\right]G=\frac 1 {2\alpha^2} \Tr(\left[\Pi(\unit-K_L)\Pi\right]^{1/2})G.
\end{align}
The previous identity, together with straightforward manipulations involving integration by parts, shows that
\begin{align}
\label{eq:boundA}
A&=\frac 1 {4\alpha^4}\left(\langle \psi_\varphi G \eta| \psi_\varphi (-\Delta_{\lambda} G)\eta \rangle+ \int_{\TL\times \mathbb{R}^{N}}  G^2  |\nabla_{\lambda} (\eta \psi_{\varphi})|^2 \right)  +\left\langle \Psi \left|\sum_{i=4}^N \mu_i(\lambda-\lambda^L_i)^2\right|\Psi\right \rangle \notag\\
&\leq\frac 1 {2\alpha^2} \Tr(\left[\Pi(\unit-K_L)\Pi\right]^{1/2})\|\Psi\|^2 \nonumber\\  & \quad +\frac 1 {2\alpha^4}\left[ \int_{\TL\times \mathbb{R}^{N}} G^2 \eta^2 |\nabla_{\lambda} \psi_{\varphi}|^2 +\int_{\TL\times \mathbb{R}^{N}} G^2 |\nabla_{\lambda}\eta|^2 |\psi_{\varphi}|^2 \right] \nonumber\\
&=:\frac 1 {2\alpha^2} \Tr(\left[\Pi(\unit-K_L)\Pi\right]^{1/2})\|\Psi\|^2+A_1+A_2,
\end{align}
where the first term is clearly a leading order energy correction whereas $A_1$ and $A_2$ have to be interpreted as error terms, as we now proceed to show. By standard first order perturbation theory (using that the phase of $\psi_{\varphi}$ is chosen so that it is the unique positive minimizer of $h_{\varphi}$) we have
\begin{align}
\partial_{\lambda_n} \psi_{\varphi}=-\frac {Q_{\psi_{\varphi}}}{h_{\varphi}-e(\varphi)} V_{\varphi_n} \psi_{\varphi},
\end{align}
where we recall that $Q_{\psi_{\varphi}}=\unit-\ket{\psi_{\varphi}}\bra{\psi_{\varphi}}$. This implies that, for fixed $\varphi$, 
\begin{align}
&\int_{\TL} |\nabla_{\lambda} \psi_{\varphi}(x)|^2 dx=\sum_{n=1}^N\left\|\frac {Q_{\psi_{\varphi}}}{h_{\varphi}-e(\varphi)} V_{\varphi_n} \psi_{\varphi}\right\|^2_{L^2(\TL)}\notag\\
&=\sum_{n=1}^N \expval{(-\Delta_L)^{-1/2} \psi_{\varphi} \left(\frac{Q_{\psi_{\varphi}}}{h_{\varphi}-e(\varphi)}\right)^2 \psi_{\varphi}(-\Delta_L)^{-1/2}}{\varphi_n}\notag\\
&=\Tr(\Pi(-\Delta_L)^{-1/2} \psi_{\varphi} \left(\frac{Q_{\psi_{\varphi}}}{h_{\varphi}-e(\varphi)}\right)^2 \psi_{\varphi}(-\Delta_L)^{-1/2}\Pi),
\end{align}
where $\psi_{\varphi}$ is interpreted as a multiplication operator in the last two expressions. Since $(-\Delta_L+1)^{1/2} \left(\frac{Q_{\psi_{\varphi}}}{h_{\varphi}-e(\varphi)}\right)^2 (-\Delta_L+1)^{1/2}$ is uniformly bounded over the support of $\eta$ (the potential $V_{\varphi}$ being uniformly infinitesimally relatively bounded with respect to $-\Delta_L$ by Corollary \ref{cor:Vphiinfrelbdd}) and recalling that $\psi_{\varphi}$ is normalized by definition, we get
\begin{align}
&\Tr(\Pi(-\Delta_L)^{-1/2} \psi_{\varphi} \left(\frac{Q_{\psi_{\varphi}}}{h_{\varphi}-e(\varphi)}\right)^2 \psi_{\varphi}(-\Delta_L)^{-1/2}\Pi)\nonumber\\
&\lesssim_L \Tr(\Pi(-\Delta_L)^{-1/2}\psi_{\varphi} (-\Delta_L+1)^{-1} \psi_{\varphi} (-\Delta_L)^{-1/2}\Pi)\lesssim_L 1.
\end{align}
In summary, we conclude that   
\begin{align}
\label{eq:A1bound}
A_1\lesssim_L \frac 1 {\alpha^4} \|\Psi\|^2.
\end{align}

Finally, we proceed to bound $A_2$. Recalling the definition of $\eta$ and $T_{\varepsilon}$, we see that
\begin{align}
\label{eq:gradetaest}
|\nabla_{\lambda} \eta|^2&=\left|\nabla_{\lambda}\left[\chi\left(\varepsilon^{-1}\|(-\Delta_L+1)^{-1/2}(\varphi-\varphi_L)\|_{L^2(\TL)}\right)\right]\right|^2 \nonumber\\
&\lesssim\varepsilon^{-2}\unit_{T_{\varepsilon}}(\varphi)\left|\nabla_{\lambda} \|(-\Delta_L+1)^{-1/2}(\varphi-\varphi_L)\|_{L^2(\TL)}\right|^2\nonumber\\
&\lesssim \varepsilon^{-2}\unit_{T_{\varepsilon}}(\varphi)\frac{\|(-\Delta_L+1)^{-1}(\varphi-\varphi_L)\|^2}{\|(-\Delta_L+1)^{-1/2}(\varphi-\varphi_L)\|^2}\leq \unit_{T_{\varepsilon}}(\varphi)\varepsilon^{-2},
\end{align}
where we used that $\eta$ is supported on $T_{\varepsilon}$ and that $\chi$ is smooth and compactly supported. Therefore, using also the normalization of $\psi_{\varphi}$, we obtain
\begin{align}
\label{eq:A2boundfirst}
A_2\lesssim \frac 1 {\alpha^4\varepsilon^2} \|\unit_{T_{\varepsilon}}G\|_{L^2(\R^N)}^2.
\end{align}
We now need to bound $\|\unit_{T_{\varepsilon}}G\|_{L^2(\R^N)}$ in terms of $\|\Psi\|= \|\eta G\|_{L^2(\R^N)}$. We define 
\begin{align}
S_{\nu}:=\{\varphi\in \ran \Pi \,|\, \|\Pi'(\varphi-\varphi_L)\|_2\leq \nu\}
\end{align}
and observe that on $S_{\nu}\cap T_{\varepsilon}$ we have, by the triangle inequality,
\begin{align}
\label{eq:insidebound}
\|(-\Delta_L+1)^{-1/2}\GradProj \varphi\|_2 \leq \varepsilon+\nu,
\end{align}
and that on $S_{\nu}^c$ 
\begin{align}
\label{eq:outsidebound}
G(\lambda)\leq \exp(-\alpha^2 \tau_L^{1/2} \nu^2),
\end{align}
where we used that $\left[\Pi(\unit-K_L)\Pi\right]^{1/2}\geq \tau_L^{1/2} \Pi'$ (with $\tau_L$ being the constant appearing in Proposition \ref{prop:Hessianstrictpos}). We then have, using \eqref{eq:insidebound}, that
\begin{align}
\|\unit_{T_{\varepsilon}}G\|_2^2&= \|\unit_{T_{\varepsilon}\cap S_{\nu}}G\|_2^2+\|\unit_{T_{\varepsilon}\cap S^c_{\nu}}G\|_2^2\notag\\
&\leq \int_{\{\|(-\Delta_L+1)^{-1/2}\GradProj \varphi\|_2 \leq \varepsilon+\nu\}\cap S_{\nu}} G^2 d\lambda_1 \dots d\lambda_N+\int_{T_{\varepsilon}\cap S^c_{\nu}} G^2 d\lambda_1 \dots d\lambda_N.
\end{align} 
We now perform the change of variables $(\lambda_1,\lambda_2,\lambda_3)=3(\lambda'_1,\lambda'_2,\lambda'_3)$ in the first integral and the change of variables $\lambda-\lambda_L=2(\lambda'-\lambda_L)$ in the second integral and fix $\nu=\varepsilon/8$, obtaining
\begin{align}
\|\unit_{T_{\varepsilon}}G\|_2^2&\leq 27 \int_{\{\|(-\Delta_L+1)^{-1/2}\GradProj \varphi\|_2 \leq (\varepsilon+\nu)/3\}\cap S_{\nu}} G^2 d\lambda+2^N\int_{T_{\varepsilon/2}\cap S_{\nu/2}^c} G(\lambda')^8 d\lambda'\notag\\
&\leq \left(27+2^N \exp(-6\alpha^2 \tau_L^{1/2} \nu^2/4)\right) \int_{T_{\varepsilon/2}} G^2 d\lambda \nonumber\\
&\leq \left(27+2^N \exp(-6\alpha^2 \tau_L^{1/2} \nu^2/4)\right)\|\Psi\|^2,
\end{align}
where in the second step we  used that $\{\|(-\Delta_L+1)^{-1/2}\GradProj \varphi\|_2 \leq (\varepsilon+\nu)/3\}\cap S_{\nu} \subset T_{\varepsilon/2}$ by the triangle inequality if $\nu=\varepsilon/8$, and \eqref{eq:outsidebound} to estimate the Gaussian factor on $S_{\nu/2}^c$. Therefore, as long as $\sqrt{N}\leq C_L^1 \alpha \varepsilon$ for a sufficiently small $C_L^1$, we  conclude that
\begin{align}
\label{eq:A2boundsecond}
A_2\lesssim \frac 1 {\alpha^4 \varepsilon^2} \|\Psi\|^2.
\end{align}
Plugging estimates \eqref{eq:orderofB}, \eqref{eq:boundA}, \eqref{eq:A1bound}, and \eqref{eq:A2boundsecond} into \eqref{eq:firstevaluationtrial}, we infer, for $\sqrt{N}\leq C_L^1 \alpha\varepsilon$, that for a sufficiently large $C_L^2$
\begin{align}
\label{eq:FinalUpEst}
&\frac{\expval{\HL}{\Psi}}{\bra{\Psi}\ket{\Psi}}\leq \eL-\frac 1 {2\alpha^2}\Tr(\Pi-\left[\Pi(\unit-K_L)\Pi\right]^{1/2})+ C_L^2(\varepsilon^3+\alpha^{-4}\varepsilon^{-2}).
\end{align}

We now proceed to choose a real orthonormal basis for $\ran \Pi$ which is convenient to bound the r.h.s. of \eqref{eq:FinalUpEst}. Let $\{g_j\}_{j\in \mathbb{N}}$ be an orthonormal basis of eigenfunctions of $K_L$ with corresponding eigenvalue $k_j$, ordered such that $k_{j+1}\geq k_j$. By Proposition \ref{prop:Hessianstrictpos} we have $k_j=1$ for $j=1,2,3$ and $k_j<1$ for $j>3$. Moreover, $\GradProj$ coincides with the projection onto $\spn\{g_1,g_2,g_3\}$. We  pick $\Pi'$ to be the projection onto $\spn\{g_4,\dots, g_{N}\}$ if $\varphi_L$ is spanned by $\{g_1,\dots, g_N\}$ and onto $\spn\{g_4,\dots, g_{N-1}, \varphi_L\}$ otherwise. With this choice the eigenvalues $\mu_i$ of $\Pi(\unit-K_L)\Pi$ appearing in the Gaussian factor $G$ are equal to
\begin{align}
\mu_j=1-k_j, \,\,\,\, j=1, \dots, N-1, \quad \mu_N=\begin{cases}
1-k_N & \text{if} \,\, \varphi_L \in \spn\{g_1,\dots,g_N\},\\
\expval{\unit-K_L}{\tilde{\varphi}_L} & \text{otherwise},
\end{cases}
\end{align}
with $\tilde{\varphi}_L:= \frac{\varphi_L-\sum_{j=4}^{N-1} g_j \bra{g_j}\ket{\varphi_L}}{\|\varphi_L-\sum_{j=4}^{N-1} g_j \bra{g_j}\ket{\varphi_L}\|_2}$.
In any case
\begin{align}
&\Tr(\Pi-\left[\Pi(\unit-K_L)\Pi\right]^{1/2})\notag\\
&\geq \sum_{j=1}^{N-1} (1-(1-k_j)^{1/2})=\Tr(\unit-(\unit-K_L)^{1/2})-\sum_{j=N}^{\infty} (1-(1-k_j)^{1/2}).
\end{align}
In order to estimate $\sum_{j=N}^{\infty} (1-(1-k_j)^{1/2})$, we note that Lemma \ref{lem:KJLaplBound} implies that $k_j\lesssim_L (l_j+1)^{-2}$, where $l_j$ denotes the ordered eigenvalues of $-\Delta_L$. Since $l_j\sim j^{2/3}$ for $j\gg1$, we have 
\begin{align}
\sum_{j=N}^{\infty} (1-(1-k_j)^{1/2})\lesssim_L N^{-1/3}.
\end{align}
This allows us to conclude that
\begin{align}
\frac{\expval{\HL}{\Psi}}{\bra{\Psi}\ket{\Psi}}\leq \eL-\frac 1 {2\alpha^2}\Tr(\unit-(\unit-K_L)^{1/2})+C_L^3(\varepsilon^3+\alpha^{-4}\varepsilon^{-2}+\alpha^{-2}N^{-1/3}),
\end{align}
as long as $\sqrt{N}\leq C_L^1 \alpha \varepsilon$. The error term is minimized, under this constraint, for $\varepsilon\sim\alpha^{-8/11}$ and $N\sim \alpha^2 \varepsilon^2 \sim\alpha^{6/11}$, which yields
\begin{align}
\frac{\expval{\HL}{\Psi}}{\bra{\Psi}\ket{\Psi}}\leq \eL-\frac 1 {2\alpha^2}\Tr(\unit-(\unit-K_L)^{1/2})+C_L \alpha^{-24/11},
\end{align}
as claimed in \eqref{eq:infspecHsharp}.

\subsection{The Cutoff Hamiltonian} \label{Sec:LowerBoundI} As a first step to derive the lower bound in \eqref{eq:infspecHsharp}, we show that it is possible to apply an ultraviolet cutoff of size $\Lambda$ to $\HL$ at an expense of order $\Lambda^{-5/2}$ (this is proven in Proposition \ref{prop:cutoffH} in Section \ref{Sec:FinalCutOff}). Our approach follows closely the one in \cite{frank2019quantum}. It relies on an application of a triple Lieb--Yamazaki bound (extending the method of \cite{lieb1958ground}) which we carry out in Section \ref{Sec:TripleLY}, and on a consequent use (in Section \ref{Sec:Gross}) of a Gross transformation \cite{gross1962particle,nelson1964interaction}. 

We shall in the following, for any \emph{real-valued} $f\in L^2(\TL)$, denote 
\begin{align}
&\Phi(f):=\ad(f)+a(f),\\ 
&\Pi(f):= \Phi(if)=i(\ad(f)-a(f)).
\end{align} 
We recall that (see \eqref{eq:FrHam}) the interaction term in the Fr\"ohlich Hamiltonian is given by 
\begin{align}
-a^{\dagger}(\elecphononcoupl)-a(\elecphononcoupl)=-\Phi(\elecphononcoupl),
\end{align} 
where $v_L$ was defined in \eqref{eq:ElPhCouplDef} and $a$ and $a^{\dagger}$ satisfy the rescaled commutation relations \eqref{eq:commrel}. We shall apply an ultraviolet cutoff of size $\Lambda$ in $k$-space, which amounts to substituting the interaction term with
\begin{align}
-a^{\dagger}(\elecphononcouplCUT)-a(\elecphononcouplCUT)=-\Phi(\elecphononcouplCUT),
\end{align}
where
\begin{align}
\label{eq:vxlambda}
v_{L,\Lambda}(y):=\sum_{0\neq k\in \frac{2\pi} L \mathbb{Z}^3\atop |k|< \Lambda} \frac 1 {|k|} \frac{e^{-i k \cdot y}}{L^3}.
\end{align}
To quantify the expense of such a cutoff we clearly need to bound
\begin{align}
-a^{\dagger}(\elecphononcouplREST)-a(\elecphononcouplREST)=-\Phi(\elecphononcouplREST),
\end{align}
where
\begin{align}
\label{eq:wx}
w_{L,\Lambda}(y)=v_{L}(y)-v_{L,\Lambda}(y)=\sum_{k\in \frac{2\pi} L \mathbb{Z}^3\atop |k|\geq \Lambda} \frac 1 {|k|} \frac{e^{-i k \cdot y}}{L^3}.
\end{align}

\subsubsection{Triple Lieb--Yamazaki Bounds} \label{Sec:TripleLY} Let us introduce the notation $p=(p_1,p_2,p_3)=-i\nabla_x$ for the electron momentum operator. Note that on any function of the form $f(x,y)=f(y-x)$, such as $\elecphononcouplREST$ for example, the operator $p$ simply acts as multiplication by $k$ in $k$-space and agrees, up to a sign, with  $-i\nabla_y$. 

The purpose of this section is to prove the following Proposition. 

\begin{prop}
	\label{prop:awxfirstbound}
	Let $w_{L,\Lambda}$ be defined as in \eqref{eq:wx} and $\Lambda>1$. Then 
	\begin{align}
	\label{eq:awxOpBoundI}
	a^{\dagger}(\elecphononcouplREST)+a(\elecphononcouplREST)=\Phi(\elecphononcouplREST)\lesssim (|p|^2+\mathbb{N}+1)^2(\Lambda^{-5/2}+\alpha^{-1} \Lambda^{-3/2}),
	\end{align}
	as quadratic forms on $L^2(\TL)\otimes \mathcal{F}(L^2(\TL))$.
\end{prop}
We first need the following Lemma. 
\begin{lem}
	\label{lem:prellybounds}
	Let $w_{L,\Lambda}$ be defined as in \eqref{eq:wx} and $\Lambda>1$. Then for any $j,l,m\in\{1,2,3\}$ 
	\begin{align}
	\label{eq:awx1}
	&a^{\dagger}\left[(\partial_j \partial_l \partial_m (-\Delta_L)^{-3}w_{L,\Lambda})^x\right]a\left[(\partial_j \partial_l \partial_m (-\Delta_L)^{-3}w_{L,\Lambda})^x\right]\lesssim \Lambda^{-5} \mathbb{N},\\
	\label{eq:awx2}
	&\|\partial_j \partial_l(-\Delta_L)^{-2} w_{L,\Lambda}\|^2_{L^2(\TL)} \lesssim \Lambda^{-3},\\
	\label{eq:awx3}
	&a^{\dagger}\left[(\partial_j \partial_l(-\Delta_L)^{-2} w_{L,\Lambda})^x\right]a\left[(\partial_j \partial_l(-\Delta_L)^{-2} w_{L,\Lambda})^x\right]\lesssim \Lambda^{-5} (|p|^2+L^{-3}\Lambda^{-1})\numero,
	\end{align}
	as quadratic forms on $L^2(\TL)\otimes \mathcal{F}(L^2(\TL))$.
\end{lem}
 
\begin{proof}
	For any $j,l,m\in\{1,2,3\}$, \eqref{eq:awx1} follows from  
	$\ad(g)a(g)\leq \|g\|_2^2 \numero$ for $g\in L^2(\TL),$  
	and then proceeding along the same lines of the proof of \eqref{eq:awx2}. To prove \eqref{eq:awx2} we  estimate
	\begin{equation}
	\|\partial_j \partial_l(-\Delta_L)^{-2} w_{L,\Lambda}\|^2_{L^2(\TL)} = \frac 1{L^3} \sum_{|k|\geq \Lambda \atop k\in \frac{2\pi} L \mathbb{Z}^3} \frac {k_j^2 k_l^2} {|k|^{10}} 
	\lesssim \int_{B_{\Lambda}^c} \frac 1 {|t|^6} dt =\frac{4\pi}{3} \Lambda^{-3}. 
	\end{equation} 
	If we denote  $f_{j,l}^x:=(-\partial_j \partial_l (-\Delta_L)^{-2} w_{L,\Lambda})^x$,
	in order to show \eqref{eq:awx3} it suffices to prove that 
	\begin{align}
	\ket{f_{j,l}^x}\bra{f^x_{j,l}} \lesssim \Lambda^{-5}\left(|p|^2+\Lambda^{-1}\right) \quad \text{on} \quad L^2(\TL)\otimes L^2(\TL),
	\end{align}
	where the bracket notation refers to the second factor in the tensor product, i.e., the left side is a rank-one projection on the second factor parametrized by $x$, which acts via multiplication on the first factor.   For any $\Psi \in L^2(\TL)\otimes L^2(\TL)$ with Fourier coefficients $\Psi_{q,k}$, we have
	\begin{align}
	\label{eq:EstimateSup}
	&\left\langle \Psi \Big|\ket{f_{j,l}^x}\bra{f^x_{j,l}} \Big| \Psi \right \rangle=\int dx \left|\int dy \overline{f_{j,l}^x(y)} \Psi(x,y) \right|^2 =\sum_{q\in \frac{2\pi} L \mathbb{Z}^3} \left| \sum_{k\in \frac{2\pi} L \mathbb{Z}^3\atop |k|\geq \Lambda} \frac{k_j k_l}{L^{3/2}|k|^5} \Psi_{q-k,k}\right|^2\notag\\
	&\leq \sum_{q \in \frac{2\pi} L \mathbb{Z}^3} \left(\sum_{k\in \frac{2\pi} L \mathbb{Z}^3\atop |k|\geq \Lambda, \,\,k\neq q} \frac{1}{L^3|k|^6|q-k|^2}\right)\left(\sum_{k\in \frac{2\pi} L \mathbb{Z}^3} |q-k|^2|\Psi_{q-k,k}|^2\right)+\sum_{q \in \frac{2\pi} L \mathbb{Z}^3\atop |q|\geq \Lambda} \frac{|\Psi_{0,q}|^2}{L^3|q|^6}\notag\\
	&\leq \sup_{q\in\frac{2\pi} L \mathbb{Z}^3}\left(\sum_{k\in \frac{2\pi} L \mathbb{Z}^3\atop |k|\geq \Lambda, \, \, k\neq q} \frac{L^{-3}}{|k|^6|q-k|^2}\right) \expval{|p|^2}{\Psi}+L^{-3} \Lambda^{-6} \| \Psi\|^2 \nonumber \\ &\lesssim \expval{\Lambda^{-5} (|p|^2+ L^{-3}\Lambda^{-1})}{\Psi},
	\end{align}
	which shows our claim. We only need to justify the last step, i.e.,  that the supremum appearing in \eqref{eq:EstimateSup} is bounded by $C\Lambda^{-5}$. We have
	\begin{align}
	\sum_{0\neq k\in \frac{2\pi} L \mathbb{Z}^3\atop |k|\geq \Lambda, \,\, k\neq q} \frac{L^{-3}}{|k|^6|q-k|^2}&\lesssim \int_{B_{\Lambda}^c} \frac 1 {|x|^6|q-x|^2} dx=\Lambda^{-5}\int_{B_1^c}\frac 1 {|x|^6|\Lambda^{-1}q-x|^2}\nonumber\\
	&\leq \Lambda^{-5} \left(\int_{B_1(\Lambda^{-1}q)}\frac 1 {|\Lambda^{-1}q-x|^2}+\int_{B_1^c} |x|^{-6}\right)\leq \frac {16\pi} 3 \Lambda^{-5}.
\end{align}
	This concludes the proof. 
\end{proof}

We are now able to prove Proposition \ref{prop:awxfirstbound}.

\begin{proof}[Proof of Proposition \ref{prop:awxfirstbound}]
	Following the approach by Lieb and Yamazaki in \cite{lieb1958ground}, we have
	\begin{align}
	\sum_{j=1}^3 [p_j,a(p_j |p|^{-2}\elecphononcouplREST)]=-a(\elecphononcouplREST).
	\end{align}
	Applying this three times, we obtain
	\begin{align}
	\sum_{j,k,l =1}^3 [p_j,[p_k,[p_l,a(p_jp_kp_l |p|^{-6}\elecphononcouplREST)]]]=-a(\elecphononcouplREST).
	\end{align}
	Similarly,
	\begin{align}
	\sum_{j,k,l =1}^3 [p_j,[p_k,[p_l,a^{\dagger}(p_jp_kp_l |p|^{-6}\elecphononcouplREST)]]]=a^{\dagger}(\elecphononcouplREST).
	\end{align}
	Therefore, if we define 
	\begin{align}
	B_{jkl}&:=a^{\dagger}(p_jp_kp_l |p|^{-6}\elecphononcouplREST)-a(p_jp_kp_l |p|^{-6}\elecphononcouplREST)\nonumber\\
	&\;=a^{\dagger}\left[(\partial_j \partial_l \partial_m (-\Delta_L)^{-3}w_{L,\Lambda})^x\right]-a\left[(\partial_j \partial_l \partial_m (-\Delta_L)^{-3}w_{L,\Lambda})^x\right],
	\end{align}
	we have
	\begin{align}
	a^{\dagger}(\elecphononcouplREST)+a(\elecphononcouplREST)=\Phi(\elecphononcouplREST)=\sum_{j,k,l =1}^3 [p_j,[p_k,[p_l,B_{jkl}]]].
	\end{align}
	Using that $B_{jkl}^{\dagger}=-B_{jkl}$ and that $B_{jkl}$ is invariant under exchange of indices, we arrive at
	\begin{align}
	\label{eq:lyexpression}
	\Phi(\elecphononcouplREST)=\sum_{j,k,l =1}^3 \left(p_j p_k [p_l,B_{jkl}]+[B_{jkl}^{\dagger},p_l]p_jp_k\right) - 2\sum_{j,k,l=1}^3 \left(p_j p_k B_{jkl} p_l +p_l B_{jkl}^{\dagger} p_j p_k \right).
	\end{align} 
	By the Cauchy--Schwarz inequality, we have for any $\lambda>0$ 
	\begin{align}
	\label{eq:boundB1}
	-p_j p_k B_{jkl} p_l -p_l B_{jkl}^{\dagger} p_j p_k \leq \lambda p_j^2 p_k^2+ \lambda^{-1} p_l B_{jkl}^{\dagger} B_{jkl} p_l.
	\end{align}
	Moreover, using \eqref{eq:awx1} and the rescaled commutation relations \eqref{eq:commrel} satisfied by $a$ and $a^{\dagger}$, we have
	\begin{align}
	\label{eq:boundB2}
	B_{jkl}^{\dagger} B_{jkl}\leq C\left(4\mathbb{N}+2\alpha^{-2}\right)\Lambda^{-5}.
	\end{align}
	Using \eqref{eq:boundB1} and \eqref{eq:boundB2}  and picking $\lambda=C^{1/2}\Lambda^{-5/2}$
	we  conclude that
	\begin{align}
	\label{eq:lzboundII}
	-2\sum_{j,k,l=1}^3 \left(p_j p_k B_{jkl} p_l +p_l B_{jkl}^{\dagger} p_j p_k \right)\lesssim \Lambda^{-5/2}\left(|p|^4+3 |p|^2(4\mathbb{N}+2\alpha^{-1})\right).
	\end{align}
	We now define 
	\begin{align}
	C_{jk}&:= \sum_{l=1}^3 [p_l,B_{jkl}]=a^{\dagger}(p_jp_k|p|^{-4}\elecphononcouplREST)+a(p_jp_k|p|^{-4}\elecphononcouplREST)\nonumber\\
	&\;=a^{\dagger}\left[(\partial_j \partial_k (-\Delta_L)^{-2}w_{L,\Lambda})^x\right]+a\left[(\partial_j \partial_k (-\Delta_L)^{-2}w_{L,\Lambda})^x\right]=C_{jk}^{\dagger}.
	\end{align}
	Using \eqref{eq:awx2}, \eqref{eq:awx3}  
	and the Cauchy-Schwarz inequality,  we have for any $\lambda>0$
	\begin{align}
	p_jp_k C_{jk}+C_{jk}p_j p_k \leq \lambda p_j^2 p_k^2 +\lambda^{-1} C_{jk}^2 .
	\end{align}
	Moreover, 
	\begin{align}
	C_{jk}^2&\leq 4 a^{\dagger}(p_jp_k |p|^{-4}\elecphononcouplREST) a(p_jp_k |p|^{-4}\elecphononcouplREST)+2\alpha^{-2} \|p_jp_k |p|^{-4} \elecphononcouplREST\|_2^2\nonumber\\
	&\lesssim \Lambda^{-5} (|p|^2+\Lambda^{-1}) \mathbb{N}+\alpha^{-2} \Lambda^{-3}.
	\end{align} 
	Picking $\lambda=\Lambda^{-5/2}+\alpha^{-1}\Lambda^{-3/2}$, we therefore conclude that
	\begin{align}\nonumber
	&\sum_{j,k,l =1}^3 \left(p_j p_k [p_l,B_{jkl}]+[B_{jkl}^{\dagger},p_l]p_jp_k\right) \\ & \lesssim (\Lambda^{-5/2}+\alpha^{-1}\Lambda^{-3/2})[|p|^4+\mathbb{N}(|p|^2+ L^{-3}\Lambda^{-1})+1].  \label{eq:lzboundI}
	\end{align}
	Applying \eqref{eq:lzboundII} and \eqref{eq:lzboundI} in \eqref{eq:lyexpression}, we finally obtain
	\begin{align}
	\Phi(\elecphononcouplREST)&\lesssim (\Lambda^{-5/2}+\alpha^{-1}\Lambda^{-3/2})\left[|p|^4+\mathbb{N}(|p|^2+ L^{-3}\Lambda^{-1})+1\right]  \notag\\ & \quad +\Lambda^{-5/2}\left(|p|^4+3 |p|^2(4\mathbb{N}+2\alpha^{-1})\right)\notag\\
	&\lesssim (|p|^2+\mathbb{N}+1)^2(\Lambda^{-5/2}+\alpha^{-1} \Lambda^{-3/2}),
	\end{align}
	as claimed.
\end{proof}

\subsubsection{Gross Transformation} \label{Sec:Gross} The bound \eqref{eq:awxOpBoundI}, derived in Proposition \ref{prop:awxfirstbound}, is not immediately useful as it stands. In order to relate the r.h.s. of \eqref{eq:awxOpBoundI} to the square of the Fr\"ohlich Hamiltonian $\HL$ in \eqref{eq:FrHam}, we shall apply a Gross transformation \cite{gross1962particle}, \cite{nelson1964interaction}. 

For a real-valued $f\in H^1(\TL)$, recalling that $f^x(\,\cdot\,)=f(\,\cdot\,-x)$, we consider the following unitary transformation on $L^2(\TL)\otimes \mathcal{F}$ 
\begin{align}\label{def:U}
U=e^{a(\alpha^2f^x)-a^{\dagger}(\alpha^2f^x)}=e^{i\Pi(\alpha^2f^x)},
\end{align}
where $U$ is understood to act as a `multiplication' with respect to the $x$ variable. For any $g\in L^2(\TL)$, we have
\begin{align}
\label{eq:Ua}
Ua(g)U^{\dagger}=a(g)+\bra{g}\ket{f^x} \quad \text{and} \quad U\ad(g)U^{\dagger}=\ad(g)+\bra{f^x}\ket{g},
\end{align}
and therefore
\begin{align}
U\numero U^{\dagger} = \numero+\Phi(f^x)+\|f\|_2^2.
\end{align}
Moreover, 
\begin{align}
UpU^{\dagger}=p+\alpha^2\Phi(p f^x)=p+\alpha^2\Phi[(i\nabla f)^x].
\end{align}
This implies that
\begin{align}
U p^2 U^{\dagger}=p^2+\alpha^4 (\Phi[(i\nabla f)^x])^2+2\alpha^2 p\cdot a[(i\nabla f)^x]+2\alpha^2 \ad[(i\nabla f)^x]\cdot p+\alpha^2 \Phi[(-\Delta_L f)^x].
\end{align}
Therefore, we also have 
\begin{align}
\label{eq:UHUd}
U \HL U^{\dagger}& = |p|^2+\alpha^4 (\Phi[(i\nabla f)^x])^2+2\alpha^2 p\cdot a[(i\nabla f)^x]+2\alpha^2 \ad[(i\nabla f)^x]\cdot p\nonumber\\
& \quad + \Phi[(-\alpha^2\Delta_L f +f-v_L)^x]+\numero +\|f\|^2_2 -2\bra{v_L}\ket{f}.
\end{align}
We denote 
\begin{align}
\label{eq:gxexpr}
g= -\alpha^2\Delta_L f +f-v_L,
\end{align}
and we shall pick 
\begin{align}
\label{eq:fxexpr}
f(y)&=\left[(-\alpha^2 \Delta_L +1)^{-1}(-\Delta_L)^{-1/2} \chi_{B_{K^2}^c}(-\Delta_L)\right](0,y)\nonumber\\
&=\sum_{|k|\geq K \atop k\in \frac {2\pi} L \mathbb{Z}^3} \frac 1 {(\alpha^2 |k|^2 +1)|k|} \frac{e^{-ik\cdot y}}{L^3}
\end{align}
for some $K>0$. Recalling \eqref{eq:vxlambda}, this implies that
\begin{align}
\label{gxexplexpr}
g(y)= -v_{L,K}(y)=-\sum_{0\neq k\in \frac{2\pi} L \mathbb{Z}^3\atop |k|< K} \frac 1 {|k|} \frac{e^{-i k \cdot y}}{L^3}.
\end{align}
For simplicity we suppress the dependence on $K$ in the notation for $f$ and $g$, but we will keep track of the parameter $K$ by denoting the operator $U$ related to this choice of $f$ (depending on $\alpha$ and $K$) via \eqref{def:U} by $U^K_{\alpha}$. We shall need the following estimates for norms involving $f$ and $g$. We have
\begin{align}
\label{eq:estgross1}
&\|g\|_2^2=\sum_{0\neq k\in \frac {2\pi}{L} \mathbb{Z}^3\atop |k|<K}\frac 1 {L^3|k|^2}\lesssim K,\\
\label{eq:estgross2}
&\|f\|_2^2=\sum_{0\neq k\in \frac {2\pi}{L} \mathbb{Z}^3\atop |k|\geq K} \frac 1 {L^3|k|^2(\alpha^2|k|^2+1)^2}\lesssim \alpha^{-4}\int_{B_K^c} \frac 1 {|t|^6} dt \lesssim \alpha^{-4} K^{-3},\\
\label{eq:estgross3}
&\bra{v_L}\ket{f}=\sum_{k\in \frac {2\pi}{L} \mathbb{Z}^3\atop |k|\geq K}\frac 1 {L^3|k|^2(\alpha^2|k|^2+1)}\lesssim \alpha^{-2} \int_{B_K^c} \frac 1 {|t|^4} dt\lesssim \alpha^{-2} K^{-1},\\
\label{eq:estgross4}
&\|\nabla f\|_2^2=\sum_{k\in \frac {2\pi}{L} \mathbb{Z}^3\atop |k|\geq K} \frac 1 {L^3(\alpha^2|k|^2+1)^2}\lesssim \alpha^{-4} \int_{B_K^c} \frac 1 {|t|^4} dt\lesssim \alpha^{-4} K^{-1}.
\end{align}
We  now state and prove the main result of this subsection, the proof of which follows the approach used in \cite{griesemer2016self} for the analogous statement on $\Rtre$, and in \cite{frank2019quantum} for the analogous statement on a domain with Dirichlet boundary conditions.

\begin{prop}
	\label{prop:Grossfinal}
	For any $\varepsilon>0$ there exist $K_{\varepsilon}>0$ and $C_{\varepsilon}>0$ such that, for all $\alpha\gtrsim 1$ and any $\Psi\in L^2(\TL)\otimes \mathcal{F}$ in the domain of $|p|^2+\mathbb{N}$
	\begin{align}
	\label{eq:Grossfinal}
	(1-\varepsilon)\|(|p|^2+\mathbb{N})\Psi\|-C_{\varepsilon}\|\Psi\|\leq\|U^{K_{\varepsilon}}_{\alpha} \HL (U^{K_{\varepsilon}}_{\alpha})^{\dagger}\Psi\|\leq(1+\varepsilon)\|(|p|^2+\mathbb{N})\Psi\|+C_{\varepsilon}\|\Psi\|.
	\end{align}
\end{prop}

\begin{proof}
	We shall use the following standard (given the rescaled commutation relations satisfied by $a$ and $\ad$) properties, which hold for any $\Psi \in \mathcal{F}$, any $f\in L^2(\TL)$ and any function $h: [0,\infty)\to \mathbb{R}$ 
	\begin{align}
	&\|a(f)\Psi\|\leq \|f\|_2\|\sqrt \numero \Psi\|, \quad \|\ad(f)\Psi\|\leq \|f\|_2\|\sqrt{\numero+\alpha^{-2}} \Psi\|,\\
	&h(\numero+\alpha^{-2})a=a h(\numero),\quad h(\numero) \ad=\ad h(\numero+\alpha^{-2}).
	\end{align}
	It is then straightforward, with the aid of the estimates \eqref{eq:estgross1}, \eqref{eq:estgross2}, \eqref{eq:estgross3} and \eqref{eq:estgross4},  to show, for any $\Psi \in L^2(\TL)\otimes \mathcal{F}$, any $\delta>0$ and any $K>0$, that 
	\begin{align}
	\label{eq:Grossest1}
	&\alpha^4\|(\Phi[(i\nabla f)^x])^2 \Psi\|\lesssim \alpha^4 \|\nabla f\|^2 \|(\numero+\alpha^{-2}) \Psi\|\lesssim K^{-1} \|(\numero+\alpha^{-2}) \Psi\|,\\
	\label{eq:Grossest2}
	& \|\Phi(g^x) \Psi\|\lesssim K^{1/2} \|\sqrt{\numero+\alpha^{-2}} \Psi\|\lesssim \delta \|(\numero+\alpha^{-2})\Psi\|+\delta^{-1} K \|\Psi\|,\\
	\label{eq:Grossest3}
	&\alpha^2\|\ad[(i\nabla f)^x]\cdot p\Psi\|\lesssim K^{-1/2}\|\sqrt{\numero+\alpha^{-2}}\sqrt{|p|^2} \Psi\|\lesssim K^{-1/2}\|(|p|^2+\numero+\alpha^2)\Psi\|.
	\end{align}
	It remains to bound the term
	\begin{align}
	\|\alpha^2 p \cdot a[(i\nabla f)^x] \Psi\|\leq \|\alpha^2 a[(i\nabla f)^x] \cdot p \Psi\|+\|a[(-\alpha^2 \Delta_L f)^x] \Psi\|=:(\mathrm{I})+(\mathrm{II}).
	\end{align}
	As in \eqref{eq:Grossest3}, we can easily bound
	\begin{align}
	\label{eq:GrosslasttermI}
	(\mathrm{I})\lesssim K^{-1/2}\|(|p|^2+\numero+\alpha^{-2})\Psi\|.
	\end{align}
	By \eqref{eq:gxexpr} and \eqref{gxexplexpr} and recalling \eqref{eq:vxlambda} and \eqref{eq:wx}, we have 
	\begin{align}
	a[(-\alpha^2 \Delta_L f)^x]=a[(g-f+v_L)^x]=-a(f^x)+a(w_{L,K}^x).
	\end{align}
	With the same arguments used in the proof of Lemma \ref{lem:prellybounds} we obtain
	\begin{align}
	\|a(w_{L,K}^x)\Psi\|\lesssim K^{-1/2}\|\sqrt{\numero(|p|^2+K^{-1})}\Psi\|,
	\end{align}
	and therefore, using \eqref{eq:estgross2} to bound $\|a(f^x) \Psi\|$ , we arrive at 
	\begin{align}
	\label{eq:GrosslasttermII}
	(\mathrm{II})&\lesssim \alpha^{-2}K^{-3/2}\|\sqrt{\numero} \Psi\|+K^{-1/2}\|\sqrt{\numero (|p|^2+K^{-1})} \Psi\|\nonumber\\
	&\lesssim \alpha^{-2}K^{-3/2}(\|(\numero+\alpha^{-2})\Psi\|+\|\Psi\|)+K^{-1/2}\|(|p|^2+\numero+K^{-1})\Psi\|.
	\end{align}
	Combining \eqref{eq:Grossest1}--\eqref{eq:Grossest3}, \eqref{eq:GrosslasttermI}, \eqref{eq:GrosslasttermII}, \eqref{eq:estgross2} and \eqref{eq:estgross3}  with \eqref{eq:UHUd}, we obtain, for any $K\geq 1$
	\begin{align}
	&\|U^K_{\alpha}\HL (U^{k}_{\alpha})^{\dagger}\Psi\|\leq [1+C(K^{-1/2}+\delta)]\|(|p|^2+\mathbb{N})\Psi\|+C(\delta^{-1} K+3\alpha^{-2}K^{-1})\|\Psi\|,\\
	&\|U^K_{\alpha}\HL (U^K_{\alpha})^{\dagger}\Psi\|\geq [1-C(K^{-1/2}+\delta)]\|(|p|^2+\mathbb{N})\Psi\|-C(\delta^{-1} K+3\alpha^{-2}K^{-1})\|\Psi\|,
	\end{align}
	which allows to conclude the proof by picking $K_{\varepsilon}\sim \varepsilon^{-2}$, $\delta \sim \varepsilon$ and $C_{\varepsilon}\sim \varepsilon^{-3}$.
\end{proof}

\begin{rem}
	\label{rem:Gross}
	Proposition \ref{prop:Grossfinal} has as an important consequence the fact that the ground state energy of $\HL$ is uniformly bounded for $\alpha \gtrsim 1$.
\end{rem}

\subsubsection{Final Estimates for Cut-off Hamiltonian} \label{Sec:FinalCutOff} With Propositions \ref{prop:awxfirstbound} and \ref{prop:Grossfinal} at hand, we are finally ready to prove the main result of this section. Note that all the estimates performed in this section are actually independent of $L$.

\begin{prop}
	\label{prop:cutoffH}
	Let
	\begin{align}
	\label{eq:Hlambda}
	\HL^{\Lambda}=-\Delta_L-\Phi(\elecphononcouplCUT)+\numero,
	\end{align}
	where $v_{L,\Lambda}$ is defined in \eqref{eq:vxlambda}. Then, for any $\Lambda\gtrsim 1$ and $\alpha \gtrsim 1$, 
	\begin{align}
	\label{eq:cutofffinal}
	\infspec{\HL}-\infspec{\HL^{\Lambda}}\gtrsim -(\Lambda^{-5/2}+\alpha^{-1}\Lambda^{-3/2}+\alpha^{-2}\Lambda^{-1}).
	\end{align}
\end{prop}

Note that for the error term introduced in \eqref{eq:cutofffinal} to be negligible compared to $\alpha^{-2}$ it suffices to pick $\Lambda \gg \alpha^{4/5}$.

\begin{proof}
	We begin by recalling that Proposition \ref{prop:awxfirstbound} implies that
	\begin{align}
	a(\elecphononcouplREST)+\ad(\elecphononcouplREST)=\Phi(\elecphononcouplREST)\lesssim (\Lambda^{-5/2}+\alpha^{-1}\Lambda^{-3/2})(|p|^2+\numero+1)^2.
	\end{align}
	Applying the unitary Gross transformation $U^K_{\alpha}$ introduced in the previous subsection (with $f$ defined in \eqref{eq:fxexpr} and $K$ large enough for Proposition \ref{prop:Grossfinal} to hold for some $0<\varepsilon<1$) to both sides of the previous inequality and recalling \eqref{eq:Ua}, we obtain
	\begin{align}
	\label{eq:phiwU}
	(U^K_{\alpha})^{\dagger}\Phi(\elecphononcouplREST) U^K_{\alpha}&=\Phi(\elecphononcouplREST)+2 \bra{f}\ket{w_{L,\Lambda}}\nonumber\\
	&\lesssim (\Lambda^{-5/2}+\alpha^{-1}\Lambda^{-3/2})(U^K_{\alpha})^{\dagger}(|p|^2+\numero+1)^2U^K_{\alpha}.
	\end{align}
 	Proposition \ref{prop:Grossfinal} implies that
	\begin{align}
	\label{eq:hc2}
	(U^K_{\alpha})^{\dagger}(|p|^2+\numero+1)^2U^K_{\alpha}\lesssim (\HL+C)^2,
	\end{align}
	where $C$ is a positive constant (independent of $\alpha$ for $\alpha \gtrsim 1$). Recalling the definitions of $f$ and $w_{L,\Lambda}$ we also have
	\begin{align}
	|\bra{f}\ket{w_{L,\Lambda}}|\leq \sum_{0\neq k\in \frac {2\pi}{L} \mathbb{Z}^3\atop |k|>\Lambda} \frac 1{L^3(\alpha^2|k|^2+1)|k|^2}\lesssim \alpha^{-2}\Lambda^{-1},
	\end{align}
	and this allows us to conclude, in combination with \eqref{eq:phiwU} and \eqref{eq:hc2}, that
	\begin{align}
	\Phi(\elecphononcouplREST)\lesssim (\Lambda^{-5/2}+\alpha^{-1}\Lambda^{-3/2}+\alpha^{-2}\Lambda^{-1})(\HL+C)^2.
	\end{align}
	Hence
	\begin{align}
	\expval{\HL}{\Psi}\geq\expval{\HL^{\Lambda}}{\Psi}-(\Lambda^{-5/2}+\alpha^{-1}\Lambda^{-3/2}+\alpha^{-2}\Lambda^{-1})\expval{(\HL+C)^2}{\Psi}.
	\end{align}
	By Remark \ref{rem:Gross}, to compute the ground state energy of $\HL$ it is clearly sufficient to restrict to the spectral subspace relative to $|\HL|\leq C$ for some suitable $C$, which then yields \eqref{eq:cutofffinal}.	This concludes the proof and the section. 
\end{proof}

\subsection{Final Lower Bound} \label{Sec:LowerBoundII} In this section we show the validity of the lower bound in \eqref{eq:infspecHsharp}, thus completing the proof of Theorem \ref{theo:GSEexpansion}. With Proposition \ref{prop:cutoffH} at hand, we have good estimates on the cost of substituting $\HL$ with $\HL^{\Lambda}$ and, in particular, we know that the difference between the ground state energies of the two is negligible for  $\Lambda\gg \alpha^{4/5}$. We are thus left with the task of giving a lower bound on $\infspec \HL^{\Lambda}$. 

While the previous steps in the lower bound follow closely the analogous strategy in \cite{frank2019quantum}, the translation invariance of our model leads to substantial complications in the subsequent steps, and the analysis given in this subsection is the main novel part of our proof. In contrast to the case considered in \cite{frank2019quantum}, the set of minimizers $\MinLf=\Omega_L(\varphi_L)$ is a three-dimensional manifold, and in order to decouple the resulting zero-modes of the Hessian of the Pekar functional we find it necessary introduce a suitable diffeomorphism that 'flattens' the manifold of minimizers and the region close to it. Special attention also has to be paid on the metric in which this closeness is measured, necessitating the introduction of the family of norms in \eqref{eq:WTdef}. 

We emphasize that the non-uniformity in $L$ also results from the subsequent analysis, where the compactness of resolvent of $-\Delta_L$ enters in an essential way. 

\bigskip

Let $\Pi$ denote the projection 
\begin{align}
\label{eq:defPI}
\ran \Pi=\spn\left\{ L^{-3/2}e^{i k\cdot x}, \,\, k \in \frac {2\pi}L \mathbb{Z}^3, \,\,|k|\leq \Lambda\right\}, \quad N=\dim_{\mathbb{C}} \ran \Pi.
\end{align}
For later use we note that 
\begin{align}
\label{eq:Nlambdaasymp}
N \sim \left(\frac L {2\pi}\right)^3 \Lambda^3 \quad \text{as} \,\, \Lambda \to \infty.
\end{align}
The Fock space $\mathcal{F}(L^2(\TL))$ naturally factorizes into the tensor product $\mathcal{F}(\Pi L^2(\TL))\otimes \mathcal{F}((\unit-\Pi) L^2(\TL))$ and $\HL^{\Lambda}$ is of the form $\mathbb{A}\otimes \unit + \unit \otimes \mathbb{N}^>$, with $\mathbb{A}$ acting on $L^2(\TL)\otimes \mathcal{F}(\Pi L^2(\TL))$ and $\mathbb{N}^>$ being the number operator on $\mathcal{F}((\unit-\Pi) L^2(\TL))$. In particular, $\infspec \HL^{\Lambda}=\infspec \mathbb{A}$.

As in Section \ref{Sec:UpperBound}, we can, for any $L^2$-orthonormal basis of real-valued functions $\{f_n\}$ of $\ran \Pi$, identify $\mathcal{F}(\Pi L^2(\TL))$ with $L^2(\mathbb{R}^N)$ through the $Q$-space representation (see \cite{reed1975ii}). In particular, any real-valued $\varphi \in \ran \Pi$ corresponds to a point $\lambda \in \R^N$ via
\begin{align}
\label{eq:identification}
\varphi=\Pi \varphi= \sum_{n=1}^N \lambda_n f_n\cong(\lambda_1,\dots,\lambda_N)=\lambda.
\end{align}
Note that, compared to Section \ref{Sec:UpperBound}, we are using a different choice of $\Pi$ here for the decomposition $L^2(\TL)=\ran \Pi \oplus (\ran \Pi)^{\perp}$.

In the representation given by \eqref{eq:identification}, the operator $\mathbb{A}$ is given by 
\begin{align}
\mathbb{A}=-\Delta_L+V_{\varphi}(x)+\sum_{n=1}^N \left(-\frac 1 {4\alpha^4} \partial_{\lambda_n}^2+\lambda_n^2-\frac 1 {2\alpha^2}\right)
\end{align}
on $L^2(\TL)\otimes L^2(\mathbb{R}^N)$. For a lower bound, we can replace $h_{\varphi}=-\Delta_L+V_{\varphi}$ with the infimum of its spectrum $e(\varphi)$, obtaining
\begin{align}
\infspec \HL^{\Lambda}\geq \infspec \mathbb{K},
\end{align}
where $\mathbb{K}$ is the operator on $L^2(\mathbb{R}^N)$ defined as
\begin{align}
\label{eq:defK}
\mathbb{K}=-\frac 1 {4\alpha^4} \sum_{n=1}^N \partial_{\lambda_n}^2- \frac N {2 \alpha^2}+\FL(\varphi)=\frac 1 {4\alpha^4} (-\Delta_{\lambda})- \frac N {2 \alpha^2}+\FL(\lambda),
\end{align}
where $\FL$, which is understood as a multiplication operator in \eqref{eq:defK}, can be seen as a function of $\varphi\in \spn_{\R}\{f_j\}_{j=1}^N$ or $\lambda\in \R^N$ through the identification \eqref{eq:identification}. 

Using IMS localization we shall split $\mathbb{R}^N$ into two regions, one localized around the surface of minimizers of $\FL$, i.e., $\MinLf=\Omega_L(\varphi_L)$, and the other localized away from it. On each of these regions we can bound $\FL$ from below with the estimates contained in Proposition \ref{prop:FHessian} and in Corollary \ref{cor:uniquenessANDcoercivity}, respectively. Because of the prefactor $\alpha^{-4}$ in front of $-\Delta_{\lambda}$  the outer region turns out to be negligible compared to the inner one (at least if we define the inner and outer region with respect to  an appropriate norm). At the same time, employing an appropriate diffeomorphism,  the inner region can be treated as if $\Omega_L(\varphi_L)$ was a a flat torus, leading to a system of harmonic oscillators whose ground state energy can be calculated explicitly. 

We start by specifying the norm with respect to  which we measure closeness to $\Omega_L(\varphi_L)$. Recall the definition of the $W_T$-norms given in \eqref{eq:WTdef}. Note that for $T\geq \Lambda$ the $L^2$-norm coincides with the $W_T$-norm on $\ran \Pi$, which makes $0<T< \Lambda$ the relevant regime for our discussion. In fact, we shall pick 
\begin{align}
\label{eq:TlambdaRegime}
1\ll T \ll \Lambda^{2/3} \ , \quad  \alpha^{4/5}\ll\Lambda,
\end{align}
where $T\gg 1$ is needed for the inner region to yield  the right contribution, and $T\ll \Lambda^{2/3}$ ensures that the outer region contribution is negligible.

We proceed by introducing an IMS type localization with respect to  $\|\cdot\|_{W_T}$. Let $\chi:\mathbb{R}_+\to [0,1]$ be a smooth function such that $\chi(t)=1$ for $t\leq 1/2$ and $\chi(t)=0$ for $t\geq 1$. Let $\varepsilon>0$ and let $j_1$ and $j_2$ denote the multiplication operators on $L^2(\mathbb{R}^N)$ 
\begin{align}
j_1=\chi\left(\varepsilon^{-1} \text{dist}_{W_T}(\varphi,\Omega_L(\varphi_L))\right), \quad j_2=\sqrt{1-j_1^2}.
\end{align}
Then
\begin{align}\label{r1}
\mathbb{K}=j_1 \mathbb{K} j_1+j_2\mathbb{K} j_2-\mathbb{E},
\end{align}
where $\mathbb{E}$ is the IMS localization error given by
\begin{align}
\mathbb{E}=\frac 1 {4\alpha^4} \sum_{n=1}^N \left(|\partial_{\lambda_n} j_1|^2+|\partial_{\lambda_n} j_2|^2\right),
\end{align}
which is estimated in the following lemma. 

\begin{lem}
	\begin{align}\label{r2}
	\mathbb{E}\lesssim \alpha^{-4} \varepsilon^{-2} 
	\end{align}
\end{lem}
\begin{proof}
To bound $\mathbb{E}$ we apply Lemma \ref{lemma:WTneighbourhood}, which states that for $\varepsilon$ sufficiently small, 
 for any $\varphi\in \text{supp} j_1$, there exists a \emph{unique} $y_{\varphi}\in \TL$ such that
\begin{equation}
\dist_{W_T}^2(\varphi,\Omega_L(\varphi_L))=\expval{W_T}{\varphi-\varphi_L^{y_{\varphi}}}.
\end{equation}
Likewise, for any $n\in \{1,\dots,N\}$ and any $h$ sufficiently small there exists a unique $y_{n,h}\in \TL$ such that
\begin{equation}
\dist_{W_T}^2(\varphi+hf_n,\Omega_L(\varphi_L))=\expval{W_T}{\varphi+hf_n-\varphi_L^{y_{n,h}}}.
\end{equation} 
It is easy to see, using again Lemma \ref{lemma:WTneighbourhood}, that $\lim_{h\to 0}y^{h,n}= y^{\varphi}$ for any $n$. Therefore, using that $\dist_{W_T}(\varphi+hf_n,\Omega_L(\varphi_L))\leq \|\varphi-\varphi_L^{y_{\varphi}}\|_{W_T}$ and $\dist_{W_T}(\varphi,\Omega_L(\varphi_L))\leq \|\varphi-\varphi_L^{y_{h,n}}\|_{W_T}$, we arrive at
\begin{align}
& 2\bra{f_n}W_T\ket{\varphi-\varphi_L^{y_{\varphi}}}=\lim_{h\to 0}2\bra{f_n}W_T\ket{\varphi-\varphi_L^{y_{h,n}}}\notag\\
&\leq\lim_{h\to 0} h^{-1}\left(\dist_{W_T}^2(\varphi+hf_n,\Omega_L(\varphi_L))-\dist_{W_T}^2(\varphi,\Omega_L(\varphi_L))\right)\notag \\
&\leq 2\bra{f_n}W_T\ket{\varphi-\varphi_L^{y_{\varphi}}},
\end{align}
which shows that
\begin{align}
\partial_{\lambda_n} \dist^2_{W_T}(\varphi, \Omega_L(\varphi_L))=2\bra{f_n}W_T(\varphi-\varphi_L^{y_{\varphi}})\rangle.
\end{align}
Using that $|\chi'|, \left|\left[(1-\chi^2)^{1/2}\right]'\right|\lesssim \unit_{[1/2,1]}$, for $k=1,2$ we obtain
\begin{align}\nonumber
\left|\left[\partial_{\lambda_n} j_k\right](\varphi)\right|^2&\lesssim \varepsilon^{-4}\left|\partial_{\lambda_n} \dist^2_{W_T}(\varphi, \Omega_L(\varphi_L))\right|^2\unit_{\left\{\dist_{W_T}(\varphi, \Omega_L(\varphi_L))\leq \varepsilon\right\}}&\\
&\lesssim \varepsilon^{-4} |\bra{f_n}W_T(\varphi-\varphi_L^{y_{\varphi}})\rangle|^2\unit_{\left\{\dist_{W_T}(\varphi, \Omega_L(\varphi_L))\leq \varepsilon\right\}}.
\end{align}
Summing over $n$, using that $\|W_T\|\leq 1$ and that $\{f_n\}$ is an orthonormal system, we arrive at \eqref{r2}.
\end{proof}

Thus, the localization error  is negligible as long as $\varepsilon \gg \alpha^{-1}$. Hence, we are left with the task of providing lower bounds for $j_1 \mathbb{K} j_1$ and $j_2 \mathbb{K} j_2$ under the constraint $\varepsilon \gg \alpha^{-1}$. We carry out these estimates in the next two subsections, \ref{Sec:IMSinner} and \ref{Sec:IMSouter}. Finally, in Section \ref{Sec:FinalLowerBound}, we combine these bounds to prove the lower bound in \eqref{eq:infspecHsharp}.

\subsubsection{Bounds on $j_1 \mathbb{K} j_1$} \label{Sec:IMSinner}  Let us look closer at the intersection of the $\varepsilon$-neighborhood of $\Omega_L(\varphi_L)$ with respect to  the $W_T$-norm with $\ran \Pi$, i.e., the set 
\begin{align}
\Tneigh:=\{\varphi\in \ran \Pi \,|\, \bar{\varphi}=\varphi, \,\, \text{dist}_{W_T}(\varphi,\Omega_L(\varphi_L))\leq \varepsilon\}=\text{supp} j_1 \cap \ran \Pi .
\end{align}
In the following we shall show that this set is, for $\varepsilon$ small enough,  a tubular neighborhood of $\Pi\Omega_L(\varphi_L)$, which can be mapped via a suitable  diffeomorphism (given in Definition \ref{def:GroosCoord})  to a tubular neighborhood of a flat torus. 

Since $\varphi \in \ran \Pi$ and $\Pi$ commutes both with $W_T$ and with the transformation $g\mapsto g^y$ for any $y\in \TL$, we have
\begin{align}
\dist^2_{W_T}(\varphi,\Omega_L(\varphi_L))=\|(\unit-\Pi)\varphi_L\|^2_{W_T}+\dist^2_{W_T}(\varphi,\Omega_L(\Pi\varphi_L)).
\end{align}
This implies that $\Tneigh$ is non-empty if and only if 
\begin{align}
r_{T,\varepsilon}:=\sqrt{\varepsilon^2-\|(\unit-\Pi)\varphi_L\|^2_{W_T}}>0.
\end{align}
Since $\varphi_L\in C^{\infty}(\TL)$, $r_{T,\varepsilon}>0$ as long as 
\begin{align}
\label{eq:restr1}
\varepsilon \gtrsim_L \Lambda^{-h}
\end{align} 
for some $h>0$ and $\Lambda$ sufficiently large. In particular, \eqref{eq:restr1} is satisfied with $h=5/4$ for $\alpha$ large enough since, as discussed above, we need to pick $\varepsilon\gg \alpha^{-1}$ and $\Lambda \gg \alpha^{4/5}$ for the IMS and the cutoff errors to be negligible.

Lemma \ref{lemma:WTneighbourhood} implies that any $\varphi \in \Tneigh$, for $\varepsilon\leq \varepsilon'_L$ (independently of $T$ and $N$), admits a unique $W_T$-projection $\varphi_L^{y_{\varphi}}$ onto $\Omega_L(\varphi_L)$ and 
\begin{align}
\label{eq:Decomp}
\varphi=\varphi_L^{y_{\varphi}}+(v_{\varphi})^{y_{\varphi}}, \quad \text{with} \quad v_{\varphi} \in (\spn \{\Pi W_T \partial_j \varphi_L\}_{j=1}^3)^{\perp_{L^2}}.
\end{align}
%where, again, we use $\spn \{\Pi W_T \nabla \varphi_L\}$ to denote $\spn \{\Pi W_T \partial_j \varphi_L\}_{j=1}^3$. 
Since $W_T$ and $\Pi$ commute,  $\Omega_L(\varphi_L)$ is `parallel' to $\ran \Pi$ with respect to  $\|\cdot\|_{W_T}$, i.e., $\dist_{W_T}(\ran \Pi, \varphi_L^y)$ is independent of $y$ and the $W_T$-projection of $\varphi_L^y$ onto $\Pi$ is simply $\Pi (\varphi_L^y)=(\Pi \varphi_L)^y$. Therefore, for $\varepsilon\leq \varepsilon'_L$, any $\varphi\in \Tneigh$ admits a \emph{unique} $W_T$-projection $(\Pi \varphi_L)^{y_\varphi}$ onto $\Omega_L(\Pi\varphi_L)$ and \eqref{eq:Decomp} induces a unique decomposition of the form
\begin{align}
\label{eq:ProjDecomp}
\varphi=(\Pi \varphi_L)^{y_{\varphi}}+(\eta_{\varphi})^{y_{\varphi}},\quad \text{with} \quad \eta_{\varphi}\in (\spn \{\Pi W_T \partial_j \varphi_L\}_{j=1}^3)^{\perp_{L^2}}, \,\, \|\eta_{\varphi}\|_{W_T}\leq r_{T,\varepsilon},
\end{align}
where $\eta_{\varphi}=\Pi v_{\varphi}$ (note that $(\unit-\Pi)v_{\varphi}=-(\unit-\Pi)\varphi_L$). 
This allows to introduce the following diffeomorphism,  which is a central object in our discussion. It maps $\Tneigh$ onto a tubular neighborhood of a flat torus. We shall call this diffeomorphism \emph{Gross coordinates}, as it is inspired by an approach introduced in \cite{gross1976strong}.

\begin{defin}[Gross coordinates]
	\label{def:GroosCoord}
	For
	\begin{align}
	\label{eq:BTLedef}
	B^{T,\Lambda}_{\varepsilon}:=\left\{\eta\in (\spn \{\Pi W_T \partial_j \varphi_L\}_{j=1}^3)^{\perp_{L^2}}\cap \ran \Pi\,\,|\,\,\|\eta\|_{W_T}\leq r_{T,\varepsilon}\right\}\subset \ran \Pi,
	\end{align} 
 	we define the Gross coordinates map $u$ as 
	\begin{align}
	\label{eq:uinversedef}
	u:\Tneigh& \rightarrow \TL \times B^{T,\Lambda}_{\varepsilon},\nonumber\\
	\varphi& \mapsto (y_{\varphi},\eta_{\varphi}),
	\end{align}
	where $y_{\varphi}$ and $\eta_{\varphi}$ are defined through the decomposition \eqref{eq:ProjDecomp}.
\end{defin}

By the discussion above it is clear that $u$ is well-defined and invertible, for $\varepsilon\leq \varepsilon'_L$ (defined in Lemma~\ref{lemma:WTneighbourhood}), with inverse $u^{-1}$ given by  
\begin{align}
\label{eq:diffeoimpl}
u^{-1}: \TL \times B^{T,\Lambda}_{\varepsilon} &\rightarrow \Tneigh \nonumber\\
(y,\eta)&\mapsto (\Pi\varphi_L)^y+\eta^y.
\end{align}
We emphasize that the whole aim of the discussion above is to show that $u$ is well-defined, 
since once that has been shown the invertibility of $u$ and the form of $u^{-1}$ are obvious. In other words, the map $u^{-1}$ as defined in \eqref{eq:diffeoimpl} is trivially-well defined, but it is injective and surjective with inverse $u$ only thanks to the existence and uniqueness of the decomposition \eqref{eq:ProjDecomp}.

To show that $u$ is a smooth diffeomorphism, we prefer to work with its inverse $u^{-1}$, which we proceed to write down more explicitly. For this purpose, we pick a real $L^2$-orthonormal basis $\{f_k\}_{k=1}^N$ of $\ran \Pi$, such that $f_1$, $f_2$ and $f_3$ are an orthonormal basis of $\spn \{\Pi W_T \partial_j \varphi_L\}_{j=1}^3$ and $f_4=\frac{\Pi \varphi_L}{\|\Pi \varphi_L\|_2}$. Note that $\spn \{\Pi W_T \partial_j \varphi_L\}_{j=1}^3$ is three dimensional, as remarked after \eqref{eq:projdef}, at least for $N$ and $T$ large enough, and that $f_4$ is indeed orthogonal to $f_1$, $f_2$ and $f_3$ since in $k$-space $W_T$ and $\Pi$ are even multiplication operators while the partial derivatives are odd multiplication operators. We denote the projection onto $\spn \{\Pi W_T \partial_j \varphi_L\}_{j=1}^3$ by
\begin{align}
\GradProjT:=\sum_{k=1}^3 \ket{f_k}\bra{f_k}.
\end{align}
Having fixed a real orthonormal $L^2$-basis, we can identify any real-valued function in $\ran \Pi$ (and hence also any function in $\Tneigh$) with a point $(\lambda_1,\dots,\lambda_N)$ via \eqref{eq:identification}. In these coordinates, the orthogonal transformation that acts on functions in $\ran \Pi$ as the translation by $y$, i.e., $\varphi\mapsto\varphi^y$, reads 
\begin{align}\label{def:Ry}
R(y):=\sum_{k=1}^{N} \ket{f_k^y}\bra{f_k},
\end{align}
and we can write $B^{T,\Lambda}_{\varepsilon}$ in \eqref{eq:BTLedef} as 
\begin{align}
B^{T,\Lambda}_{\varepsilon}:=\left\{\eta=(\eta_4,\dots,\eta_N)\in \spn_{\R}\{f_4,\dots, f_N\}\; \Big| \; \left\|\sum_{k=4}^N \eta_k f_k\right\|_{W_T}\leq r_{T,\varepsilon}\right\}.
\end{align}
In this basis, we can write $u^{-1}$ explicitly as
\begin{equation}
\label{eq:diffeo}
u^{-1}(y,\eta)=(\Pi\varphi_L)^y+\eta^y = R(y)(0,0,0,\|\Pi\varphi_L\|_2+\eta_4,\eta_5,\dots,\eta_N).
\end{equation}

The following Lemma uses this explicit expression for $u^{-1}$ and shows that it is a smooth diffeomorphism (therefore showing that the Gross coordinates map $u$ is as well).

\begin{lem}
	\label{lem:diffeomorphism}
	Let $u^{-1}$ be the map defined in \eqref{eq:diffeo}. There exists $\varepsilon^1_L\leq \varepsilon'_L$ (independent of $T$ and $N$) and $N_L>0$ such that for any $\varepsilon\leq \varepsilon_L^1$, any $T>0$ and any $N>N_L$ the map $u^{-1}$ is a $C^{\infty}$-diffeomorphism from $\TL\times B_{\varepsilon}^{T,\Lambda}$ onto $\Tneigh$. Moreover, for $\varepsilon\leq \varepsilon^1_L$, $|\det Du^{-1}|$ and all its derivatives are uniformly bounded independently of $T$ and $N$.
\end{lem}
\begin{proof}
	We introduce the notation $J(y,\eta)=D u^{-1}(y,\eta)$ and $d(y,\eta):=| \det J(y,\eta)|$. 	Note that $R(y)$ in \eqref{def:Ry}  satisfies $R(-y)=R(y)^{-1}=R(y)^t$ since $\{f_j^y\}_{j=1}^N$ is an orthonormal basis of $\ran \Pi$ for any $y$. Hence, for $j=1,\dots,N$ we have
	\begin{align}
	(u^{-1})_j(y,\eta)&=\bra{f_j}\ket{u^{-1}(y,\eta)}=
	\bra{R(-y)f_j}\ket{\Pi \varphi_L+\sum_{l=4}^N\eta_l f_l}.
	\end{align}
	This yields the smoothness of $u^{-1}$ in $\eta$ and in $y$  (noting that $\{f_j\}_{j=1}^N\subset \ran \Pi$ is a set of smooth functions for any $N$). We proceed to compute $J$. We have, for $4\leq k \leq N$, 
	\begin{align}
	\partial_{\eta_k} (u^{-1})_j (y,\eta)=\bra{R(-y) f_j}\ket{f_k}=\bra{f_j}\ket{R(y) f_k},
	\end{align}
	and
	\begin{align}\nonumber
	\partial_{y_k} (u^{-1})_j (y,\eta) & = \bra{f_j}\ket{\partial_{y_k}R(y)\left(\Pi \varphi_L+\sum_{l=4}^N\eta_l f_l\right)} \\ & = - \bra{f_j}\ket{R(y) \partial_k\left(\Pi \varphi_L+\sum_{l=4}^N\eta_l f_l\right)}
	\end{align}
	for $1\leq k \leq 3$. 
	Therefore
	\begin{align}
	J (y,\eta)&= R(y)\left[\sum_{k=1}^3 \ket{v_k}\bra{f_k}+\sum_{k\geq 4} \ket{f_k}\bra{f_k}\right]\nonumber\\
	&=R(y)\left(\unit-\GradProjT+ \sum_{k=1}^3 \ket{v_k}\bra{f_k}\right)=:R(y)J_0(\eta), \label{JRJ}
	\end{align}
	where $v_k(\eta):= -\partial_k u^{-1}(0,\eta)=-\partial_k \left(\Pi \varphi_L+\sum_{l=4}^N\eta_l f_l\right)$. Since $R(y)$ is orthogonal, we see that $d=|\det J_0|$ (implying, in particular, that $d$ is independent of $y$). 
	
	Observe that 
	\begin{align}
	\label{eq:Jzero}
	J_0=
	\begin{pmatrix}
	A_0 & 0     \\
	A_1 & \unit
	\end{pmatrix},
	\end{align}
	where $A_0$ is the $3\times 3$ matrix given by 
	\begin{align}
	\label{eq:Adef}
	(A_0)_{jk}=\bra{f_j}\ket{v_k}=\bra{f_j}\ket{-\partial_k\left(\Pi \varphi_L +\sum_{l=4}^N \eta_l f_l\right)}, \quad j,k\in\{1,2,3\},
	\end{align} 
	and $A_1$ is the $(N-3)\times 3$ matrix defined by
	\begin{align}
	(A_1)_{jk}=\bra{f_{j+3}}\ket{-\partial_k\left(\Pi \varphi_L+\sum_{l=4}^N\eta_l f_l\right)} \quad j\in \{1,\dots, N-3\},\;\;k\in\{1,2,3\}.
	\end{align}
	Since $J_0$ is the identity in the bottom-right $(N-3)\times (N-3)$ corner and $0$ in the top-right $3\times (N-3)$ corner, $d=|\det A_0|$.
	On $\ran \GradProjT$ the operators $\partial_k$ with $k=1,2,3$ and $W_T^{-1}$ are uniformly bounded in $N$ and $T$. Recall also that $\|\eta\|_{W_T}\leq \varepsilon^1_L$. Hence, for some constant $C_L$ independent of $N$ and $T$, and for any $j,k\in \{1,2,3\}$, we have
	\begin{align}
	\label{eq:Abdd}
	|(A_0)_{jk}|\leq \|\partial_k f_j\|_2 \|\Pi \varphi_L\|_2+\|W_T^{-1}\partial_k f_j\|_{W_T}\|\eta\|_{W_T}\leq C_L.
	\end{align}
	Moreover, for any $j, k \in \{1,2,3\}$ and any $l, l_1, l_2\in\{4,\dots,N\}$, we also have
	\begin{align}
	\label{eq:Ader}
	\partial_{\eta_l} (A_0)_{jk}=\bra{\partial_k f_j}\ket{f_l}, \quad \partial_{\eta_{l_1}}\partial_{\eta_{l_2}} (A_0)_{jk}=0.
	\end{align}
	Clearly, \eqref{eq:Abdd} and \eqref{eq:Ader} together with the fact that $d=|\det A_0|$ show that $d$ and all its derivatives are uniformly bounded in $N$ and $T$. 
	To show that there exists $\varepsilon_L^1$ and $N_L$ such that $d\geq C_L>0$ for all $\varepsilon\leq \varepsilon_L^1$, $T>0$ and $N>N_L$, we show that the image of the $3$-dimensional unit sphere under $A_0$ is uniformly bounded away from $0$, which clearly yields our claim. For this purpose, we observe that the $k$-th column of $A_0$ is given by $\GradProjT \left[-\partial_k \left(\Pi \varphi_L+\sum_{l=4}^N\eta_l f_l\right)\right]$ and therefore, for any unit vector $a=(a_1,a_2,a_3)\in \R^3$, 
	\begin{align}
	A_0 a&=\sum_{k=1}^3 a_k \GradProjT \left[-\partial_k \left(\Pi \varphi_L+\sum_{l=4}^N\eta_l f_l\right)\right]
	=-\GradProjT\partial_a u^{-1}(0,\eta), 
	\end{align}
	where we denote  $\sum_{k=1}^3 a_k \partial_k = \partial_{a}$.  
	To bound the norm of $A_0 a$ from below, it is then sufficient to test $\partial_a u^{-1}(0,\eta)$ against one normalized element of $\ran \GradProjT$, say $\frac{\Pi W_T \partial_a \varphi_L}{\|\Pi W_T \partial_a \varphi_L\|_2}$. We obtain 
	\begin{align}
	\label{eq:partialaest}
	\|A_0a\|_2^2&=\|\GradProjT\partial_a u^{-1}(0,\eta)\|_2^2\geq \left|\bra{\frac{\Pi W_T \partial_a \varphi_L}{\|\Pi W_T \partial_a \varphi_L\|_2}}\ket{\partial_a\left(\Pi \varphi_L+\sum_{l=4}^N\eta_l f_l\right)}\right|^2\nonumber\\
	&=\|\Pi W_T \partial_a \varphi_L\|_2^{-2}\left|\|\Pi W_T^{1/2}\partial_a \varphi_L\|_2^2-\bra{\Pi \partial_a^2 \varphi_L}\ket{\eta}_{W_T}\right|^2\notag\\
	&\geq \|\partial_a \varphi_L\|_2^{-2}\left(\|\Pi W_0^{1/2}\partial_a \varphi_L\|_2^2-\|\Pi \partial_a^2 \varphi_L\|_{W_T}\|\eta\|_{W_T}\right)_+^2\nonumber\\
	&\geq \|\partial_a \varphi_L\|_2^{-2}\left(\|\Pi W_0^{1/2}\partial_a \varphi_L\|_2^2-\varepsilon\|\partial_a^2 \varphi_L\|_2\right)_+^2,
	\end{align}
	where we used that $\|\eta\|_{W_T}\leq \varepsilon$,  $0\leq W_T\leq \unit$ and  $\Pi\leq \unit$, and $(\, \cdot\,)_+$ denotes the positive part. As remarked after \eqref{eq:projdef},  $\partial_a \varphi_L =(-\Delta_L)^{-1/2} \partial_a |\psi_L|^2\neq 0$ and since $\varphi_L\in C^{\infty}$,  $\partial_a \varphi_L$ and $\partial_a^2 \varphi_L$ are uniformly bounded in $a$. We can thus find $N_L>0$ and $\varepsilon_L^1$ such that the r.h.s. of \eqref{eq:partialaest} is bounded from below by some constant $C_L>0$ uniformly for $T>0$, $N>N_L$ and $\varepsilon\leq \varepsilon_L^1$. This shows that $A_0$ (and hence $J$) is invertible at every point and  that $d\geq C_L>0$ uniformly in $T>0$, $N>N_L$ and $\varepsilon\leq\varepsilon^1_L$, as claimed. This concludes the proof.
\end{proof}

Since $u$ is a diffeomorphism, we can introduce a unitary operator that lifts $u^{-1}$ to $L^2$, defined by
\begin{align}
\label{eq:unitary}
&U:L^2(\TL \times B^{T,\Lambda}_{\varepsilon}) \longrightarrow L^2(\Tneigh) \nonumber\\
&U(\psi):=|\det \left(D u\right)|^{1/2} \psi \circ u.
\end{align}
Recall that $j_1$ is supported in $\Tneigh$, hence we can apply $U$ to 
 $j_1 \mathbb{K} j_1$, obtaining an operator that acts on functions on $\TL \times \R^{N-3}$ that are supported in  $\TL \times B^{T,\Lambda}_{\varepsilon}$.  In particular, 
\begin{align}
j_1 \mathbb{K} j_1\geq j_1^2 \infspec_{H^1_0\left(\TL\times B^{T,\Lambda}_{\varepsilon}\right)}[U^*\mathbb{K}U],
\end{align}  
where the subscript indicates that the operator has to be understood as the corresponding quadratic form with form domain $H^1_0(\TL\times B^{T,\Lambda}_{\varepsilon})$ (i.e., with Dirichlet boundary conditions on the boundary of $B^{T,\Lambda}_{\varepsilon}$). 
We are hence left with the task of giving a lower bound on $\infspec_{H^1_0\left(\TL\times B^{T,\Lambda}_{\varepsilon}\right)}[U^*\mathbb{K}U]$, which will be done in the remainder of this subsection.

Recalling the definition of $\mathbb{K}$ given in \eqref{eq:defK}, we proceed to find a convenient lower bound for $U^* \FL U$. Any $(\Pi\varphi_L)^{y_{\varphi}}+(w_{\varphi})^{y_{\varphi}}=\varphi \in \Tneigh$ satisfies \eqref{eq:FHessianregime} with $\varphi_L^{y_{\varphi}}$ in place of $\varphi_L$,  and we can therefore expand $\FL(\varphi)$ using Proposition \ref{prop:FHessian}, obtaining 
\begin{align}
\FL(\varphi)-\eL\geq&\expval{\unit-K_L^{y_{\varphi}}-\varepsilon C_LJ_L^{y_{\varphi}}}{(w_{\varphi})^{y_{\varphi}}-((\unit-\Pi)\varphi_L)^{y_{\varphi}}}\notag\\
=& \expval{(\unit-\Pi)(\unit-K_L-\varepsilon C_LJ_L)(\unit-\Pi)}{\varphi_L}\notag\\
&-2\bra{(\unit-\Pi)\varphi_L} \unit-K_L-\varepsilon C_LJ_L\ket{w_{\varphi}}+\expval{\unit-K_L-\varepsilon C_LJ_L}{w_{\varphi}}.
\end{align} 
Since $K_L$ and $J_L$ are trace class operators, 
\begin{align}
(\unit-\Pi)(\unit-K_L-\varepsilon C_LJ_L)(\unit-\Pi)>0
\end{align} 
holds for $\Lambda$ sufficiently large and $\varepsilon$ sufficiently small. Moreover, since $\varphi_L\in C^{\infty}(\TL)$
\begin{align}\nonumber
& |\bra{(\unit-\Pi)\varphi_L} \unit-K_L-\varepsilon C_LJ_L\ket{w_{\varphi}}| \\ &\leq \|W_T^{-1/2}(\unit-K_L-\varepsilon C_LJ_L)(\unit-\Pi)\varphi_L\|_2\|w_{\varphi}\|_{W_T}=O(\varepsilon \Lambda^{-h})
\end{align} 
for  arbitrary $h>0$ and uniformly in $T$. 
This implies that, for any $\varphi=(\Pi\varphi_L)^{y_{\varphi}}+(w_{\varphi})^{y_{\varphi}} \in \Tneigh$, any $\Lambda$ sufficiently large, any $\varepsilon$ sufficiently small and an arbitrary $h$ 
\begin{align}
\label{eq:FLlocalbound}
\FL(\varphi)=\FL((\Pi\varphi_L)^{y_{\varphi}}+(w_{\varphi})^{y_{\varphi}})\geq \eL -O(\varepsilon \Lambda^{-h})+\expval{\unit-K_L-\varepsilon C_LJ_L}{w_{\varphi}}.
\end{align}
Therefore, if we define the $[(N-3)\times (N-3)]$-matrix $M$ with coefficients
\begin{align}
\label{eq:Mdef}
M_{k,j}:= \bra{f_{k+3}} \unit-K_L-\varepsilon C_LJ_L \ket{f_{j+3}},
\end{align}
then, by \eqref{eq:FLlocalbound}, the multiplication operator $U^* \FL U$ satisfies 
\begin{align}\label{p1}
(U^* \FL U) (y,\eta)\geq \eL+\expval{M}{\eta}-O(\varepsilon \Lambda^{-h}).
\end{align}
It is easy to see that $M$ is a positive matrix, at least for $\varepsilon$ sufficiently small and $T$ and $\Lambda$ sufficiently large. Indeed, the positivity of $M$ is equivalent to the positivity of $(\unit-K_L-\varepsilon C_LJ_L)$ on $\ran (\Pi-\GradProjT)$ and, by Proposition \ref{prop:Hessianstrictpos}, $(\unit-K_L-\varepsilon C_LJ_L)$ is positive on any vector space with trivial intersection with $\ran \GradProj$. Clearly, since $\GradProjT\to \GradProj$ as $T\to \infty$, the bound 
\begin{align}
\label{eq:Mpos}
M\geq c_L>0
\end{align} 
holds, uniformly in $T$, $\Lambda$ and for $\varepsilon$ sufficiently small.

We now proceed to bound $- U^* \Delta_{\lambda} U$ from below. 

\begin{lem}
Let $U$ be the unitary transformation defined in \eqref{eq:unitary}. There exists $C_L>0$, independent of $N$, $T$ and $\varepsilon$, such that, for $\varepsilon \leq \varepsilon^1_L$, $T>0$ and $N>N_L$ 
\begin{align}
\label{eq:LemmaDiffLaplClaim}
U^*\left(-\Delta_{\lambda}\right)U\geq -\Delta_{\eta}-C_L.
\end{align}
\end{lem} 
\begin{proof}
Since \eqref{JRJ} shows that  $J(y,\eta)=R(y)J_0(\eta)$ with $R(y)$ orthogonal, we have 
\begin{align}
\label{eq:diffeomLapl1}
U^* \left(-\Delta_{\lambda} \right)U&=-d^{-1/2} \grad \cdot d^{1/2}\left[J^{-1} (J^{-1})^t\right] d^{1/2} \grad d^{-1/2}\nonumber\\
&=-d^{-1/2} \grad  \cdot d^{1/2}\left[J_0^{-1} (J_0^{-1})^t\right] d^{1/2} \grad d^{-1/2},
\end{align}
with $d(y,\eta)=| \det J(y,\eta)|$ and  $\grad$ denoting the gradient with respect to  $(y,\eta)\in \R^N$. 
Recalling the  expression \eqref{eq:Jzero} for $J_0$, 
we find 
\begin{align}
J_0^{-1}=
\begin{pmatrix}
A_0^{-1} & 0     \\
-A_1 A_0^{-1} & \unit
\end{pmatrix}
=
\begin{pmatrix}
0 & 0     \\
0 & \unit
\end{pmatrix}
+
\begin{pmatrix}
A_0^{-1} & 0     \\
-A_1 A_0^{-1} & 0
\end{pmatrix}
=:(1-\GradProjT)+D.
\end{align}
Since $D(\unit-\GradProjT)=(\unit-\GradProjT)D^t=0$, we have
\begin{align}
\label{eq:JzeroJzerot}
J_0^{-1} (J_0^{-1})^t=(\unit-\GradProjT)+D D^t\geq \unit-\GradProjT.
\end{align}
With \eqref{eq:diffeomLapl1} and \eqref{eq:JzeroJzerot}, we thus obtain 
\begin{align}
U^* \left(-\Delta_{\lambda}\right) U&\geq -d^{-1/2} \grad \cdot d^{1/2}\left(\unit -\GradProjT\right) d^{1/2} \grad d^{-1/2}\nonumber\\
&=-\Delta_{\eta}-(2d)^{-2} |\nabla d|^2+(2d)^{-1}\Delta d.
\end{align}
 Lemma \ref{lem:diffeomorphism} guarantees that $d$ and all its derivatives are bounded, and $d$ is bounded away from $0$ uniformly in $N>N_L$, $T>0$ and $\varepsilon\leq \varepsilon_L^1$, leading to  \eqref{eq:LemmaDiffLaplClaim}.
\end{proof}

In combination, \eqref{p1}, \eqref{eq:LemmaDiffLaplClaim} and the positivity of $M$ imply that 
\begin{align}
j_1 \mathbb{K} j_1&\geq j_1^2 \infspec_{H^1_0\left(\TL\times B^{T,\Lambda}_{\varepsilon}\right)}(U^* \mathbb{K} U)\\
&\geq j_1^2\left(\eL-\frac N {2 \alpha^2}-O(\varepsilon \Lambda^{-h})-O(\alpha^{-4})+\infspec_{L^2(\R^N)} \left[-\frac 1 {4\alpha^4} \Delta_{\eta}+\expval{M}{\eta}\right] \right)\nonumber\\
&=j_1^2\left(\eL-\frac 1 {2 \alpha^2}(N-\Tr (M^{1/2}))-O(\varepsilon \Lambda^{-h})- O(\alpha^{-4})\right).
\end{align}
Note that since we are taking $\Lambda\gg \alpha^{4/5}$, $\varepsilon\ll 1$ and $h>0$ was arbitrary, picking $h=5$  allows to absorb the error term $O(\varepsilon \Lambda^{-h})$ in the error term $O(\alpha^{-4})$. Recalling the definition of $M$ given in \eqref{eq:Mdef}, we have
\begin{align}
\Tr(M^{1/2})=\Tr\left[\sqrt{(\Pi-\GradProjT)(\unit-K_L-\varepsilon C_L J_L)(\Pi-\GradProjT)}\right].
\end{align}
With $\{t_j\}_{j=1}^{N-3}$ an orthonormal basis of $\ran(\Pi-\GradProjT)$ of eigenfunctions of $(\Pi-\GradProjT)(\unit-K_L-\varepsilon C_LJ_L)(\Pi-\GradProjT)$, we can write   
\begin{align}
\Tr(M^{1/2})&=\sum_{j=1}^{N-3} \expval{\unit-K_L-\varepsilon C_LJ_L}{t_j}^{1/2}\nonumber\\
&=\sum_{j=1}^{N-3} \left[ \expval{\unit-K_L}{t_j}^{1/2}-\frac{\varepsilon C_L}{2 \xi_j^{1/2}} \expval{J_L}{t_j}\right]
\end{align}
for some $\{\xi_j\}_{j=1}^{N-3}$ satisfying
\begin{align}
c_L\leq \expval{\unit-K_L-\varepsilon C_L J_L}{t_j}\leq \xi_j \leq \expval{\unit-K_L}{t_j}\leq 1
\end{align}
for $T$ and $\Lambda$ large enough and $\varepsilon$ small enough, 
where we used  \eqref{eq:Mpos} for the lower bound. Using the concavity of the square root and the trace class property of $J_L$, we conclude that
\begin{align}
\Tr(M^{1/2})\geq \sum_{j=1}^{N-3} \expval{\sqrt{\unit-K_L}}{t_j}-\varepsilon C_L \Tr(J_L)=\Tr\left[(\Pi-\GradProjT)\sqrt{\unit-K_L}\right]-\varepsilon C_L.
\end{align}
Since $\varphi_L\in C^{\infty}$ and recalling \eqref{eq:TlambdaRegime}, for an arbitrary $h>0$ we can bound
\begin{align}
\|\GradProj-\GradProjT\|\lesssim_L \min\{\Lambda,T\}^{-h}=T^{-h},
\end{align} 
which also implies the same estimate for the trace-norm of the difference of $\GradProj$ and $\GradProjT$,  both operators being of rank $3$. Recalling that $\GradProj$ projects onto $\ker (\unit-K_L)$, we finally obtain 
\begin{align}
\Tr(M^{1/2})\geq \Tr [\Pi\sqrt{\unit-K_L}]-O(\varepsilon)-O(T^{-h}).
\end{align} 
The error term $O(T^{-h})$ forces  $T\to \infty$ as $\alpha \to \infty$, but  allows $T$ to grow with an arbitrarily small power of $\alpha$. By picking $h$ to be sufficiently large we can absorb it in the error term $O(\varepsilon)$. 

We obtain the final lower bound
\begin{align}
j_1 \mathbb{K} j_1 &\geq j_1^2 \left[\eL-\frac 1 {2\alpha^2}\Tr[\Pi(\unit-(\unit- K_L )^{1/2})]-O(\varepsilon \alpha^{-2})-O(\alpha^{-4})\right]\nonumber\\
&\geq j_1^2 \left[ \eL-\frac 1 {2\alpha^2}\Tr[(\unit-(\unit-K_L )^{1/2})]-O(\varepsilon \alpha^{-2})-O(\alpha^{-4})\right]. \label{r3}
\end{align}

\subsubsection{Bounds on $j_2 \mathbb{K} j_2$} \label{Sec:IMSouter} We recall Corollary \ref{cor:uniquenessANDcoercivity}, which implies that, for any $\varphi \in L^2_{\mathbb{R}}(\TL)$,
\begin{align}
\FL(\varphi)\geq \eL+\inf_{y\in \TL}\expval{B}{\varphi-\varphi_L^y}, 
\end{align}
where $B$ acts in $k$-space as the multiplication by 
\begin{align}
B(k)=\begin{cases}
1 &\text{for} \,\, k=0,\\
1-(1+\kappa'|k|)^{-1} &\text{for}\,\, k\neq 0.
\end{cases}
\end{align}
Note that $B-\eta W_T>0$ for $\eta>0$ small enough (independently of $T$). Moreover, for any $\varphi$ in the support of $j_2$ and any $y\in \TL$, 
\begin{align}
\expval{W_T}{\varphi-\varphi_L^y}\geq \varepsilon^2/4.
\end{align}
Therefore, on the support of $j_2$, we have
\begin{align}
\FL(\varphi) \geq \eL+\inf_{y\in \TL}\expval{B-\eta W_T}{\varphi-\varphi_L^y}+\eta \varepsilon^2/4.
\end{align}
By the Cauchy--Schwarz inequality, using that all the operators involved commute, we have
\begin{align}
\expval{B-\eta W_T}{\varphi-\varphi_L^y}&\geq \expval{(\unit-W_{\gamma}^{1/2})(B-\eta W_T)}{\varphi}\nonumber\\
&\quad +\expval{(\unit-W_{\gamma}^{-1/2})(B-\eta W_T)}{\varphi_L}
\end{align}
for any $\gamma>0$. 
Note that the right hand side is  independent of $y$. 
Since $\varphi_L\in C^{\infty}(\TL)$, the Fourier coefficients of $\varphi_L$ satisfy
\begin{align}
(1+|k|^2)^{5/2}|(\varphi_L)_k|^2\leq  C_{L,t} \gamma^{-t} \quad \text{for} \quad |k|\geq \gamma
\end{align} 
for any $t>0$. 
Using the positivity of $B-\eta W_T$ we can bound 
\begin{align}
\expval{(\unit-W_{\gamma}^{-1/2})(\mathbb{B}-\eta W_T)}{\varphi_L}&\geq-\sum_{k\in \frac {2\pi}{L} \mathbb{Z}^3\atop |k|>\gamma} (B(k)-\eta W_T(k))(1+|k|^2)^{1/2} |(\varphi_L)_k|^2\nonumber\\
&= -\sum_{k\in \frac {2\pi}{L} \mathbb{Z}^3\atop |k|>\gamma} \frac{(B(k)-\eta W_T(k))}{(1+|k|^2)^{2}}(1+|k|^2)^{5/2} |(\varphi_L)_k|^2\nonumber\\
&\geq  -C_{L,t} \gamma^{-t} \sum_{k\in \frac {2\pi}{L} \mathbb{Z}^3\atop |k|>\gamma} \frac{1}{(1+|k|^2)^{2}}\gtrsim_L \gamma^{-t-1}  . 
\end{align}
Therefore we conclude, using the positivity of $\unit-W_{\gamma_{\beta}}^{1/2}$ and of $B-\eta W_T$, that
\begin{align}
\label{eq:j2Kbound1}
&j_2\mathbb{K} j_2\notag\\
&\geq j_2^2 \infspec\left[\eL- \frac N {2 \alpha^2}+\frac{\eta\varepsilon^2}4-O(\gamma^{-t-1})-\frac 1 {4\alpha^4} \Delta_{\lambda}+\expval{(\unit-W_{\gamma}^{1/2})(B-\eta W_T)}{\varphi}\right] \notag\\
&=j_2^2\left(\eL+\frac{\eta\varepsilon^2}4-O(\gamma^{-t-1})-\frac 1 {2\alpha^2} \Tr\left[\Pi\left(\unit-\sqrt{(\unit-W_{\gamma}^{1/2})(B-\eta W_T)}\right)\right]\right).
\end{align}

We need to estimate the behavior in $N=\rank \Pi$, $T$ and $\gamma$ of the trace appearing in the last equation, which equals
\begin{align}
\nonumber
&\Tr\left[\Pi\left(\unit-\sqrt{(\unit-W_{\gamma}^{1/2})(B-\eta W_T)}\right)\right] \\ &=\sum_{k\in \frac{2\pi} L \mathbb{Z}^3 \atop |k|\leq \Lambda} \left(1-\sqrt{(1-W_{\gamma}(k)^{1/2})(B(k)-\eta W_T(k))}\right). 
\label{eq:estimatesTraceOutside}
\end{align}
The contribution to the sum from $|k|\leq \max\{\gamma,T\}$ can be bounded by $C(L\max\{\gamma,T\})^3$. For $|k|> \max\{\gamma,T\}$, $W_\gamma(k) = W_T(k) = (1+|k|^2)^{-1}$, and the coefficient under the square root in the last line of \eqref{eq:estimatesTraceOutside} behaves asymptotically for large momenta as $1 - |k|^{-1}$. 
Hence, recalling \eqref{eq:Nlambdaasymp}, we  conclude that 
\begin{align}
\label{eq:tracebound}
\Tr\left[\Pi\left(\unit-\sqrt{(\unit-W_{\gamma}^{1/2})(B-\eta W_T)}\right)\right]\leq O\left(\max\{\gamma,T\}^3\right)+O(\Lambda^2).
\end{align}
Because of \eqref{eq:TlambdaRegime}, the first term on the right hand side is negligible compared to the second if we choose $\gamma$ to equal $\alpha$ to some small enough power. Because $t$ was arbitrary, we thus arrive at 
\begin{align}
j_2\mathbb{K} j_2\geq  
j_2^2\left(\eL+\frac{\eta\varepsilon^2}4-O(\alpha^{-2}\Lambda^2)\right).
\end{align}
Therefore, if
\begin{align}
\label{eq:restr2bis}
 \varepsilon\geq  C_L \alpha^{-1}\Lambda
\end{align}
for a sufficiently large constant $C_L$, we conclude that for sufficiently large $\alpha$ and $\Lambda$
\begin{align}\label{r4}
j_2\mathbb{K} j_2\geq j_2^2 \eL.
\end{align}

\subsubsection{Proof of Theorem \ref{theo:GSEexpansion}, lower bound} \label{Sec:FinalLowerBound}  By combining the results \eqref{r3} and \eqref{r4} of the previous two subsections with \eqref{r1} and \eqref{r2}, we obtain
\begin{align}
\mathbb{K}& \geq j_1 \mathbb{K} j_1+j_2 \mathbb{K} j_2 +O(\alpha^{-4}\varepsilon^{-2})\nonumber\\
&\geq j_1^2\left[ \eL-\frac 1 {2\alpha^2}\Tr[(\unit-(\unit-K_L )^{1/2})]+O(\varepsilon \alpha^{-2})+O(\alpha^{-4})\right]+j_2^2 \eL+O(\alpha^{-4}\varepsilon^{-2})\nonumber\\
&\geq \eL-\frac 1 {2\alpha^2}\Tr[(\unit-(\unit-K_L )^{1/2})]+O(\varepsilon \alpha^{-2})+O(\alpha^{-4})+O(\alpha^{-4}\varepsilon^{-2})
\end{align}
under the constraint \eqref{eq:restr2bis}. With Proposition \ref{prop:cutoffH}  
we can thus conclude that
\begin{align}
\infspec \HL&\geq \infspec \HL^{\Lambda}+O(\Lambda^{-5/2})+O(\alpha^{-1} \Lambda^{-3/2})+O(\alpha^{-2}\Lambda^{-1})\nonumber\\
&\geq \infspec \mathbb{K}+O(\Lambda^{-5/2})+O(\alpha^{-1} \Lambda^{-3/2})+O(\alpha^{-2}\Lambda^{-1})\nonumber\\
&\geq \eL-\frac 1 {2\alpha^2}\Tr[(\unit-(\unit-K )^{1/2})]+O(\varepsilon \alpha^{-2})+O(\alpha^{-4})+O(\alpha^{-4}\varepsilon^{-2})\nonumber\\
&\quad +O(\Lambda^{-5/2})+O(\alpha^{-1} \Lambda^{-3/2})+O(\alpha^{-2}\Lambda^{-1}).
\end{align}
To minimize the error terms under the constraint \eqref{eq:restr2bis}, we 
pick $\varepsilon \sim \alpha^{-1/7}$ and $\Lambda \sim \alpha^{6/7}$, which yields the claimed estimate
\begin{align}
\infspec \HL\geq \eL-\frac 1 {2\alpha^2}\Tr[(\unit-(\unit-K_L )^{1/2})]+O(\alpha^{-15/7}).
\end{align} 
This concludes the proof of the lower bound, and hence the proof of Theorem~\ref{theo:GSEexpansion}.

\section*{Acknowledgments}
Funding from the European Union's Horizon 2020 research and innovation programme under the ERC grant agreement No 694227 is gratefully acknowledged. We would also like to thank  Rupert Frank for many helpful discussions, especially related to the Gross coordinate transformation defined in Def.~\ref{def:GroosCoord}.
\bibliographystyle{siam}

\end{document}